\documentclass[reqno,oneside,10pt]{amsart}

\usepackage{tikz}
\usepackage{mathptmx}
\usepackage{xspace}
\usepackage{amsmath,amsfonts,amsthm,eucal}
\usepackage{setspace,geometry,multirow}
\usepackage[pdfpagemode=UseNone, pdfstartview={XYZ null null null}]{hyperref}
\hypersetup{
colorlinks=true,
citecolor=red,
linkcolor=blue,
urlcolor=blue
}

\newtheorem{Thm}{Theorem}[section]
\newtheorem{Prop}[Thm]{Proposition}
\newtheorem{Lem}[Thm]{Lemma}
\newtheorem{Coro}[Thm]{Corollary}
\newtheorem{Def}{Definition}

\newcommand{\ringK}{\mathbb C}

\newcommand{\tl}[1]{\mathsf{TL}_{#1}}

\newcommand{\atl}[1]{\mathsf{TL}^{\mathsf a}_{#1}}
\newcommand{\paires}{\Lambda^{\mathsf a}}

\newcommand{\ctl}{\mathsf{\widetilde{TL}}}

\newcommand{\catl}{\mathsf{\widetilde{TL^{a}}}}
\newcommand{\cat}{\mathcal{C}}

\newcommand{\blank}{{-}}
\newcommand{\xf}{\times_f}
\newcommand{\xfone}{\times^1_f}
\newcommand{\xftwo}{\times^2_f}
\newcommand{\xfthree}{\times^3_f}
\newcommand{\com}[2]{\eta_{#1,#2}}
\newcommand{\T}[1]{t_{#1}}

\newcommand{\Mod}{{\rm\text\,{}mod\,}}

\newcommand{\rk}{{\rm{\text{rk\,}}}}
\newcommand{\mysum}[2]{\underset{#1}{\overset{#2}{\sum{}{\hskip-2pt}'}}}
\newcommand{\myoplus}[2]{\underset{#1}{\overset{#2}{\oplus{}{\hskip+0pt}'}}}
\newcommand{\mybigoplus}[2]{\underset{#1}{\overset{#2}{\bigoplus{}{\hskip-1pt}'}}}

\newcommand{\TheS}[1]{\mathsf S_{#1}}

\newcommand{\TheW}[1]{\mathsf W_{#1}}
\newcommand{\TheH}[1]{\mathsf H_{#1}}
\newcommand{\TheR}[1]{\mathsf R_{#1}}
\newcommand{\TheP}[1]{\mathsf P_{#1}}
\newcommand{\ThePaff}[1]{\mathsf P^{\mathsf a}_{#1}}
\newcommand{\TheM}[1]{\mathsf M_{#1}}
\newcommand{\TheL}[1]{\mathsf L_{#1}}
\newcommand{\TheI}[1]{\mathsf I_{#1}}

\newcommand{\Resar}[1]{\downarrow^{a}_r\hskip -0.018em #1{}}

\newcommand{\Indar}[1]{\uparrow^{a}_r\hskip -0.018em #1{}}
\newcommand{\Resphi}[1]{ \Downarrow^{r}_a\hskip -0.018em #1{}}
\newcommand{\Indphi}[1]{ \Uparrow^{r}_a\hskip -0.018em #1{}}

\DeclareMathOperator{\id}{id}

\DeclareMathOperator{\Hom}{Hom}
\DeclareMathOperator{\End}{End}

\definecolor{rougePompier}{rgb}{0.93,0.11,0.14}
\definecolor{vertForet}{rgb}{0.04,0.75,0.07}

\title[Computing affine fusion]{On the computation of fusion\\ over the affine Temperley-Lieb algebra}

\author[J Bellet\^ete]{Jonathan Bellet\^ete}

\address[Jonathan Bellet\^ete]{
Institut de Physique Th\'eorique, CEA Saclay\\
91191 Gif Sur Yvette, France }

\email{\tt jonathan.belletete@ipht.fr}

\author[Y Saint-Aubin]{Yvan Saint-Aubin}

\address[Yvan Saint-Aubin]{
D\'{e}partement de math\'{e}matiques et de statistique\\
Universit\'{e} de Montr\'{e}al\\
Montr\'eal, QC, Canada, H3C 3J7.}

\email{yvan.saint-aubin@umontreal.ca}

\date{\today}

\begin{document}

\begin{abstract}Fusion product originates in the algebraisation of the operator product expansion in conformal field theory. Read and Saleur (2007) introduced an analogue of fusion for modules over associative algebras, for example those appearing in the description of 2d lattice models. The article extends their definition for modules over the affine Temperley-Lieb algebra $\atl n$.

Since the regular Temperley-Lieb algebra $\tl n$ is a subalgebra of the affine $\atl n$, there is a natural pair of adjoint induction-restriction functors $(\Indar{}, \Resar{})$. The existence of an algebra morphism $\phi:\atl n\to\tl n$ provides a second pair of adjoint functors $(\Indphi{},\Resphi{})$. Two fusion products between $\atl{}$-modules are proposed and studied. They are expressed in terms of these four functors. The action of these functors is computed on the standard, cell and irreducible $\atl n$-modules. As a byproduct, the Peirce decomposition of $\atl n(q+q^{-1})$, when $q$ is not a root of unity, is given as direct sum of the induction $\Indar{\TheS{n,k}}$ of standard $\tl n$-modules to $\atl n$-modules. Examples of fusion products of various pairs of affine modules are given.

\medskip

\end{abstract}

\maketitle  

\tableofcontents

\onehalfspacing

%
\begin{section}{Introduction}\label{sec:intro}
%

In conformal field theory, fusion describes the field content of the operator product of two primary fields. Under the field-state correspondence, this can be stated in terms of representations of the conformal algebra (Virasoro algebra, an affine Lie algebra, etc.). An analogue of such ``fusion product'' has been introduced by Read and Saleur \cite{ReadSaleur, ReadSaleur2} for modules over associative algebras associated with two-dimensional statistical models. The present article extends their definition for modules over the affine Termperley-Lieb algebra $\atl{}$ and develops tools to compute fusion of some basic $\atl{}$-modules.

Read and Saleur's slit-strip argument may seem naive, but it is remarkably deep. In conformal field theory, the operator product of local boundary fields is sometimes seen as bringing points on the boundary together. Imagine a two-dimensional lattice model on a strip of width $n$ with ``evolution'' (ruled by a Hamiltonian or a transfer matrix) moving upward. As an example, suppose that the evolution is given by a loop model with the Hamiltonian being the sum $\sum_i u_i$ of the generators of the Temperley-Lieb algebra $\tl n$ (in some representation). A vertical slit cuts the strip from the bottom edge to some point at finite-distance from this edge. There are $n_l$ and $n_r$ sites on each side of the slit ($n=n_l+n_r$). Until the slit's end, degrees of freedom on its left evolve independently from those on its right. The state along the slit-strip is thus broken into states that transform under some partial Hamiltonians, one in $\tl{n_l}$ for the left part and one in $\tl{n_r}$ for the right part. Let $\mathsf M_l$ and $\mathsf M_r$ be the modules within which the left and right states take their values. In which $\tl n$-module lie the state after the slit's end? Read and Saleur proposed that it should be in a ``fused'' module denoted by $\mathsf M_l\xf \mathsf M_r$ obtained by induction from the subalgebra $\tl{n_l}\times \tl{n_r}$ to $\tl n$, that is:  $\mathsf M_l\xf \mathsf M_r=\tl n\otimes_{\tl{n_l}\times\tl{n_r}}(\mathsf M_l\otimes_{\mathbb C} \mathsf M_r)$. Using this formal definition, Gainutdinov and Vasseur \cite{GainutdinovVasseur} succeeded in computing fusion product of several basic modules of the algebra $\tl n$. However some key products escaped their computing technique. Using a different method, Bellet\^ete \cite{BelFus15} was able to compute the product of irreducible $\tl{}$-modules and show that their fusion reproduces that of primary fields on the first row of the Kac table of minimal models. This constitutes a compelling support for Read and Saleur's fusion. Other key features of their fusion product are its commutativity and associativity, and the fact that, in some precise sense, it has a well-defined limit when $n_l$ and $n_r$ both go to infinity. The latter property is welcome as the fusion of a pair of modules, one over $\tl{n_l}$ and the other over $\tl{n_r}$, gives a module over, yet, a third algebra, namely $\tl{n=n_l+n_r}$.

On closer inspection, the proposed definition of fusion can be readily extended to the modules of any family of algebras $\{\mathsf A_i, i\in\mathbb N\}$ that is filtrated in the sense that $\mathsf A_i\times \mathsf A_j\subseteq \mathsf A_{i+j}$ for any $i$ and $j$. These include several families of algebras, defined diagrammatically and physically relevant. However one family does not share this property: it is the family of affine (or periodic) Temperley-Lieb algebras $\atl n$, $n\in\mathbb N$. This family is associated with many physical models with periodic boundary conditions whose continuum limit are CFT's on a cylinder or a torus. Clearly a fusion product for the affine family would be a desirable addition to the algebraic description of these theories. A formal definition of $\atl n$ will be given below. An intuitive one is that $\atl n$ is a ``Temperley-Lieb algebra defined on a circle'' of $n$ sites, where the generators $e_1$ and $e_n$ satisfy the usual neighbor property: $e_1e_ne_1=e_1$ and $e_ne_1e_n=e_n$. These identities rule out any obvious inclusion of $\atl{n_l}\times\atl{n_r}$ into $\atl{n_l+n_r}$. Read and Saleur's proposal for a fusion product will need an important extension to be applicable to this family.

These identities rule out any diagrammatic inclusion of $\atl{n_l}\times\atl{n_r}$ into $\atl{n_l+n_r}$. Thus Read and Saleur's proposal for a fusion product does not extend trivially to the affine Temperley-Lieb family. 

A way around this difficulty was proposed by Gainutdinov and Saleur \cite{GaiSalAffine}. Even though a general diagram of the algebra $\atl m\times \atl n$ does not correspond to a diagram in $\atl{m+n}$ in any obvious way, these authors were able to construct a injective morphism $\atl m\times \atl n\to\atl{m+n}$. Thanks to this remarkable map, the problem of defining an affine fusion may follow the path used for the regular Temperley-Lieb family.

The present article studies two alternative fusion products over the affine Temperley-Lieb algebras. They are constructed out of induction and restriction functors. Since the regular Temperley-Lieb algebra $\tl n$ is a subalgebra of the affine $\atl n$, there is a pair of functors associated with this inclusion, namely the induction $\Indar$ and the restriction $\Resar$. (The indices ``{\itshape a}'' and ``{\itshape r}'' refer to the affine and regular algebras $\atl{}$ and $\tl{}$, respectively.) There exists yet a second pair of functors, denoted by $\Indphi$ and $\Resphi$. They come from the existence of a remarkable morphism $\phi:\atl n\to\tl n$, based on tools available since a long time \cite{Chow}, but that may not be well-known in the present context. Section \ref{sec:TLalgebras} defines the regular and affine families of Temperley-Lieb algebras, constructs the morphism $\phi$ and uses the four functors to propose several fusion products of $\atl{}$-modules. All these fusions are based on the fusion $\xf=\times_f^{\tl{}}$ between modules over $\tl{}$-modules. The actual computation of a fusion product requires an extensive knowledge of the representation theory of both families and how their modules behave under the four functors. Section \ref{sec:modulesMorphism} recalls the definition of the basic modules: the standard, projective and irreducible modules over $\tl n$, and the standard, cell and irreducible ones over $\atl n$. It also computes the $\Hom$-groups between some of these affine modules. Finally Section \ref{sec:restrictionInduction} constructs the induction and restriction of these basic modules and ends with several examples of computation of the fusion of $\atl{}$-modules. Some concluding remarks follow. An appendix is devoted to the description of the weakest partial order, a tool to determine the composition factors of affine cell modules, and the computation of the dimensions of the irreducible modules of $\atl n$.

\end{section}

%
\begin{section}{The Temperley-Lieb algebras}\label{sec:TLalgebras}
%

This section defines the fusion functors to be considered in this paper. It uses the langage of categories introduced in the study of the affine Temperley-Lieb family by Graham and Lehrer \cite{GLaffine}. The first subsection defines the two Temperley-Lieb categories, the affine $\catl$ and the regular $\ctl$, and recalls the braiding on the regular Temperley-Lieb category. The Temperley-Lieb algebras $\atl n$ and $\tl n$ will then be identified to the sets of endomorphisms $\Hom_\catl(n,n)$ for the affine algebra and the set $\Hom_{\ctl}(n,n)$ for the regular one. These definitions make it clear that $\tl n\subset\atl n$ as algebras. Subsection \ref{sec:tlalgebras} shows that there is a morphism $\atl n\to \tl n$. The fact that $\atl n$ contains $\tl n$ provides two functors between modules over the two algebras, namely the induction $\Indar{( \blank )}: \Mod \tl{n} \to \Mod \atl{n}$ and the restriction $\Resar{(\blank)}: \Mod \atl{n} \to \Mod \tl{n}$. Similarly, the fact that $\tl n$ contains an image of $\atl n$ leads to two functors, an induction one $\Indphi{(\blank)}: \Mod\atl{n} \to \Mod\tl{n}$ and a restriction $\Resphi(\blank): \Mod \tl{n} \to \Mod{\atl{n}}$. Subsection \ref{sec:fusionFunctors} uses these functors to propose two possible fusions between $\atl{}$-modules.

%
%
\begin{subsection}{The Temperley-Lieb categories}\label{sec:tlcategories}
We start by introducing the two Temperley-Lieb categories: the regular and the affine. For both, the class of objects is simply the set $\mathbb N_{\geq 0}$ of non-negative integers and, for each object $n$, $\End(n)$ will be a realization of some Temperley-Lieb algebra, and some particular quotients of homomorphism groups will serve as basic modules over them. The presentation follows that found in \cite{GLaffine}.

Each of these categories will be constructed as a quotient of a more general one, denoted by $\cat$. The objects of $\cat$ are the non-negative integers. For a pair of integers $n,m$, $\Hom_{\cat}(n,m)$ is the free $\ringK$-module spanned by the periodic $(m,n)$-diagrams, which are defined as follows. We start with a rectangle, with $n$ vertices on the right side, and $m$ on the left one; the top and bottom of the rectangle are identified, so that it is topologically a cylinder. The vertices on each side are labelled with numbers\footnote{In what follows, these labels will always be taken modulo $m$ and $n$, respectively.} from $1$ to $m$ and $1$ to $n$, respectively, and are linked pairwise by lines without intersection. The diagram may also contain any number of closed lines, called loops. Two such diagrams define the same morphism if and only if their lines are isotopic. For example here are one element of $\Hom_{\cat}(4,4)$ and two of $\Hom_{\cat}(2,4)$:
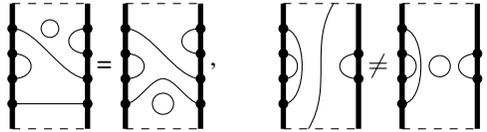
\begin{equation*}
	\begin{tikzpicture}[scale = 1/3]
	\draw[line width = 2pt] (0,-1) -- (0,4);
	\draw[line width = 2pt] (3,-1) -- (3,4);
	\draw[dashed] (0,-1) -- (3,-1);
	\draw[dashed] (0,4) -- (3,4);
	\foreach \s in {1,...,4}
	{	
		\filldraw[black] (0,\s -1) circle (5pt);
		\filldraw[black] (3,\s -1) circle (5pt);
	};
	\draw (0,0)  -- (3,0);
	\draw (0,1) .. controls (1,1) and (1,2) .. (0,2);
	\draw (0,3) .. controls (1,3) and (2,1) .. (3,1);
	\draw (3,2) .. controls (2,2) and (2,3) .. (3,3);
	\draw (3/2,3) circle (10pt);
	\node[anchor = west] at (3, 3/2) {=};
	\end{tikzpicture}
	\begin{tikzpicture}[scale = 1/3]
	\draw[line width = 2pt] (0,-1) -- (0,4);
	\draw[line width = 2pt] (3,-1) -- (3,4);
	\draw[dashed] (0,-1) -- (3,-1);
	\draw[dashed] (0,4) -- (3,4);
	\foreach \s in {1,...,4}
	{	
		\filldraw[black] (0,\s -1) circle (5pt);
		\filldraw[black] (3,\s -1) circle (5pt);
	};
	\draw (0,0) .. controls (1/2,0) and (1,1) ..  (3/2,1);
	\draw (3/2,1) .. controls (2,1) and (5/2,0) .. (3,0);
	\draw (0,1) .. controls (1,1) and (1,2) .. (0,2);
	\draw (0,3) .. controls (1,3) and (2,1) .. (3,1);
	\draw (3,2) .. controls (2,2) and (2,3) .. (3,3);
	\draw (3/2,0) circle (12pt);
	\node[anchor = west] at (3, 3/2) {,};
	\end{tikzpicture}
	\qquad
	\begin{tikzpicture}[scale = 1/3]
	\draw[line width = 2pt] (0,-1) -- (0,4);
	\draw[line width = 2pt] (3,-1) -- (3,4);
	\draw[dashed] (0,-1) -- (3,-1);
	\draw[dashed] (0,4) -- (3,4);
	\foreach \s in {1,...,4}
	{	
		\filldraw[black] (0,\s -1) circle (5pt);
	};
	\filldraw[black] (3,1) circle (5pt);
	\filldraw[black] (3,2) circle (5pt);
	\draw (0,0) .. controls (1,0) and (1,3) .. (0,3);
	\draw (0,1) .. controls (3/4,1) and (3/4,2) .. (0,2);
	\draw (3,1) .. controls (2,1) and (2,2) .. (3,2);
	\draw (2,4) .. controls (1,7/3) and (2,2/3) .. (1,-1);
	\node[anchor = west] at (3, 3/2) {$\neq$};
	\end{tikzpicture}
	\begin{tikzpicture}[scale = 1/3]
	\draw[line width = 2pt] (0,-1) -- (0,4);
	\draw[line width = 2pt] (3,-1) -- (3,4);
	\draw[dashed] (0,-1) -- (3,-1);
	\draw[dashed] (0,4) -- (3,4);
	\foreach \s in {1,...,4}
	{	
		\filldraw[black] (0,\s -1) circle (5pt);
	};
	\filldraw[black] (3,1) circle (5pt);
	\filldraw[black] (3,2) circle (5pt);
	\draw (0,0) .. controls (1,0) and (1,3) .. (0,3);
	\draw (0,1) .. controls (3/4,1) and (3/4,2) .. (0,2);
	\draw (3,1) .. controls (2,1) and (2,2) .. (3,2);
	\draw (3/2,3/2) circle (12 pt); 
	\node[anchor = west] at (3, 3/2) {$.$};
	\end{tikzpicture}
\end{equation*}
The \emph{rank} of a diagram $a$ is the minimal number of times, among all isotopic diagrams equivalent to $a$, that lines intersect the top of the fundamental rectangle. The ranks of the three above examples are $0$, $1$ and $0$ respectively. The parity of a diagram is defined as the parity of its rank. Lines that connect vertices on opposite sides of the rectangle are called \emph{through lines}; a periodic $(m,n)$-diagram having exactly $n$ ($m$) through lines is said to be \emph{monic} (\emph{epic}).
 
Composition of morphisms is defined by linearly extending the concatenation of diagrams. The composition $ba$ of two periodic diagrams $a: m \to n $ and $b: n \to t$ is obtained by identifying the vertices $1$ to $n$ on $a$ and $b$, and the sides on which they lie, joining the lines that meet there, and removing the vertices. For example
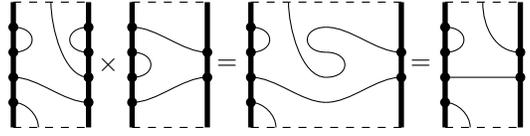
\begin{equation*}
	\begin{tikzpicture}[scale = 1/3]
	\draw[line width = 2pt] (0,-1) -- (0,4);
	\draw[line width = 2pt] (3,-1) -- (3,4);
	\draw[dashed] (0,-1) -- (3,-1);
	\draw[dashed] (0,4) -- (3,4);
	\foreach \s in {1,...,4}
	{	
		\filldraw[black] (0,\s -1) circle (5pt);
		\filldraw[black] (3,\s -1) circle (5pt);
	};
	\draw (0,0) .. controls (1/2,0) and (1,-1/2) .. (1,-1);
	\draw (3/2,4) .. controls (3/2,3) and (2,1) .. (3,1); 
	\draw (0,1) .. controls (1,1) and (2,0) .. (3,0);
	\draw (0,2) .. controls (1,2) and (1,3) .. (0,3);
	\draw (3,2) .. controls (2,2) and (2,3) .. (3,3);
	\node[anchor = west] at (3, 3/2) {$\times$};
	\end{tikzpicture}
	\begin{tikzpicture}[scale = 1/3]
	\draw[line width = 2pt] (0,-1) -- (0,4);
	\draw[line width = 2pt] (3,-1) -- (3,4);
	\draw[dashed] (0,-1) -- (3,-1);
	\draw[dashed] (0,4) -- (3,4);
	\foreach \s in {1,...,4}
	{	
		\filldraw[black] (0,\s -1) circle (5pt);
	};
	\filldraw[black] (3,1) circle (5pt);
	\filldraw[black] (3,2) circle (5pt);
	\draw (0,0) .. controls (1,0) and (2,1) .. (3,1);
	\draw (0,1) .. controls (1,1) and (1,2) .. (0,2);
	\draw (0,3) .. controls (1,3) and (2,2) .. (3,2);
	\node[anchor = west] at (3, 3/2) {$=$};
	\end{tikzpicture}
	\begin{tikzpicture}[scale = 1/3]
	\draw[line width = 2pt] (0,-1) -- (0,4);
	\draw[line width = 2pt] (6,-1) -- (6,4);
	\draw[dashed] (0,-1) -- (6,-1);
	\draw[dashed] (0,4) -- (6,4);
	\foreach \s in {1,...,4}
	{	
		\filldraw[black] (0,\s -1) circle (5pt);
	};
	\filldraw[black] (6,1) circle (5pt);
	\filldraw[black] (6,2) circle (5pt);
	\draw (0,0) .. controls (1/2,0) and (1,-1/2) .. (1,-1);
	\draw (3/2,4) .. controls (3/2,3) and (2,1) .. (3,1); 
	\draw (0,1) .. controls (1,1) and (2,0) .. (3,0);
	\draw (0,2) .. controls (1,2) and (1,3) .. (0,3);
	\draw (3,2) .. controls (2,2) and (2,3) .. (3,3);
	\draw (3,0) .. controls (4,0) and (5,1) .. (6,1);
	\draw (3,1) .. controls (4,1) and (4,2) .. (3,2);
	\draw (3,3) .. controls (4,3) and (5,2) .. (6,2);
	\node[anchor = west] at (6, 3/2) {$ = $};
	\end{tikzpicture}
	\begin{tikzpicture}[scale = 1/3]
	\draw[line width = 2pt] (0,-1) -- (0,4);
	\draw[line width = 2pt] (3,-1) -- (3,4);
	\draw[dashed] (0,-1) -- (3,-1);
	\draw[dashed] (0,4) -- (3,4);
	\foreach \s in {1,...,4}
	{	
		\filldraw[black] (0,\s -1) circle (5pt);
	};
	\filldraw[black] (3,1) circle (5pt);
	\filldraw[black] (3,2) circle (5pt);
	\draw (0,0) .. controls (1/2,0) and (1,-1/2) .. (1,-1);
	\draw (3/2,4) .. controls (3/2,3) and (2,2) .. (3,2); 
	\draw (0,1) -- (3,1);
	\draw (0,2) .. controls (1,2) and (1,3) .. (0,3);
	\node[anchor = west] at (3, 3/2) {$.$};
	\end{tikzpicture}
\end{equation*}
The following $(n,n)$-diagrams will appear below as generators of the Temperley-Lieb algebras. If $n\geq 2$, $e=e_{n,0}$ is the periodic diagram where the vertices $n$ and $1$ on each side are linked, while for all $1<j<n$, the vertex $j$ on the left is linked to the vertex $j$ on the right. For instance,
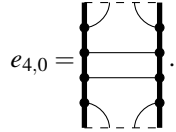
\begin{equation*}
	\begin{tikzpicture}[scale = 1/3]
	\node[anchor = east] at (0,3/2) {$e_{4,0} = $};
	\draw[line width = 2pt] (0,-1) -- (0,4);
	\draw[line width = 2pt] (3,-1) -- (3,4);
	\draw[dashed] (0,-1) -- (3,-1);
	\draw[dashed] (0,4) -- (3,4);
	\foreach \s in {1,...,4}
	{	
		\filldraw[black] (0,\s -1) circle (5pt);
		\filldraw[black] (3,\s -1) circle (5pt);
	};
	\draw (0,0) .. controls (1/2,0) and (1,-1/2) .. (1,-1);
	\draw (0,3) .. controls (1/2,3) and (1,7/2) .. (1,4);
	\draw (3,0) .. controls (5/2,0) and (2,-1/2) .. (2,-1);
	\draw (3,3) .. controls (5/2,3) and (2,7/2) .. (2,4);
	\draw (0,1) -- (3,1);
	\draw (0,2) -- (3,2);
	\node[anchor = west] at (3, 3/2) {.};
	\end{tikzpicture}
\end{equation*}
The diagram $\tau_n$ is the affine diagram where the $j$th vertex on the left side is linked to the $(j-1)$-th vertex  on the right, for all $j$. Similarly, $\tau^{-1}_n$ is the affine diagram where the vertex $j$ on the left side is linked to the vertex $j+1$ on the right, for all $j$. 
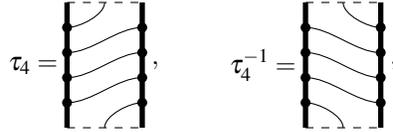
\begin{equation*}
	\begin{tikzpicture}[scale = 1/3]
	\node[anchor = east] at (0,3/2) {$\tau_4 = $};
	\draw[line width = 2pt] (0,-1) -- (0,4);
	\draw[line width = 2pt] (3,-1) -- (3,4);
	\draw[dashed] (0,-1) -- (3,-1);
	\draw[dashed] (0,4) -- (3,4);
	\foreach \s in {1,...,4}
	{	
		\filldraw[black] (0,\s -1) circle (5pt);
		\filldraw[black] (3,\s -1) circle (5pt);
	};
	\draw (3,0) .. controls (5/2,0) and (3/2,-1/2) .. (3/2,-1);
	\draw (3/2,4) .. controls (3/2,7/2) and (1/2,3) .. (0,3); 
	\draw (0,0) .. controls (1,0) and (2,1) .. (3,1);
	\draw (0,1) .. controls (1,1) and (2,2) .. (3,2);
	\draw (0,2) .. controls (1,2) and (2,3) .. (3,3);
	\node[anchor = west] at (3, 3/2) {$,$};
	\end{tikzpicture}\qquad
	\begin{tikzpicture}[scale = 1/3]
	\node[anchor = east] at (0,3/2) {$\tau^{-1}_4 = $};
	\draw[line width = 2pt] (0,-1) -- (0,4);
	\draw[line width = 2pt] (3,-1) -- (3,4);
	\draw[dashed] (0,-1) -- (3,-1);
	\draw[dashed] (0,4) -- (3,4);
	\foreach \s in {1,...,4}
	{	
		\filldraw[black] (0,\s -1) circle (5pt);
		\filldraw[black] (3,\s -1) circle (5pt);
	};
	\draw (0,0) .. controls (1/2,0) and (3/2,-1/2) .. (3/2,-1);
	\draw (3/2,4) .. controls (3/2,7/2) and (5/2,3) .. (3,3); 
	\draw (0,1) .. controls (1,1) and (2,0) .. (3,0);
	\draw (0,2) .. controls (1,2) and (2,1) .. (3,1);
	\draw (0,3) .. controls (1,3) and (2,2) .. (3,2);
	\node[anchor = west] at (3, 3/2) {$,$};
	\end{tikzpicture}
\end{equation*}
One directly sees that $\tau^{-1}_n\circ \tau_n = \tau_n \circ \tau^{-1}_n$ is the identity on $\Hom_{\cat}(n,n)$. Using these three diagrams we then define $e_{n,i} = \tau^{i}_ne_{n,0}\tau^{-i}_n$, for any integer $i$; note that $(\tau_n)^{n}$ is central in $\End_{\cat}{n}$, so that $e_{n,i+n} = e_{n,i}$.

It should be noted that if $a$ and $b$ are both odd or even, then $b \circ a$ is even, while if $a$ and $b$ have different parity, then $b \circ a$ is odd. Furthermore, if $a$ has $r$ through lines and $b$ has $s$, then $b\circ a$ has at most $\text{min}(r,s)$ through lines. Finally, we define an involution on diagrams: it sends a diagram $d$ unto its mirror reflection $d^t$ with respect to an imaginary axis that splits the fundamental rectangle vertically. For instance,
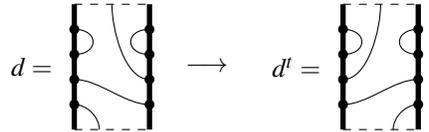
\begin{equation*}
	\begin{tikzpicture}[scale = 1/3]
	\draw[line width = 2pt] (0,-1) -- (0,4);
	\draw[line width = 2pt] (3,-1) -- (3,4);
	\draw[dashed] (0,-1) -- (3,-1);
	\draw[dashed] (0,4) -- (3,4);
	\node[anchor=east] at (-0.25,1.5) {$d =\ $};
	\foreach \s in {1,...,4}
	{	
		\filldraw[black] (0,\s -1) circle (5pt);
		\filldraw[black] (3,\s -1) circle (5pt);
	};
	\draw (0,0) .. controls (1/2,0) and (1,-1/2) .. (1,-1);
	\draw (3/2,4) .. controls (3/2,3) and (2,1) .. (3,1); 
	\draw (0,1) .. controls (1,1) and (2,0) .. (3,0);
	\draw (0,2) .. controls (1,2) and (1,3) .. (0,3);
	\draw (3,2) .. controls (2,2) and (2,3) .. (3,3);
	\node[anchor = west] at (3, 3/2) {$\quad\longrightarrow\quad$};
	\end{tikzpicture}
	\begin{tikzpicture}[scale = 1/3]
	\draw[line width = 2pt] (0,-1) -- (0,4);
	\draw[line width = 2pt] (3,-1) -- (3,4);
	\draw[dashed] (0,-1) -- (3,-1);
	\draw[dashed] (0,4) -- (3,4);
	\node[anchor=east] at (-0.25,1.5) {$d^t =\ $};
	\foreach \s in {1,...,4}
	{	
		\filldraw[black] (0,\s -1) circle (5pt);
		\filldraw[black] (3,\s -1) circle (5pt);
	};
	\draw (3,0) .. controls (5/2,0) and (2,-1/2) .. (2,-1);
	\draw (3/2,4) .. controls (3/2,3) and (1,1) .. (0,1); 
	\draw (0,0) .. controls (1,0) and (2,1) .. (3,1);
	\draw (0,2) .. controls (1,2) and (1,3) .. (0,3);
	\draw (3,2) .. controls (2,2) and (2,3) .. (3,3);
	\node[anchor = west] at (3, 3/2) {$.$};
	\end{tikzpicture}
\end{equation*}
\begin{Def}
	The \emph{affine Temperley-Lieb category} $\catl(q)$ is the quotient of category $\cat$ obtained by the identification of diagrams $a$ having $n$ loops, with the morphism $\beta^{n} a$, where $\beta = q+q^{-1}$ for some invertible element $q \in\ringK$. 
\end{Def}
\noindent For example, in $\catl$, the two morphisms in $\Hom_{\catl(q)}(2,4)$ are equal:
 
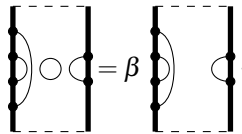
\begin{equation*}
	\begin{tikzpicture}[scale = 1/3]
	\draw[line width = 2pt] (0,-1) -- (0,4);
	\draw[line width = 2pt] (3,-1) -- (3,4);
	\draw[dashed] (0,-1) -- (3,-1);
	\draw[dashed] (0,4) -- (3,4);
	\foreach \s in {1,...,4}
	{	
		\filldraw[black] (0,\s -1) circle (5pt);
	};
	\filldraw[black] (3,1) circle (5pt);
	\filldraw[black] (3,2) circle (5pt);
	\draw (0,0) .. controls (1,0) and (1,3) .. (0,3);
	\draw (0,1) .. controls (3/4,1) and (3/4,2) .. (0,2);
	\draw (3,1) .. controls (2,1) and (2,2) .. (3,2);
	\draw (3/2,3/2) circle (12 pt); 
	\node[anchor = west] at (3, 3/2) {$= \beta $};
	\end{tikzpicture}
	\begin{tikzpicture}[scale = 1/3]
	\draw[line width = 2pt] (0,-1) -- (0,4);
	\draw[line width = 2pt] (3,-1) -- (3,4);
	\draw[dashed] (0,-1) -- (3,-1);
	\draw[dashed] (0,4) -- (3,4);
	\foreach \s in {1,...,4}
	{	
		\filldraw[black] (0,\s -1) circle (5pt);
	};
	\filldraw[black] (3,1) circle (5pt);
	\filldraw[black] (3,2) circle (5pt);
	\draw (0,0) .. controls (1,0) and (1,3) .. (0,3);
	\draw (0,1) .. controls (3/4,1) and (3/4,2) .. (0,2);
	\draw (3,1) .. controls (2,1) and (2,2) .. (3,2);
	\node[anchor = west] at (3, 3/2) {$.$};
	\end{tikzpicture}
\end{equation*}
It can be shown that an affine $(m,n)$-diagram becomes a monomorphism (epimorphism) in $\catl{(q)}$ if and only if it is monic (epic). Furthermore, it can be shown that every affine $(m,n)$-diagram can be expressed as $a \circ f \circ b$, where $a \in \End{m}$, $b \in \End{n}$, and $f \in \Hom_{\catl}(n,m)$ is a diagram of rank zero.  


\begin{Def}
	The \emph{regular Temperley-Lieb category} $\ctl(q)$ is the subcategory of $\catl(q)$ where $\Hom_{\ctl}(n,m)$, for $m,n\in\mathbb N_{\geq0}$, is spanned by all diagrams of rank zero.  
\end{Def}

The category $\ctl{(q)}$ is a strict monoidal braided category \cite{turaev}. The tensor product on objects is simply $n \otimes_{\ctl} m \equiv n+m$ while, on morphisms, it is defined by extending bilinearly the tensor product of diagrams. Let $f$ be a periodic $(m,r)$-diagram, and $g$ be a periodic $(n,s)$-diagram, both of rank zero. The tensor product $f\otimes g$ is the $(m+n,r+s)$-diagram obtained by \emph{gluing} $f$ on top of $g$. For instance, here is the product of a $(3,1)$-diagram with a $(2,2)$-diagram:
\begin{equation}
	\begin{tikzpicture}[scale = 1/3]
	\draw[line width = 2pt] (0,-1) -- (0,3);
	\draw[line width = 2pt] (3,-1) -- (3,3);
	\draw[dashed] (0,-1) -- (3,-1);
	\draw[dashed] (0,3) -- (3,3);
	\foreach \s in {1,...,3}
	{	
		\filldraw[black] (0,\s -1) circle (5pt);
	};
	\filldraw[black] (3,2) circle (5pt);
	\draw (0,0) .. controls (1,0) and (2,2) .. (3,2);
	\draw (0,1) .. controls (1,1) and (1,2) .. (0,2);
	\node[anchor = west] at (3, 3/4) {$\otimes_{\ctl}$};
	\node at (3/2,-2) {\phantom{}};
	\end{tikzpicture}
	\begin{tikzpicture}[scale = 1/3]
	\draw[line width = 2pt] (0,-1) -- (0,2);
	\draw[line width = 2pt] (3,-1) -- (3,2);
	\draw[dashed] (0,-1) -- (3,-1);
	\draw[dashed] (0,2) -- (3,2);
	\foreach \s in {1,...,2}
	{	
		\filldraw[black] (0,\s -1) circle (5pt);
		\filldraw[black] (3,\s -1) circle (5pt);
	};
	\draw (0,0) .. controls (1,0) and (1,1) .. (0,1);
	\draw (3,0) .. controls (2,0) and (2,1) .. (3,1);
	\node[anchor = west] at (3, 1/4) {$\equiv$};
	\node at (3/2,-5/2) {\phantom{}};
	\end{tikzpicture}
	\begin{tikzpicture}[scale = 1/3]
	\draw[line width = 2pt] (0,-1) -- (0,5);
	\draw[line width = 2pt] (3,-1) -- (3,5);
	\draw[dashed] (0,-1) -- (3,-1);
	\draw[dashed] (0,5) -- (3,5);
	\foreach \s in {1,...,5}
	{	
		\filldraw[black] (0,\s -1) circle (5pt);
	};
	\filldraw[black] (3,4) circle (5pt);
	\filldraw[black] (3,3) circle (5pt);
	\filldraw[black] (3,2) circle (5pt);
	\draw (0,2) .. controls (1,2) and (2,4) .. (3,4);
	\draw (0,3) .. controls (1,3) and (1,4) .. (0,4);
	\draw (0,0) .. controls (1,0) and (1,1) .. (0,1);
	\draw (3,2) .. controls (2,2) and (2,3) .. (3,3);
	\node[anchor = west] at (3, 2) {.};
	\end{tikzpicture}
\end{equation}
The commutor (or braiding) is a natural isomorphism $\eta:-_1\otimes -_2\rightarrow -_2\otimes -_1$ such that $\eta_{r,s}\circ(f\otimes g)=(g\otimes f)\circ \eta_{m,n}$, for all $m,n,r,s\in\mathbb N_{\geq0}$, $f\in\Hom(m,r)$ and $g\in\Hom(n,s)$. For $\ctl(q)$, the commutors are given by
	\begin{align}\label{eq:commutor}
		\com{r}{s} = \prod_{i=1}^s\big(\prod_{j=r-1}^0 \T{r+s,i+j} \big) = \prod_{i=r}^1\big(\prod_{j=0}^{s-1} \T{r+s,i+j}  \big),
	\end{align}
where, for $n\geq 2$ and $1 \leq i \leq n-1$, the $\T{n,i}$ can be chosen as
	\begin{equation}
		\T{n,i} = (-q)^{1/2} 1_{\End{n}} + (-q)^{-1/2} e_{n,i}.
	\end{equation}
and the factors in a product are listed starting from the right, that is, $\prod_{i=1}^s\T{i}=\T s\T{s-1}\dots\T2\T1$ and $\prod_{i=s}^1\equiv\T1\T2\dots\T{s-1}\T s$. (Whenever the context is clear, the number $n$ of sites on the elements of $\tl n$ will omitted from now on.) The elements $\T{i}$ satisfy many different identities; in particular we note 
\begin{equation}\label{eq:magic.t.and.e}
	e_{i}\T{i+1}\T{i} = e_{i}e_{i+1} = \T{i}\T{i+1}e_{i+1}, \qquad \T{i}\T{i+1}\T{i} = \T{i+1}\T{i}\T{i+1},\quad \T{i}\T{j} = \T{j}\T{i}, \quad \forall |i-j|\geq 2.
\end{equation}
Furthermore, it can be shown that for all $n>1$, $(\eta_{1,n-1})^{n}$ and $(\eta_{n-1,1})^{n}$ are both central in $\End_{\ctl}(n)$.
\end{subsection}
%
%
\begin{subsection}{The Temperley-Lieb algebras}\label{sec:tlalgebras}
The \emph{affine Temperley-Lieb algebra} $\atl{n}(q)$ is defined as the $\ringK$-algebra of endomorphisms of $n$ in $\catl(q)$. Equivalently, it is the $\ringK$-algebra generated by $e_0,\tau$, and $\tau^{-1}$ (that is, $e_{n,0}, \tau_n$ and $\tau^{-1}_n$). The defining relations are, for $n > 2$,
\begin{equation}\label{eq:rel.tl.1}
	e_i e_i = (q+q^{-1})e_{i}, \qquad e_{i}e_{i \pm 1}e_i = e_i, \quad \forall i = 1, 2,\hdots n,
\end{equation}
\begin{equation}\label{eq:rel.tl.2}
	e_i e_j = e_j e_i, \quad \text{ if } |i-j| \geq 2,
\end{equation}
\begin{equation}\label{eq:rel.tl.3}
	\tau e_i = e_{i+1}\tau, \qquad e_1\tau^{2} = e_1e_2e_3\hdots e_{n-1}, \qquad \tau\tau^{-1}=\tau^{-1}\tau=1_{\End{n}},
\end{equation}
where $e_i = \tau^{i}e_0\tau^{-i}$. (As noted earlier, $e_{i+n}=e_i$.) The case $n= 2 $ is slightly different, the relations are then simply
\begin{equation}
	e_{1}e_{1} = (q+q^{-1})e_{1}, \qquad \tau^{2}e_{1} = e_{1}\tau^{2} = e_{1}.
\end{equation}

The \emph{regular Temperley-Lieb algebra} $\tl{n}(q)$ is defined as the $\ringK$-algebra of endomorphisms of $n$ in $\ctl(q)$. Equivalently, it is the $\ringK$-algebra generated\footnote{While $u_{n,i}$ and $e_{n,i}$ correspond to the same diagram, we use a different notation to distinguish elements of $\tl{n}$ and $\atl{n}$.} by the $u_i \equiv e_i$ for $1\leq i \leq n-1$, with the defining relations \eqref{eq:rel.tl.1} and \eqref{eq:rel.tl.2} for $n\geq 2$ but with the index $i$ limited to $1\leq i\leq n-1$. The definition makes it clear that $\tl{n} \subset \atl{n}$ as algebras. What is less obvious is that $\tl{n}$ is isomorphic to a quotient algebra of $\atl{n}$.~\cite{GaiSalAffine}
\begin{Prop}\label{prop:defphi}
Let $\phi: \atl{n} \to \tl{n}$ be the linear map defined, for $n>2$, by
	\begin{equation}
		\phi(\tau) = (-q)^{3/2}\eta_{n-1,1} , \qquad \phi(e) = \com{n-1}{1}^{-1}u_{1}\eta_{n-1,1};
	\end{equation}
for $n=2$, by 
	\begin{equation}
		\phi(\tau) = (-q)^{3/2}\eta_{1,1}=q(q\cdot 1_{\tl2} -u_1), \qquad \phi(e_1) = u_{1};
	\end{equation}
and for $n=1$, by
$$\phi(\tau) = (-q)^{3/2} 1_{\tl{1}}.$$
This map is a surjective morphism of algebras.
\end{Prop}
\begin{proof}
	We need to verify that $\phi(\tau)$ and $\phi(e)$ satisfy all the affine Temperley-Lieb relations, namely equations \eqref{eq:rel.tl.1}  to \eqref{eq:rel.tl.3}. We write $\eta \equiv \com{n-1}{1}$, to shorten the notation. For $n=1$ all conditions are trivially satisfied. 
For $n = 2$, one must simply recognize that
	 		\begin{equation*}
	 			\phi(\tau)\phi(e_1) = q(q\cdot 1_{\tl2} -u_1)u_{1} = - u_{1} = \phi(e_1)\phi(\tau),
	 		\end{equation*}
	 	so that $(\phi(\tau))^{2}\phi(e_{1}) = \phi(e_{1})(\phi(\tau))^{2} = -(-u_{1}) = u_{1}$.
	 For $n>2$, the first equation is simply
	\begin{equation}
		\phi(e_i)\phi(e_i) = \eta^{i-1}u_1\eta^{1-i}\eta^{i-1}u_1\eta^{1-i} = \beta\eta^{i-1}u_1\eta^{1-i} = \beta \phi(e_i).
	\end{equation}
	The second condition is 
	\begin{align*}
		\phi(e_i)\phi(e_{i + 1})\phi(e_i) &= \eta^{i-1}u_1\eta^{1-i}\eta^{i}u_1\eta^{-i}\eta^{i-1}u_1\eta^{1-i}\\
			&= \eta^{i-1} (u_1\eta u_1 \eta^{-1}u_1) \eta^{1-i}\\
			&= \eta^{i-1} (u_1u_2\eta\eta^{-1}u_1) \eta^{1-i}\\
			&= \eta^{i-1} u_1 \eta^{1-i} = \phi(e_i),
	\end{align*}
where we used the fact that $\eta$ is a commutor, that $u_{n,1} = u_{n-1,1} \otimes 1_{\End{1}}$, and that $u_1u_2u_1 = u_1$. The same arguments also give
	\begin{align*}
		\phi(e_i)\phi(e_{i - 1})\phi(e_i) & = \eta^{i-1}u_1\eta^{1-i}\eta^{i-2}u_1\eta^{2-i}\eta^{i-1}u_1\eta^{1-i}\\
			& = \eta^{i-1}u_1\eta^{-1}u_1\eta u_1\eta^{1-i}\\
			& = \eta^{i-2}u_2 u_1 u_2 \eta^{2-i} = \eta^{i-2} u_2 \eta^{2-i} \\
			& = \phi(e_i),
	\end{align*}
	so that the conditions \eqref{eq:rel.tl.1} are satisfied. Next, suppose that $n > i\geq 2+j>0$, then
	\begin{align*}
		\phi(e_i)\phi(e_j) &= \eta^{i-1} u_1 \eta^{j-i} u_1 \eta^{1-j} = \eta^{j-1}( \eta^{i-j}u_1 \eta^{j-i} u_1 \eta^{i-j})\eta^{1-i}\\
		& = \eta^{j-1}( u_{1+i-j} u_1 \eta^{i-j})\eta^{1-i}\\
		& = \eta^{j-1}( u_1 \eta^{i-j} u_1 )\eta^{1-i} = \phi(e_j)\phi(e_i),
	\end{align*}
	where we again used properties of the commutors with the fact that the $u_k$s themselves satisfy \eqref{eq:rel.tl.2}. Finally, repeatedly applying equation \eqref{eq:magic.t.and.e} gives
	\begin{align*}
(-q)^{-3}\phi(e_1\tau^{2})= u_1 \eta^2 	
							& = u_1 \T{1}\T{2}\T{3}\hdots\T{n-1}\T{1}\T{2}\T{3}\hdots \T{n-1}\\
							& = u_1 (\T{1}\T{2}\T{1})\T{3}\hdots\T{n-1}\T{2}\T{3}\hdots \T{n-1}\\
							& =(u_1 \T{2}\T{1})\T{2}\T{3}\hdots\T{n-1}\T{2}\T{3}\hdots \T{n-1}\\
							& = u_1u_2(\T{2}\T{3}\T{2})\T{4}\hdots\T{n-1}\T{3}\hdots \T{n-1}\\
							& = (u_1u_2\T{3}\T{2})\T{3}\T{4}\hdots \T{n-1}\T{3}\hdots\T{n-1}\\
							& = u_1u_2u_3\T{3}\T{4}\hdots\T{n-1}\T{3}\hdots \T{n-1}\\
							& = \hdots\\
							& = u_1u_2u_3u_4\hdots u_{n-1}\T{n-1}\T{n-1}.
	\end{align*}
	We can then check directly that $u_{n-1}\T{n-1}\T{n-1} = (-q)^{-3}u_{n-1}$, so that $$\phi(e_1)\eta^{2} = \phi(e_1)\phi(e_2)\hdots \phi(e_{n-1}).$$
	 Finally, one checks that the generators $u_i$ are all in the image of $\phi$.
\end{proof}

The morphism $\phi$ is (almost) an idempotent. More precisely:
\begin{Prop}Let $\iota: \tl{n} \to \atl{n}$, be the canonical injection. Then $\phi\circ\iota\circ\phi=\phi$.
\end{Prop}
\begin{proof}This follows from the fact that $\phi(e_i)=u_i$:
\begin{align*}
\phi(e_i)&=\phi(\tau^i e_0\tau^{-i})=(-q)^{3i/2}\eta^i_{n-1,1}(\eta_{n-1,1}^{-1}u_1\eta_{n-1,1})(-q)^{-3i/2}\eta^{-i}_{n-1,1}\\
&=\eta_{n-1,1}^{i-1}u_1\eta_{n-1,1}^{-i+1}=u_{i-1+1}\eta_{n-1,1}^{i-1}\eta_{n-1,1}^{-i+1}=u_i.
\end{align*}
Then $\phi(\eta_{n-1,1})=\eta_{n-1,1}$ where, in the latter $\eta$, the $e_i$ are replaced by $u_i$. From these observations, $\phi\circ\iota\circ\phi(\tau)=\phi(\tau)$ follows easily and
$$\phi\circ\iota\circ\phi(e_0)=\phi(\eta_{n-1,1}^{-1}e_1\eta_{n-1,1})=\phi(\eta_{n-1,1}^{-1}\tau e_0\tau^{-1}\eta_{n-1,1})=\phi(e_0).$$
\end{proof}

The various morphisms between the Temperley-Lieb algebras produces many functors between their module categories. We shall use these to study the module categories themselves, but also to construct different classes of fusion functors on these categories.

\begin{Def}
	Again let $\iota: \tl{n} \to \atl{n}$ denote the canonical injection. The associated induction and restriction functors are
	\begin{equation}
		\Indar{}: \Mod \tl{n} \to \Mod \atl{n}, \qquad \Resar{}: \Mod \atl{n} \to \Mod \tl{n}.
	\end{equation}
Similarly the induction and restriction functors related to the morphism $\phi: \atl{n} \to \tl{n}$
	\begin{equation}
		\Indphi{}: \Mod\atl{n} \to \Mod\tl{n}, \qquad \Resphi{}: \Mod \tl{n} \to \Mod{\atl{n}}
	\end{equation}
are given by the usual functors: $\mathsf V\mapsto\ \Indphi{(\mathsf V)}= {}_{\tl n}{\tl n}_{\atl n}\otimes_{\atl n}\mathsf V$ if $\mathsf V\in\Mod\atl n$ and $\mathsf X\mapsto\ \Resphi(\mathsf X)=\Hom_{\tl n}(\tl n,\mathsf X)$ if $\mathsf X\in\Mod\tl n$. (The right $\atl n$-action on $\tl n$ is through the morphism $\phi$.) Both $(\Indar{}, \Resar)$ and $(\Indphi,\Resphi)$ are adjoint pairs.
\end{Def}

The compositions of some of these functors are simple. To study the first identity, we indicate by subscripts which action is at play. For example, because of $\phi$, $\tl n$ can be understood as both a left and a right $\atl n$-module. Consider ${}_{\tl n}{\tl n}_{\atl n} \otimes_{\atl n}{\atl n}_{\tl n}$. Since $\phi$ restricted to $\tl n$ is the identity, any element of the right factor $\atl n$ can be moved to the left one. Thus 
${}_{\tl n}{\tl n}_{\atl n} \otimes_{\atl n}{\atl n}_{\tl n}\simeq {}_{\tl n}{\tl n}_{\tl n}$ and 
$$\Indphi{} \circ \Indar{}(\mathsf X)={}_{\tl n}{\tl n}_{\atl n}\otimes_{\atl n}({\atl n}\otimes_{\tl n}\mathsf X)\simeq{}_{\tl n}\tl n\otimes_{\tl n}\mathsf X\simeq \mathsf X$$
for any $\tl n$-module $\mathsf X$, that is, $\Indphi{} \circ \Indar{} \xrightarrow{\sim} \id_{\Mod\tl{n}}$. To see the second identity, $\Resar{} \circ \Resphi{} \xrightarrow{\sim} \id_{\Mod\tl{n}}$, recall that $f\in\Hom_{\tl n}(\tl n,\mathsf X)$ is completely determined by its value at $1_{\tl n}$ and there is thus a one-to-one correspondence $f\leftrightarrow f(1_{\tl n})=x_f\in \mathsf X$. The action of $\atl n$ on $\Resphi{(\mathsf X)}=\Hom_{\tl n}(\tl n,\mathsf X)$ is given by $(af)(x)=f(x\phi(a))$, $a\in\atl n$, which corresponds to
\begin{equation}\label{eq:oneToOne}af\quad\longleftrightarrow\quad x_{af}=(af)(1_{\tl n})=f(1_{\tl n}\phi(a))=\phi(a)x_f\end{equation}
which is equal to $ax_f$ if $a$ is in $\tl n$. In other words the correspondence $f\leftrightarrow x_f$ is a $\tl n$-isomorphism and $\Resar{} \circ \Resphi{} \xrightarrow{\sim} \id_{\Mod\tl{n}}$ as claimed. The third identity, $\Indphi{}\circ\Resphi{} \xrightarrow{\sim} \id_{\Mod\tl{n}}$, follows from the surjectivity of $\phi$. For all $a\in\tl n$, there exists a $b\in\atl n$ such that $\phi(b)=a$ and, thus, any element $a\otimes_{\atl n}f\in \Indphi{}\circ\Resphi{}(\mathsf X)$ can be written as $1_{\tl n}\otimes_{\atl n}bf$. Again, because of the correspondence between elements $f\in\Hom(\tl n,\mathsf X)$ and the elements $x$ of $\mathsf X$, it follows that $\Indphi{}\circ\Resphi{}(\mathsf X)$ and $\mathsf X$ are the same as vector spaces. The actions of $\tl n$ on $\Indphi{}\circ\Resphi{}(\mathsf X)$ and $\mathsf X$ coincide because of \eqref{eq:oneToOne}. The following lemma sums up these observations.
\begin{Lem}\label{lem:functorId}The restriction and induction functors satisfy
\begin{equation}\label{eq:functorId1}
\Indphi{} \circ \Indar{} \xrightarrow{\sim}\ \Indphi{}\circ\Resphi{} \xrightarrow{\sim}\id_{\Mod\tl{n}},\end{equation}
\begin{equation}\label{eq:functorId2}
\text{\rm and}\quad\Resar{} \circ \Resphi{} \xrightarrow{\sim} \id_{\Mod\tl{n}}.
\end{equation}
\end{Lem}
\begin{Prop}\label{prop:res.ind.hom}
	For $\mathsf U,\mathsf V \in \Mod\tl{n}$,
	\begin{equation}\label{eq:resphi.and.hom}
		\Hom_{\atl{n}}\left(\Resphi{\mathsf U},\Resphi{\mathsf V}\right) \simeq \Hom_{\tl{n}}\left(\mathsf U,\mathsf V\right).
	\end{equation} 
In particular, $\Resphi{\mathsf V} \simeq \Resphi{\mathsf U} $ if and only if $\mathsf U \simeq \mathsf V$. Furthermore,
	\begin{equation}\label{eq:resphi.and.hom.2}
		\Hom_{\atl{n}}\left(\Indar{\mathsf U},\Resphi{\mathsf V}\right) \simeq \Hom_{\tl{n}}\left(\mathsf U,\mathsf V\right),
	\end{equation}
and $\Indar{\mathsf U} \simeq \Indar{\mathsf V}$ if and only if $\mathsf U \simeq \mathsf V$.
\end{Prop}
\begin{proof}Both \eqref{eq:resphi.and.hom} and \eqref{eq:resphi.and.hom.2} are consequences of the Frobenius theorem for adjoint pairs. For example, by \eqref{eq:functorId1}
$$\Hom_{\tl n}(\mathsf U,\mathsf V)\simeq\Hom_{\tl n}(\Indphi{}\circ\Resphi{}\mathsf U,\mathsf V)\simeq\Hom_{\atl n}(\Resphi{\mathsf U},\Resphi{\mathsf V}).$$
Furthermore, if $\Resphi{\mathsf U} \simeq \Resphi{\mathsf V}$, then for all $ \mathsf W \in \Mod\tl{n}$,
	\begin{equation}
		\Hom_{\tl{n}}\left(\mathsf U,\mathsf W\right) \simeq \Hom_{\atl{n}}\left(\Resphi{\mathsf U},\Resphi{\mathsf W}\right)  \simeq \Hom_{\atl{n}}\left(\Resphi{\mathsf V},\Resphi{\mathsf W}\right) \simeq \Hom_{\tl{n}}\left(\mathsf V,\mathsf W\right).
	\end{equation}
Since the $\Hom$ functors are fully faithful, it follows that $\mathsf U \simeq \mathsf V$. The other statements are obtained by similar arguments.
\end{proof}
Finally, we note that since induction functors preserves projectivity, if $U$ is a projective $\tl{n}$-module, and $\mathsf M$ is a projective $\atl{n}$-module, then $\Indar{\mathsf U}$, and $\Indphi{\mathsf M}$ are also projective. In particular, starting from the Peirce decomposition of the regular Temperley-Lieb algebra (which is known), the functor $\Indar{}$ gives the Peirce decomposition of the affine Temperley-Lieb algebra. Such a decomposition will be obtained in section \ref{sub:Indar} when $q$ if not a root of unity.
\end{subsection}
%
%
\begin{subsection}{The fusion functors}\label{sec:fusionFunctors}
Let $R$ be a commutative ring, and $\cat$ be an $R$-linear category; a \emph{fusion product} on $\cat$ is given by a family of bifunctors $ (\blank \xf \blank)_{A,B} : \Mod \End{A} \times \Mod \End{B} \to \Mod \End{C_{A,B}} $ for objects $A,B,C_{A,B} \in \cat$, called \emph{fusion functors}. For instance, a monoidal category can be endowed with a fusion product in a very natural manner: for any pair of objects $A,B$, define
\begin{equation}
	( \blank \xf \blank)_{A,B} = \End_{\cat}\left(A\otimes B\right) \otimes_{\End{A}\otimes\End{B}} \left( \blank \otimes_{R} \blank \right),
\end{equation}
where $\otimes_{R} $ is the tensor product on free $R$-modules. A fusion product is said to be \emph{commutative} if for all objects $A,B$, $ (-_1 \xf -_2)_{A,B} \xrightarrow{\sim} (-_2 \xf -_1)_{B,A}$; this assumes of course that $C_{A,B} = C_{B,A}$. Similarly, a fusion product is said to be \emph{associative} if  there is a natural isomorphism $\left((\blank,\blank)_{A,B}, \blank \right)_{C_{A,B},D} \xrightarrow{\sim} \left(\blank,(\blank, \blank )_{B,D} \right)_{A,C_{B,D}}$ for all objects $A,B,D \in \cat$.
   
  The regular Temperley-Lieb category $\ctl$ can be equipped with a fusion product that is both associative and commutative, because it is a braided monoidal category. Our goal is to define, and compute, similar fusion products for the affine Temperley-Lieb category.
\begin{Def}\label{def:fusionFunctors}
The fusion functors $\xfone$ and $\xftwo$ are defined as
\begin{align}
  		( \blank \xfone \blank) & \equiv\ \Resphi{((\Indphi{\blank}) \xf (\Indphi{\blank})) },\label{eq:fusion1}\\
  		( \blank \xftwo \blank) & \equiv\ \Indar{((\Indphi{\blank} )\xf (\Indphi{\blank})) },\label{eq:fusion2}
\end{align}
  where $( \blank \xf \blank)$ will denote, from now on, the fusion product on the regular Temperley-Lieb category.\footnote{The functors $\xfone$ and $\xftwo$ are not functors on $\catl$, but on a category whose objects are themselves functors $\catl\to \mathsf{Vect}_{\mathbb C}$. These objects are defined in \cite{GLaffine}. A category whose objects are these functors on $\ctl$ is constructed in \cite{BSA}.}
\end{Def}
  It is relatively straightforward to show that $( \blank \xfone \blank)_{m,n}$ defines a commutative and associative fusion product on the affine Temperley-Lieb category. It is commutative because the fusion functors on $\ctl$ are, and it is associative because $\Indphi{} \circ \Resphi{}  \xrightarrow{\sim} \id_{\Mod\tl{n}}$ and the fusion product on $\ctl$ is associative. Similarly, the bifunctors $( \blank \xftwo \blank)_{m,n} $ defines a commutative and associative fusion product on the affine Temperley-Lieb category, for the same reasons.
  
  The main reason while we are interested in these particular fusion products is that their behaviour is somewhat ``bounded''. The first fusion product sends pairs of finitely-generated affine modules to a finite-dimensional module, while the second would send them to another, possibly infinite-dimensional, finitely-generated module.
  
  However, many other fusion products could be defined on this category. For instance, one could choose
  $$( \blank \xfthree \blank)_{m,n} \equiv\ \Resphi{((\Resar{\blank}) \xf (\Resar{\blank}))_{m,n} }. $$
  Since restriction functors preserve the module's dimension, if $\mathsf M$ is an infinite-dimensional affine module, say $\atl{n}$ itself, then $\Resar{\mathsf M}$ is also infinite-dimensional. Since $\tl{n}$ only has finitely many non-isomorphic indecomposable modules, all of which are finite-dimensional, it follows that $\Resar{\mathsf M}$ is a direct sum of finite modules, at least one of which must appear infinitely many times. To get a physical interpretation of this fusion, one would have to truncate these sums somehow, or only consider fusion of finite-dimensional modules. We chose to limit ourselves here to the two fusion products $( \blank \xfone \blank)$ and $( \blank \xftwo \blank)$.
\end{subsection}
\end{section}
%
\begin{section}{Modules and morphisms}\label{sec:modulesMorphism}
%

This section introduces the most common modules over both the regular algebra $\tl n$ and the affine one $\atl n$. Their $\Hom$-groups are also described.
%
%
\begin{subsection}{Modules over $\tl n$}\label{sec:modulesOfTln}
The representation theory of the algebra $\tl n$ was studied first by Goodman and Wenzl \cite{GoodWenzl93} and independently by Martin \cite{Martin}. (See also \cite{GLaffine}, \cite{WesRep95} and \cite{RSA}.) A complete list of all finite-dimensional indecomposable modules over $\tl n(\beta)$, up to isomorphisms, is known for all values of $\beta$ \cite{WesRep95, BRSA}. We shall need only a subset of these, namely the irreducible, standard and projective modules. Here is their description.

Let $n,k$ be two non-negative integers of the same parity. Let $\TheM{n,k}$ denote $\Hom_{\ctl}(k,n)$ with the left $\tl n$-action. In this $\tl n$-module, the non-monic diagrams form a submodule. The {\em standard} module $\TheS{n,k}$ is the quotient of $\TheM{n,k}$ by this submodule. If $k>n$, then all diagrams are non-monic and $\TheS{n,k}=0$. Set $\Lambda_n=\{k\in\mathbb N_0\,|\, 0\leq k\leq n\textrm{ and }k\equiv n\Mod 2\}$. Then the module $\TheS{n,k}$, $k\in\Lambda_n$, is non-zero and its dimension is expressed in terms of binomials:
\begin{equation}\label{eq:dimSnk}\dim \TheS{n,k}=\begin{pmatrix}n\\ (n-k)/2\end{pmatrix}-\begin{pmatrix}n\\ (n-k)/2-1\end{pmatrix}\end{equation}
with the convention that $\left(\begin{smallmatrix}n\\ m\end{smallmatrix}\right)=0$ if $m<0$. We identify monic diagrams in $\TheM{n,k}$ with their equivalence classes in the quotient $\TheS{n,k}$. If $q$ is not a root of unity, then $q$ is said to be {\em generic}, the algebra $\tl n(\beta)$ is semisimple, and the $\TheS{n,k},k\in\Lambda_n$, form a complete set of non-isomorphic irreducible modules over $\tl n$ and, as a left-module, $\tl n\simeq \oplus'_{0\leq k\leq n}\dim(\TheS{n,k})\cdot \TheS{n,k}$, where the $\oplus'$ indicates that the direct sum runs over $k$'s of the parity of $n$.

Let $q$ be a root of unity and $\ell\geq 2$ be the smallest integer such that $q^{2\ell}=1$. In this case, $\tl n(\beta=q+q^{-1})$ is non-semisimple (unless either $n<\ell$ or $\beta=0$ with $n$ odd), and the standard $\TheS{n,k}$ are not all irreducible. Call $k$ {\em critical} (or critical for $q$) if $k+1\equiv 0\Mod \ell$. The set $\Lambda_n$ is then partitioned as follows. If $k\in\Lambda_n$ is critical, it forms its own class in the partition. If the element $k$ is not critical, then its class $[k]$ consists of the images (in $\Lambda_n$) of $k$ generated by reflections with respect to the critical integers. If $k_c$ is a critical integer, then $2k_c-k$ is the reflection of $k$ through $k_c$.  The class of a non-critical $k$ thus contains precisely one integer between each pair of consecutive critical ones. Neighboring elements in a non-critical class $[k]$ are ordered as $k_L<\dots < k^{--}<k^-<k<k^+<k^{++}<\dots <k_R$, so that $k_L\ge 0$ and $k_R\le n$ are the smallest and largest elements in $[k]\subset\Lambda_n$. The notation $k^j$ ($k^{-j}$) is also used to refer to the $j$-th element to the right of $k$ (to its left) so that, for example, $k^{--}=k^{-2}$ and $k^{+++}=k^3$. 

If $k$ is non-critical and $k<k_R$, then $\TheS{n,k}$ is reducible but indecomposable. Let $\TheI{n,k}$ denote its irreducible quotient. Then the short sequence
\begin{equation}\label{eq:sesTheS}0\longrightarrow \TheI{n,k^+}\longrightarrow \TheS{n,k}\longrightarrow \TheI{n,k}\longrightarrow 0\end{equation}
is exact and non-split. For $k=k_R$, the standard $\TheS{n,k_R}$ is irreducible. The set of $\TheI{n,k}$ for $k\in\Lambda_n$ forms a complete set of non-isomorphic irreducible of $\tl n(\beta)$ (except for $\beta=0$ with $n$ even for which the statement holds for the set $\Lambda_n\setminus\{0\}$). The projective covers of the irreducible $\TheI{n,k}$ are denoted by $\TheP{n,k}$. If $k$ is critical, then $\TheS{n,k}=\TheI{n,k}=\TheP{n,k}$. If $k$ is non-critical, then the sequence
\begin{equation}\label{eq:sesTheP}0\longrightarrow \TheS{n,k^-}\longrightarrow \TheP{n,k}\longrightarrow \TheS{n,k}\longrightarrow 0\end{equation}
is exact. (If $k=k_L$ and thus $k^-\not\in\Lambda_n$, the $\TheS{n,k^-}$ is understood to be $0$.) The number of compositions factors of the $\TheS{n,k}$ is either one or two and the number for the projective is either one, two, three or four.
 
For $k$ non-critical, the $\Hom$-groups among these three classes of modules (irreducible, standard and projective) are given by the following table:
\begin{center}
\begin{tabular}{cc|ccc}
\multicolumn{2}{c|}{\multirow{2}{*}{$\Hom({\mathsf M},{\mathsf N})$}} & \multicolumn{3}{c}{${\mathsf N}$} \\
     &  & \ \qquad $\TheI{n,k'}$ \qquad \ & \ \qquad $\TheS{n,k'}$ \qquad \ & \ \qquad $\TheP{n,k'}$ \qquad \ \\
\hline
\multicolumn{1}{c}{\multirow{3}{*}{${\mathsf M}$}} & $\TheI{n,k}$ & $\delta_{k',k}\mathbb C{}$ & $\delta_{k',k^-}\mathbb C{}$ & $\delta_{k',k}\mathbb C{}$ \\
\multicolumn{1}{c}{}                           & $\TheS{n,k}$ & $\delta_{k',k}\mathbb C{}$ & $({\delta_{k',k}+\delta_{k',k^-}})\mathbb C{}$ & $({\delta_{k',k}+\delta_{k',k^+}})\mathbb C{}$ \\
\multicolumn{1}{c}{}                           & $\TheP{n,k}$ & $\delta_{k',k}\mathbb C{}$ & $
({\delta_{k',k}+\delta_{k',k^-}})\mathbb C{}$ & $({2\:\delta_{k',k}+\delta_{k',k^-}+\delta_{k',k^+}})\mathbb C{}$\ .
\end{tabular}
\end{center}
with the following conventions. The $\Hom$-groups involving $\TheS{n,k_R}$ are to be read in the row and column of $\TheI{n,k_R}(=\TheS{n,k_R})$, and those with $\TheP{n,k_L}$ in the row and column of $\TheS{n,k_L}(=\TheP{n,k_L})$. Finally, if $\beta=0$ and $n$ is even, the case $k=0$ for $\TheS{n,k}$ is to be read in $\TheI{n,2}(=\TheS{n,0})$ and, again, the module $\TheI{n,0}$ is then trivial.
\end{subsection}
%
%
\begin{subsection}{Modules over $\atl n$}\label{sec:affine.modules}

The definition of the basic modules over $\atl n(\beta)$ follows similar lines. Again $n$ and $k$ are non-negative integers of identical parity. Let $\TheH{n,k}$ be the left $\atl n$-module $\Hom_{\catl}(k,n)$. Again the non-monic diagrams in $\TheH{n,k}$ span a submodule and the {\em standard module} $\TheW{n,k}$ is the quotient of $\TheH{n,k}$ by these non-monic diagrams. If $k>n$, then $\TheW{n,k} = 0$, since a monic diagram in $\TheH{n,k}$ would need to have more than $n$ through lines, which is impossible. Again, we identify monic diagrams with their equivalence classes in $\TheW{n,k}$. Contrarily to the $\TheS{n,k}$ of $\tl n$, the standard modules $\TheW{n,k}$ of $\atl n$ are not finite-dimensional. For example, a basis of $\TheW{4,2}$ is given by the set
\begin{align*}
& 
\begin{tikzpicture}[scale = 1/3]
	\draw[line width = 2pt] (0,-1) -- (0,4);
	\draw[line width = 2pt] (3,-1) -- (3,4);
	\draw[dashed] (0,-1) -- (3,-1);
	\draw[dashed] (0,4) -- (3,4);
	\node[anchor = east] at (0,1.5) {$\Big\{\dots,\ \  v_1\tau=\ $};
	\foreach \s in {1,...,4}
	{	
		\filldraw[black] (0,\s -1) circle (5pt);
	};
	\filldraw[black] (3,3) circle (5pt);
	\filldraw[black] (3,2) circle (5pt);
	\draw (0,2) .. controls (1,2) and (1,3) .. (0,3);%
	\draw (0,0) .. controls (1,0) and (2,3) .. (3,3);%
	\draw (2,-1) .. controls (2,0) and (2,2) .. (3,2); %
	\draw (0,1) .. controls (1,1) and (3/2,5/2) .. (3/2,4);%
	\node[anchor = west] at (3, 1.0) {$,$};
	\end{tikzpicture} 
\begin{tikzpicture}[scale = 1/3]
	\draw[line width = 2pt] (0,-1) -- (0,4);
	\draw[line width = 2pt] (3,-1) -- (3,4);
	\draw[dashed] (0,-1) -- (3,-1);
	\draw[dashed] (0,4) -- (3,4);
	\node[anchor = east] at (0,1.25) {$v_1=\ $};
	\foreach \s in {1,...,4}
	{	
		\filldraw[black] (0,\s -1) circle (5pt);
	};
	\filldraw[black] (3,3) circle (5pt);
	\filldraw[black] (3,2) circle (5pt);
	\draw (0,2) .. controls (1,2) and (1,3) .. (0,3);
	\draw (0,0) .. controls (1,0) and (2,2) .. (3,2);
	\draw (0,1) .. controls (1,1) and (2,3) .. (3,3);
	\node[anchor = west] at (3, 3/2) {$,$};
\end{tikzpicture} 
\begin{tikzpicture}[scale = 1/3]
	\draw[line width = 2pt] (0,-1) -- (0,4);
	\draw[line width = 2pt] (3,-1) -- (3,4);
	\draw[dashed] (0,-1) -- (3,-1);
	\draw[dashed] (0,4) -- (3,4);
	\node[anchor = east] at (0,1.5) {$v_1\tau^{-1}=\ $};
	\foreach \s in {1,...,4}
	{	
		\filldraw[black] (0,\s -1) circle (5pt);
	};
	\filldraw[black] (3,3) circle (5pt);
	\filldraw[black] (3,2) circle (5pt);
	\draw (0,3) .. controls (1,3) and (1,2) .. (0,2); %
	\draw (3,3) .. controls (5/2,3) and (2,7/2) .. (2,4);
	\draw (0,0) .. controls (1/2,0) and (1,-1/2) .. (1,-1);
	\draw (0,1) .. controls (1,1) and (2,2) .. (3,2); %
	\node[anchor = west] at (3, 1.2) {$,\dots ,$};
\end{tikzpicture}  
\begin{tikzpicture}[scale = 1/3]
	\draw[line width = 2pt] (0,-1) -- (0,4);
	\draw[line width = 2pt] (3,-1) -- (3,4);
	\draw[dashed] (0,-1) -- (3,-1);
	\draw[dashed] (0,4) -- (3,4);
	\node[anchor = east] at (0,1.25)  {$ v_2\tau=\ $};
	\foreach \s in {1,...,4}
	{	
		\filldraw[black] (0,\s -1) circle (5pt);
	};
	\filldraw[black] (3,3) circle (5pt);
	\filldraw[black] (3,2) circle (5pt);
	\draw (0,1) .. controls (1,1) and (1,2) .. (0,2);
	\draw (0,0) .. controls (1,0) and (2,3) .. (3,3);
	\draw (0,3) .. controls (1/2,3) and (3/2, 7/2) .. (3/2,4);
	\draw (3/2,-1) .. controls (3/2,-1/2) and (2,2) .. (3,2);
	\node[anchor = west] at (3, 3/2) {$,$};
\end{tikzpicture}  
\begin{tikzpicture}[scale = 1/3]
	\draw[line width = 2pt] (0,-1) -- (0,4);
	\draw[line width = 2pt] (3,-1) -- (3,4);
	\draw[dashed] (0,-1) -- (3,-1);
	\draw[dashed] (0,4) -- (3,4);
	\node[anchor = east] at (0,1.25) {$v_2=\ $};
	\foreach \s in {1,...,4}
	{	
		\filldraw[black] (0,\s -1) circle (5pt);
	};
	\filldraw[black] (3,3) circle (5pt);
	\filldraw[black] (3,2) circle (5pt);
	\draw (0,2) .. controls (1,2) and (1,1) .. (0,1);
	\draw (0,3) -- (3,3);
	\draw (0,0) .. controls (1,0) and (2,2) .. (3,2);
	\node[anchor = west] at (3, 3/2) {$,$};
\end{tikzpicture}  
\begin{tikzpicture}[scale = 1/3]
	\draw[line width = 2pt] (0,-1) -- (0,4);
	\draw[line width = 2pt] (3,-1) -- (3,4);
	\draw[dashed] (0,-1) -- (3,-1);
	\draw[dashed] (0,4) -- (3,4);
	\node[anchor = east] at (0,1.25) {$v_2\tau^{-1}=\ $};
	\foreach \s in {1,...,4}
	{	
		\filldraw[black] (0,\s -1) circle (5pt);
	};
	\filldraw[black] (3,3) circle (5pt);
	\filldraw[black] (3,2) circle (5pt);
	\draw (0,0) .. controls (1/2,0) and (1,-1/2) .. (1,-1);%
	\draw (0,2) .. controls (1,2) and (1,1) .. (0,1);%
	\draw (0,3) .. controls (1,3) and (2,2) .. (3,2);%
	\draw (3,3) .. controls (5/2,3) and (2,7/2) .. (2,4);%
	\node[anchor = west] at (3, 3/2) {$,\dots,$};
\end{tikzpicture}\\  
& 
\begin{tikzpicture}[scale = 1/3]
	\draw[line width = 2pt] (0,-1) -- (0,4);
	\draw[line width = 2pt] (3,-1) -- (3,4);
	\draw[dashed] (0,-1) -- (3,-1);
	\draw[dashed] (0,4) -- (3,4);
	\node[anchor = east] at (0,1.5) {$\ \ \ \dots,\ \ v_3\tau =\ $};
	\foreach \s in {1,...,4}
	{	
		\filldraw[black] (0,\s -1) circle (5pt);
	};
	\filldraw[black] (3,3) circle (5pt);
	\filldraw[black] (3,2) circle (5pt);
	\draw (0,3) .. controls (1/2,3) and (1,7/2) .. (1,4);
	\draw (0,2) .. controls (1,2) and (2,3) .. (3,3); %
	\draw (0,1) .. controls (1,1) and (1,0) .. (0,0); %
	\draw (3/2,-1) .. controls (3/2,-1/2) and (2,2) .. (3,2);
	\node[anchor = west] at (3, 3/2) {$,$};
	\end{tikzpicture}  
\begin{tikzpicture}[scale = 1/3]
	\draw[line width = 2pt] (0,-1) -- (0,4);
	\draw[line width = 2pt] (3,-1) -- (3,4);
	\draw[dashed] (0,-1) -- (3,-1);
	\draw[dashed] (0,4) -- (3,4);
	\node[anchor = east] at (0,1.5) {$v_3=\ $};
	\foreach \s in {1,...,4}
	{	
		\filldraw[black] (0,\s -1) circle (5pt);
	};
	\filldraw[black] (3,3) circle (5pt);
	\filldraw[black] (3,2) circle (5pt);
	\draw (0,3) -- (3,3);
	\draw (0,2) -- (3,2);
	\draw (0,1) .. controls (1,1) and (1,0) .. (0,0);
	\node[anchor = west] at (3, 3/2) {$,$};
\end{tikzpicture}  
\begin{tikzpicture}[scale = 1/3]
	\draw[line width = 2pt] (0,-1) -- (0,4);
	\draw[line width = 2pt] (3,-1) -- (3,4);
	\draw[dashed] (0,-1) -- (3,-1);
	\draw[dashed] (0,4) -- (3,4);
	\node[anchor = east] at (0,1.5) {$v_3\tau^{-1}=\ $};
	\foreach \s in {1,...,4}
	{	
		\filldraw[black] (0,\s -1) circle (5pt);
	};
	\filldraw[black] (3,3) circle (5pt);
	\filldraw[black] (3,2) circle (5pt);
	\draw (3,3) .. controls (5/2,3) and (2,7/2) .. (2,4);%
	\draw (0,3) .. controls (1,3) and (2,2) .. (3,2); %
	\draw (0,1) .. controls (1,1) and (1,0) .. (0,0); %
	\draw (0,2) .. controls (1,2) and (3/2,1/2) .. (3/2,-1);%
	\node[anchor = west] at (3, 3/2) {$,\dots,$};
	\end{tikzpicture} 
\begin{tikzpicture}[scale = 1/3]
	\draw[line width = 2pt] (0,-1) -- (0,4);
	\draw[line width = 2pt] (3,-1) -- (3,4);
	\draw[dashed] (0,-1) -- (3,-1);
	\draw[dashed] (0,4) -- (3,4);
	\node[anchor = east] at (0,1.5) {$ v_4\tau=\ $};
	\foreach \s in {1,...,4}
	{	
		\filldraw[black] (0,\s -1) circle (5pt);
	};
	\filldraw[black] (3,3) circle (5pt);
	\filldraw[black] (3,2) circle (5pt);
	\draw (0,1) .. controls (1,1) and (2,3) .. (3,3); %
	\draw (0,3) .. controls (1/2,3) and (1,7/2) .. (1,4); %
	\draw (0,2) .. controls (1,2) and (3/2,7/2) .. (3/2,4); %
	\draw (2,-1) .. controls (2,0) and (2,2) .. (3,2); %
	\draw (0,0) .. controls (1/2,0) and (1,-1/2) .. (1,-1);
	\node[anchor = west] at (3, 1.0) {$,$};
\end{tikzpicture}
\begin{tikzpicture}[scale = 1/3]
	\draw[line width = 2pt] (0,-1) -- (0,4);
	\draw[line width = 2pt] (3,-1) -- (3,4);
	\draw[dashed] (0,-1) -- (3,-1);
	\draw[dashed] (0,4) -- (3,4);
	\node[anchor = east] at (0,1.5) {$v_4=\ $};
	\foreach \s in {1,...,4}
	{	
		\filldraw[black] (0,\s -1) circle (5pt);
	};
	\filldraw[black] (3,3) circle (5pt);
	\filldraw[black] (3,2) circle (5pt);
	\draw (0,3) .. controls (1/2,3) and (1,7/2) .. (1,4);
	\draw (0,0) .. controls (1/2,0) and (1,-1/2) .. (1,-1);
	\draw (0,2) .. controls (1,2) and (2,3) .. (3,3);
	\draw (0,1) .. controls (1,1) and (2,2) .. (3,2); 
	\node[anchor = west] at (3, 3/2) {$,$};
\end{tikzpicture} 
\begin{tikzpicture}[scale = 1/3]
	\draw[line width = 2pt] (0,-1) -- (0,4);
	\draw[line width = 2pt] (3,-1) -- (3,4);
	\draw[dashed] (0,-1) -- (3,-1);
	\draw[dashed] (0,4) -- (3,4);
	\node[anchor = east] at (0,1.5) {$v_4\tau^{-1}=\ $};
	\foreach \s in {1,...,4}
	{	
		\filldraw[black] (0,\s -1) circle (5pt);
	};
	\filldraw[black] (3,3) circle (5pt);
	\filldraw[black] (3,2) circle (5pt);
	\draw (0,3) .. controls (1/2,3) and (1,7/2) .. (1,4);
	\draw (3,3) .. controls (5/2,3) and (2,7/2) .. (2,4);
	\draw (0,0) .. controls (1/2,0) and (1,-1/2) .. (1,-1);
	\draw (0,1) .. controls (1,1) and (3/2,1/2) .. (3/2,-1);
	\draw (0,2) -- (3,2);
	\node[anchor = west] at (3, 1.) {$,\ \ \dots\Big\}$};
\end{tikzpicture} 
\end{align*}
A detailed study of the property of these modules for $k =0$ and $1$ was obtained in \cite{QasimiStokman}, but the module's structure for $k \geq 2 $ is still largely unknown.

For an invertible $z\in\ringK$ and a positive $k\in\Lambda_n$, the map $f_z: x \to x \circ(\tau - z\,\id)$ defines an endomorphism of $\TheW{n,k}$. The \emph{cell module} $\TheW{n,k;z}$ is defined as the cokernel of this morphism. For $k=0$ and again $z$ an invertible element in $\ringK$, the endomorphism $g$ of $\TheW{n,0}$ is the map adding a non-contractible loop to any element; the cell module $\TheW{n,0;z}$ is then the cokernel of $g - (z+z^{-1})\id $. The dimension of the cell module is
\begin{equation}\label{eq:dimWnkz}\dim\TheW{n,k;z}=\begin{pmatrix}n\\ (n-k)/2\end{pmatrix}.\end{equation}

To describe the structure of the cell modules $\TheW{n,k;z}$, we follow Graham and Lehrer \cite{GLaffine}. Let $\paires$ the set of pairs $(k,z)$ with $k$ a non-negative integer, $z$ an invertible element in $\ringK$, quotiented by the identification $(0,z)\sim(0,z^{-1})$. Let also 
$$\paires_n=\{(k,z)\in \paires\ |\ k\leq n\textrm{ and }k\equiv n\Mod 2\}$$
except if $q+q^{-1}=0$ and $n$ is a non-zero even integer, in which case the pair $(0,q)$ is removed. On standard modules, the \emph{Gram form} is the bilinear form 
\begin{equation}
	\langle - , - \rangle_{n,k}^{\mathsf W}: \TheW{n,k}\times \TheW{n,k} \to \TheW{k,k}, \qquad \langle x , y \rangle_{n,k}^{\mathsf W} = y^{t}\circ x,
\end{equation}
where $x^{t}$ is the involution of $x$ naturally obtained from the involution on morphisms. Note that for all $f \in \atl{n}$, 
\begin{equation}
	\langle f x , y \rangle_{n,k}^{\mathsf W} = y^{t} \circ f x = (f^{t}y)^{t}x = \langle x , f^{t}y \rangle_{n,k}^{\mathsf W},
\end{equation}
so that the radical $\TheR{n,k}$ of this form is a submodule. Because $\tau^t=\tau^{-1}$, the Gram form pairs distinct cell modules:
\begin{equation}
	\langle - , - \rangle_{n,k;z}^{\mathsf W}: \TheW{n,k;z}\times \TheW{n,k;z^{-1}} \to \ringK[z], \qquad \langle x , y \rangle_{n,k;z}^{\mathsf W} = y^{t}\circ x.
\end{equation}
The radical $\TheR{n,k;z}$ in $\TheW{n,k;z}$ of this form also defines a submodule. Let $\TheL{n,k;z}$ be the quotient of $\TheW{n,k;z}/\TheR{n,k;z}$.
%
%
\begin{Thm}[Thm.~(2.8) of \cite{GLaffine}]Let $K=\mathbb C$, $n$ a non-negative integer and $q\in\mathbb C^\times$.
\begin{itemize}
\item[(1)] For $(k,z)$ and $(j,y)\in\paires$, $\mathsf N$ a $\atl n$-submodules of $\TheW{n,k;z}$ and $f:\TheW{n,j;y}\rightarrow\TheW{n,k;z}/\mathsf N$ a non-zero morphism. Then $j\geq k$. Moreover, if $j=k$, then $\Hom(\TheW{n,j;y},\TheW{n,k;z}/\mathsf N)\simeq \mathbb C$.
\item[(2)] The radical of $\TheW{n,k;z}$ is the radical of the form $\langle - , - \rangle_{n,k;z}^{W}$.
\item[(3)] The family of $\TheL{n,k;z}$ indexed by $(k,z)\in\paires_n\,$ form a complete set of non-isomorphic irreducible $\atl n$-modules.
\end{itemize}
\end{Thm}
\noindent Let $q\in\mathbb C^\times$ and $(k,z)$ and $(j,y)\in\paires$. The pair $(j,y)$ is said to succeed $(k,z)$ if there exists a non-negative integer $m$ such that $j=k+2m$ and either
\begin{equation}\label{eq:AandB}\textrm{(A)}\ \ z^2=(-q)^j\textrm{\ \ and\ \ }y=z(-q)^{-m}\qquad\textrm{or}\qquad
\textrm{(B)}\ \ z^2=(-q)^{-j}\textrm{\ \ and\ \ }y=z(-q)^{m}.\end{equation}
Let $\preceq$ be the weakest partial order on $\paires$ that contains $(k,z)\preceq (j,y)$ for all successors $(j,y)$ of $(k,z)$.
\begin{Thm}[(3.4) and (5.1) of \cite{GLaffine}]\label{thm:GL34}
Let $k,j,n$ be integers such that $0\leq k<j\leq n$ and $q\in\mathbb C^\times$.
\begin{itemize}
\item[(1)] $\Hom(\TheW{n,j;y},\TheW{n,k;z})\simeq \mathbb C$ if $(k,z)\preceq (j,y)$ and $\simeq 0$ otherwise.
\item[(2)] The multiplicity of the irreducible module $\TheL{n,j;y}$ in $\TheW{n,k;z}$ is one if $(k,z)\preceq (j,y)$ and zero otherwise.
\end{itemize}
\end{Thm}
\noindent Graham and Lehrer's results are even stronger than what is stated here as they construct explicitly the (injective) morphism $\TheW{n,j;y}\to\TheW{n,k;z}$ for $(j,y)$ an immediate successor of $(k,z)$.
It is not too hard to show that, when $q$ is not a root of unity, there is at most one element $(j,y)$ that succeeds (strictly) the element $(k,z)$. If $q$ is a root of unity and there is an element succeeding the pair $(k,z)$, then there are infinitely many of them in $\paires$. Because of the requirement that $m$ be non-negative in conditions (A) and (B), a (strict) successor $(j,y)$ of $(k,z)$ has always $j>k$. The numbers of successors of $(k,z)$ in $\paires_n$ is thus always finite.

Generically the solutions of the conditions (A) and (B) in \eqref{eq:AandB} are distinct and it may seem that the number of successors may increase exponentially: there are $2$ successors of $(k,z)$, there are $4$ successors of these successors, and so on. However one can show that there are at most two distinct successors of the successors. Appendix \ref{app:wpo} proves this in full generality and computes the dimensions of the irreducible $\atl{}$-modules. However, for the task of computing fusion, we shall need an explicit form of these successors only when $z$ is a (half-)integer power of $-q$. So let
$$z_k=(-q)^{k/2}.$$
The next result describes the weakest partial order when $z$ is equal to $z_k$ for some $k$.

\begin{Prop}\label{prop:weakest}Let $q$ be a root of unity other than $\pm 1$ and $\ell\geq 2$ be the smallest positive integer such that $q^{2\ell}=1$. If $k$ is non-critical, then the successors of $(k_0,u_0)=(k,z_{k+2})$ in the weakest partial order $\preceq$ are organized as 
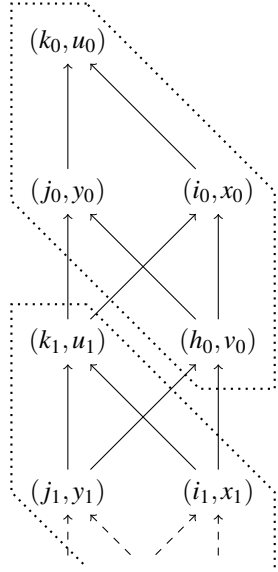
\begin{equation*}
\begin{tikzpicture}[baseline={(current bounding box.center)},scale=0.5]
\node (k0) at (0,12) [] {$(k_0,u_0)$};
\node (j0) at (0,8) [] {$(j_0,y_0)$};
\node (k1) at (0,4) [] {$(k_1,u_1)$};
\node (j1) at (0,0) [] {$(j_1,y_1)$};
\node (i0) at (4,8) [] {$(i_0,x_0)$};
\node (h0) at (4,4) [] {$(h_0,v_0)$};
\node (i1) at (4,0) [] {$(i_1,x_1)$};
\node (k2) at (0,-2) [] {$\ $};
\node (h1) at (4,-2) [] {$\ $};
\node (mid) at (2,-2) [] {$\ $};
\draw[<-] (k0) -- (j0);\draw[<-] (k0) -- (i0);
\draw[<-] (j0) -- (k1);\draw[<-] (j0) -- (h0);
\draw[<-] (i0) -- (k1);\draw[<-] (i0) -- (h0);
\draw[<-] (k1) -- (j1);\draw[<-] (k1) -- (i1);
\draw[<-] (h0) -- (j1);\draw[<-] (h0) -- (i1);
\draw[dashed,<-] (j1) -- (k2);\draw[dashed,<-] (j1) -- (mid);
\draw[dashed,<-] (i1) -- (mid);\draw[dashed,<-] (i1) -- (h1);
\draw[thick,dotted] (0.5,13) -- (-1.5,13) -- (-1.5,7.75) -- (3.5,2.75) -- (5.5,2.75) -- (5.5,8) -- (0.5,13);
\draw[thick,dotted] (0.5,5) -- (-1.5,5) -- (-1.5,-0.0) -- (0.5,-2.);
\draw[thick,dotted] (5.5,-2) -- (5.5,0) -- (0.5,5);\end{tikzpicture}
\end{equation*}
where $(k,z)\leftarrow(j,y)$ stands for $(k,z)\prec(j,y)$, a vertical arrow means that the pair joined solves the condition (A) and a diagonal one the condition (B) of \eqref{eq:AandB}. The pairs appearing in the partial order are
\begin{equation*}\left.
\begin{aligned}
(k_a,u_a)&= (k+2\ell a,z_{k+2+2\ell a})\\
(j_a,y_a)&=(k+2+2\ell a,z_{k+2\ell a})\\
(i_a,x_a)&=(k^++2\ell a,z_{k^++2+2\ell a})\\
(h_a,v_a)&=(k^++2+2\ell a,z_{k^++2\ell a})
\end{aligned}\right\}\quad \textrm{for }a\geq 0.
\end{equation*}

If $k$ is critical, then the partial order is
\begin{equation*}
\begin{tikzpicture}[baseline={(current bounding box.center)},scale=0.5]
\node (k0) at (0,0) [] {$(k_0,u_0)$};
\node (g0) at (4,0) [] {$(j_0,y_0)$};
\node (k1) at (8,0) [] {$(k_1,u_1)$};
\node (g1) at (12,0) [] {$(j_1,y_1)$};
\node (k2) at (15,0) [] {$\ $};
\draw[->] (k0) -- (g0);\draw[->] (g0) -- (k1);
\draw[->] (k1) -- (g1);\draw[dashed,->] (g1) -- (k2);
\end{tikzpicture}
\end{equation*}
with the pairs $(k_a,u_a)$ and $(j_a,y_a)$ above.
\end{Prop}
\begin{proof}The result follows from direct computations of which we give two examples. First the pairs $(j_a,y_a)$ are obtained by solving condition (A) starting with $(k_0,u_0)=(k,z_{k+2})$. This condition is $u_0^2=(-q)^{k+2}=(-q)^{j}$. The smallest $j_0$ is then $j_0=k+2$ with $y_0=z_{k+2}(-q)^{-1}=z_k$. It is also clear that, since $q^{2\ell}=1$, the general solution is $j_a=k+2+\ell a$ and $y_a=z_{k+2\ell a}$ with $a\in\mathbb N_{\geq 0}$.

Second, if $k$ is non-critical, there are (unique) integers $o$ and $p$ such that $k=o\ell -1-p$ with $1\leq p\leq \ell-1$. Then the condition (B) gives the pair $(i_a,x_a)$. The integer $i$ (or $i_a$) must solve $u_0^2=(-q)^{-i}$, that is
$$(-q)^i=(-q)^{-k-2}=(-q)^{-ol+1+p-2}=(-q)^{k^+-2o\ell}$$
where $k^+=o\ell-1+p$ is the reflection of $k$ through the first critical line on its right. (See section \ref{sec:modulesOfTln}.) Hence the general solution is $i_a=k^++2\ell a$ for $a\geq 0$ with
$$x_a=u_0(-q)^{(k^++2\ell a-k)/2}=(-q)^{k/2+1+(k^++2\ell a-k)/2}=z_{k^++2+2\ell a}.$$
The family $(h_a,v_a)$ and the weakest partial order when $k$ is critical are obtained similarly.
\end{proof}

This section ends with the description of the $\Hom$-groups between the families $\TheH{n,k}$ and $\TheW{n,k}$ of modules.
\begin{Prop}\label{prop:hom.h.and.w}
For all $n\geq  k,k' \geq 0$, the following isomorphims hold as $(\atl{k},\atl{k'})$-bimodules:
\begin{align}\label{eq:hom.theh}
		\Hom_{\atl{n}}\left( \TheH{n,k}, \TheH{n,k'}\right) & \simeq \TheH{k,k'},\\
	\label{eq:hom.theh.to.thew}
		\Hom_{\atl{n}}\left( \TheH{n,k}, \TheW{n,k'}\right) & \simeq \TheW{k,k'},\\
	\label{eq:hom.thew.to.thew.1}
		\Hom_{\atl{n}}\left( \TheW{n,k}, \TheW{n,k'}\right)& \simeq \Hom_{\atl{k}}\left( \TheW{k,k}, \TheW{k,k'}\right).
\end{align}
\end{Prop}
\begin{proof}
	We detail the proof of \eqref{eq:hom.theh} as the proof of \eqref{eq:hom.theh.to.thew} is similar. Yoneda's lemma gives $\TheH{k,k'} \simeq \mathsf{Nat}\left( \TheH{-,k}, \TheH{-,k'}\right)$, which gives a natural embedding $G: \TheH{k,k'} \to \Hom_{\atl{n}}\left( \TheH{n,k}, \TheH{n,k'}\right)$. Explicitly, for $b \in\TheH{k,k'}$, $G(b)(v) = v \circ b$ for any $v \in \TheH{n,k}$. If $\beta \neq 0 $ or if $k \neq 0 $, the proof consists in constructing an inverse for $G$. (The case $\beta=0$ and $k=0$ adopts another strategy and is treated at the end of the proof.) The existence of this inverse rests upon two observations:
	\begin{enumerate}
		\item $\TheH{n,k}$ is cyclic, so there exists $x\in\TheH{n,k}$ such that $\TheH{n,k} = \atl{n} x$;
		\item this generator $x$ can be chosen so that there exists $y \in \TheH{k,n}$ such that $y\circ x = \id_{\atl{k}}$.
	\end{enumerate}
The generator $x=x_{n,k}\in\TheH{n,k}$ can be chosen as having $(n-k)/2$ non-nested arcs on its left side followed by $k$ through lines. If $k>0$, the corresponding $y=y_{n,k}$ is then obtained by first taking the involution $x^t$ of $x$ and then moving the first through line of $x$ to the top position. For instance, if $k=2, n=6$, theses choices are
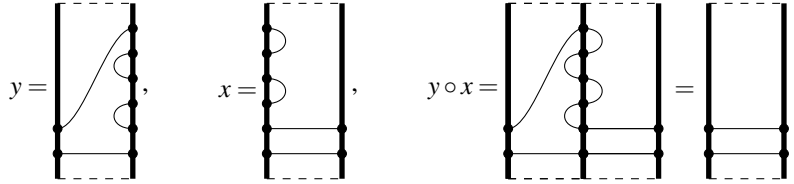
\begin{equation*}
\begin{tikzpicture}[scale = 1/3]
	\draw[line width = 2pt] (0,-1) -- (0,6);
	\draw[line width = 2pt] (3,-1) -- (3,6);
	\draw[dashed] (0,-1) -- (3,-1);
	\draw[dashed] (0,6) -- (3,6);
	\foreach \s in {1,...,6}
	{	
		\filldraw[black] (3,\s -1) circle (5pt);
	};
	\filldraw[black] (0,1) circle (5pt);
	\filldraw[black] (0,0) circle (5pt);
	\draw (0,0) -- (3,0);
	\draw (3,5) .. controls (2,5) and (1,1) ..  (0,1);
	\draw (3,3) .. controls (2,3) and (2,4) .. (3,4);
	\draw (3,2) .. controls (2,2) and (2,1) .. (3,1);
	\node[anchor = west] at (3, 5/2) {,};
	\node[anchor = east] at (0,5/2) {$y =$ };
\end{tikzpicture}\qquad
\begin{tikzpicture}[scale = 1/3]
	\draw[line width = 2pt] (0,-1) -- (0,6);
	\draw[line width = 2pt] (3,-1) -- (3,6);
	\draw[dashed] (0,-1) -- (3,-1);
	\draw[dashed] (0,6) -- (3,6);
	\foreach \s in {1,...,6}
	{	
		\filldraw[black] (0,\s -1) circle (5pt);
	};
	\filldraw[black] (3,1) circle (5pt);
	\filldraw[black] (3,0) circle (5pt);
	\draw (0,0) -- (3,0);
	\draw (0,1) -- (3,1);
	\draw (0,3) .. controls (1,3) and (1,2) .. (0,2);
	\draw (0,5) .. controls (1,5) and (1,4) .. (0,4);
	\node[anchor = west] at (3,5/2) {,};
	\node[anchor = east] at (0,5/2) {$x =$ };
\end{tikzpicture}\qquad
\begin{tikzpicture}[scale = 1/3]
	\draw[line width = 2pt] (0,-1) -- (0,6);
	\draw[line width = 2pt] (3,-1) -- (3,6);
	\draw[dashed] (0,-1) -- (3,-1);
	\draw[dashed] (0,6) -- (3,6);
	\foreach \s in {1,...,6}
	{	
		\filldraw[black] (3,\s -1) circle (5pt);
	};
	\filldraw[black] (0,1) circle (5pt);
	\filldraw[black] (0,0) circle (5pt);
	\draw (0,0) -- (3,0);
	\draw (3,5) .. controls (2,5) and (1,1) ..  (0,1);
	\draw (3,3) .. controls (2,3) and (2,4) .. (3,4);
	\draw (3,2) .. controls (2,2) and (2,1) .. (3,1);
	\node[anchor = east] at (0,5/2) {$y\circ x =$ };
	\draw[line width = 2pt] (3,-1) -- (3,6);
	\draw[line width = 2pt] (6,-1) -- (6,6);
	\draw[line width = 2pt] (0,-1) -- (0,6);
	\draw[line width = 2pt] (3,-1) -- (3,6);
	\draw[dashed] (0,-1) -- (3,-1);
	\draw[dashed] (0,6) -- (3,6);
	\filldraw[black] (6,1) circle (5pt);
	\filldraw[black] (6,0) circle (5pt);
	\draw (6,0) -- (3,0);
	\draw (6,1) -- (3,1);
	\draw (3,3) .. controls (4,3) and (4,2) .. (3,2);
	\draw (3,5) .. controls (4,5) and (4,4) .. (3,4);
	\draw (3,0) -- (6,0);
	\draw (3,1) -- (6,1);
	\draw[dashed] (3,-1) -- (6,-1);
	\draw[dashed] (3,6) -- (6,6);
	\foreach \s in {1,...,6}
	{	
		\filldraw[black] (3,\s -1) circle (5pt);
	};
	\node[anchor = west] at (6, 5/2) {$\ =$};
	\draw[line width = 2pt] (8,-1) -- (8,6);
	\draw[line width = 2pt] (11,-1) -- (11,6);
	\draw[dashed] (8,-1) -- (11,-1);
	\draw[dashed] (8,6) -- (11,6);
	\foreach \s in {1,2}
	{	
		\filldraw[black] (8,\s -1) circle (5pt);
		\filldraw[black] (11,\s -1) circle (5pt);
	};
	\draw (8,0) -- (11,0);
	\draw (8,1) -- (11,1);
	\node[anchor = west] at (11, 5/2) { . };
\end{tikzpicture}
\end{equation*}
If $k=0$ but $\beta\neq 0$, the corresponding $y$ is
$$
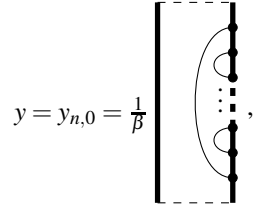
\begin{tikzpicture}[scale = 1/3]
	\draw[line width = 2pt] (0,-1) -- (0,7);
	\draw[line width = 2pt] (3,-1) -- (3,2);
	\draw[line width = 2pt] (3,4) -- (3,7);
	\draw[line width = 2pt,dashed] (3,2) -- (3,4);
	\node[anchor = west] at (2, 3.4) {$\vdots$};
	\draw[dashed] (0,-1) -- (3,-1);
	\draw[dashed] (0,7) -- (3,7);
	\foreach \s in {1,...,3}
	{	
		\filldraw[black] (3,\s -1) circle (5pt);
	};
	\foreach \s in {5,...,7}
	{	
		\filldraw[black] (3,\s -1) circle (5pt);
	};
	\draw (3,6) .. controls (1,6) and (1,0) ..  (3,0);
	\draw (3,4) .. controls (2,4) and (2,5) .. (3,5);
	\draw (3,2) .. controls (2,2) and (2,1) .. (3,1);
	\node[anchor = west] at (3.2, 5/2) {,};
	\node[anchor = east] at (0,5/2) {$y =y_{n,0}=\frac1\beta$ };
\end{tikzpicture}
$$
so that, indeed, $y_{n,0}\circ x_{n,0} = \id_{\atl{0}}$ which is the empty link diagram.
Define then $F:\Hom_{\atl{n}}\left( \TheH{n,k}, \TheH{n,k'}\right) \to \TheH{k,k'}$, by $F(f)= y \circ f(x)$, for a morphism $f \in \Hom_{\atl{n}}\left( \TheH{n,k}, \TheH{n,k'}\right)$. Then, for any $v\in\TheH{n,k}$ and $w\in\TheH{k,k'}$: 
\begin{equation}
	((G\circ F)(f))(v) = v \circ F(f) = \underbrace{v \circ y}_{\in \atl{n}} \circ f(x) = f(v\circ \underbrace{y\circ x}_{\id_{\atl{k}}}) = f(v),
\end{equation}
and
\begin{equation}
	(F\circ G)(w) = F((v \to v\circ w)) = \underbrace{y\circ x}_{\id_{\atl{k}}} \circ w = w.
\end{equation}
It thus follows that $G$ and $F$ are isomorphisms. Note in particular that, since $G$ is surjective, it follows readily that $F$ is a morphism of bimodules: for all $a \in  \atl{k}$, $c \in  \atl{k'}$, $b\in \TheH{k,k'}$,
\begin{equation}
	F(aG(b)c) = y\circ (aG(b)c)(x) = y\circ x\circ a\circ b \circ c = a\circ b \circ c= a F(G(b))c.
\end{equation}

To obtain equation \eqref{eq:hom.thew.to.thew.1}, note that
\begin{equation}
	\TheW{n,k} \simeq \Hom_{\catl}(k,n)\otimes_{\atl{k}}\TheW{k,k},
\end{equation}
and use the adjunction theorem together with \eqref{eq:hom.theh.to.thew}.

The case where both $k$ and $\beta$ are zero remains. For this, we instead prove that $G$ is surjective. Choose $x = x_{n,0}$ as in the generic case, but now pick $y = x_{n,2} \circ y_{n,2} \in \atl{n}$. One can verify directly that $y \circ x = x$, and thus for all $f \in \Hom_{\atl{n}}\left( \TheH{n,0}, \TheH{n,k'}\right) $, $y f(x) = f(y \circ x) = f(x)$. In particular this implies that 
\begin{equation*}
	f(x) = x_{n,2} \circ g,
\end{equation*}
for some $g \in \TheH{2,k'}$. However, $\tau^{-2}x = x$ impose the further condition $\tau^{-2} x_{n,2} \circ g = x_{n,2} \circ g$, but since $\tau^{-2} x_{n,2} \neq x_{n,2}$ it follows that $g$ cannot contain any monic diagrams. In other words, $g$ can be written as a sum of elements of the form $g' \circ h $ for some $g' \in \TheH{2,0} $, $h \in \TheH{0,k'}$ and, thus, $g'$ can only be 
\begin{tikzpicture}[scale = 1/10]
	\draw[line width = 0.5pt] (0,0) -- (0,3);
	\draw[line width = 0.5pt] (2,0) -- (2,3);
	\foreach \s in {1,...,2}
	{	
		\filldraw[black] (0,\s ) circle (5pt);
	};
	\draw (0,2) .. controls (1.5,2) and (1.5,1) .. (0,1);
\end{tikzpicture}
or
\begin{tikzpicture}[scale=0.10]	
	\draw[line width = 0.5pt] (0,0) -- (0,3);
	\draw[line width = 0.5pt] (2,0) -- (2,3);
	\foreach \s in {0,1}
	{	
		\filldraw[black] (0,\s+1) circle (5pt);
	};
	\draw (0,2) .. controls (1,2) .. (1,3);
	\draw (0,1) .. controls (1,1) .. (1,0);
\end{tikzpicture}
up to constant in $\mathbb C$. However note that $e_{n-1}x = 0$, but $e_{n-1}x_{n,2} \circ g = x_{n,2} \circ e_{1}g$, which rules out the second possibility. Hence $g' = x_{2,0} \circ \alpha$ for some $\alpha \in \TheH{0,0}$, and $f(x) = x \circ \alpha \circ h'$ so 
\begin{equation*}
	f = G(\alpha \circ h').
\end{equation*} 
Since $G$ is both injective and surjective, it is an isomorphism.
\end{proof}
\begin{Coro}\label{coro:hom.h.and.w}
	If $n \geq k'>k$, then
	\begin{equation}
		\Hom_{\atl{n}}\left( \TheH{n,k}, \TheW{n,k'}\right) \simeq \Hom_{\atl{n}}\left( \TheW{n,k}, \TheW{n,k'}\right) \simeq 0,
	\end{equation}
	\begin{equation}
		\End_{\atl{n}}(\TheH{n,k}) \simeq \atl{k}, \qquad \End_{\atl{n}}(\TheW{n,k}) \simeq \ringK[\tau,\tau^{-1}].
	\end{equation}
	Furthermore, $\TheW{n,k}$ is indecomposable.
\end{Coro}
\begin{proof}
	Note that if $k'>k$, $\TheW{k,k'}= 0$, so the first equation follows directly from the previous proposition. The second equation is obtained by putting $k' = k$ in \eqref{eq:hom.theh} and \eqref{eq:hom.thew.to.thew.1}, noting that $\atl{k} \equiv \TheH{k,k}$, and that the set of monic diagrams in $\atl{k}$ is given by the integer powers of $\tau$. Finally, it follows readily from the second equation that the ring of endomorphisms of $\TheW{n,k}$ is an integral domain, so that there are no non-trivial idempotents in $\End_{\atl n}(\TheW{n,k})$.
\end{proof}

Finally the last proposition describes the homomorphisms between the standard modules $\TheW{n,k}$.
\begin{Prop}\label{prop:hom.w.to.w}
	If $n\geq k,k' \geq 0$, 
	then
	\begin{equation}
		\Hom_{\atl{n}}\left( \TheW{n,k}, \TheW{n,k'} \right) = \delta_{k,k'} \mathbb C[\tau ,\tau^{-1}].
	\end{equation}
\end{Prop}
\begin{proof}The cases $k'\geq k$ are covered by corollary \ref{coro:hom.h.and.w} and proposition \ref{prop:hom.h.and.w} has shown that $\Hom_{\atl{n}}\left( \TheW{n,k}, \TheW{n,k'}\right)\simeq \Hom_{\atl{k}}\left( \TheW{k,k}, \TheW{k,k'}\right)$. The rest of the proof assumes $k'<k$. A basis for $\TheW{k,k'}$ can be built from a finite number of $(k,k')$-diagrams $v_i, i\in I$ (with $|I|=\dim \TheW{k,k';z}$), together with their twisted partners $v_i\tau^{j}$, $j\in\mathbb Z\setminus \{0\}$. (The four $(4,2)$-diagrams $v_1,v_2,v_3$ and $v_4$ of the basis for $W_{4,2}$ given above provide an example.) Then a common basis for the modules $\TheW{k,k';z}$ is the (finite) set of equivalence classes $\{[v_i],i\in I\}$. Let $f\in \Hom_{\atl{k}}\left( \TheW{k,k}, \TheW{k,k'} \right)$ and $x \in \TheW{k,k}$. For all $i \in \mathbb{Z}$, $e_i x = 0$ and thus $y^{t}\circ f(x) = 0$ for all $y \in  \TheW{k,k'}$. Let $\mu_z : \TheW{k,k'} \to \TheW{k,k';z} $ be the canonical projection to the cell module. If $z$ is seen as a formal parameter, the element $\mu_z(f(x))$ has coefficients in $\mathbb C[z,z^{-1}]$ in the basis described above (because $\tau^jv_i\mapsto z^jv_i$). Moreover, by the previous argument, $\mu_z(f(x))$ must be in the radical of the Gram form. Proposition \ref{thm:GL34} states that this radical vanishes, except possibly for finitely many $z\in \mathbb C$, and thus all coefficients of $\mu_z(f(x))$ must be identically zero as polynomials of $z$ and $z^{-1}$. It follows that $f(x) = 0$ for all $x$.
\end{proof}
\end{subsection}
\end{section}
%
\begin{section}{Restriction and induction functors}\label{sec:restrictionInduction}
%
In this section the functors $\Indphi{}$, $\Resphi$ and $\Indar{}$ are applied to the ``basic" modules defined in the previous section. The results will be the main tools for computing fusion.
%
%
\begin{subsection}{The functors $\Indphi$ and $\Resphi$}\label{sec:twoFunctors}
 We start with the induction $\Indphi{}$.
%
%
\begin{Prop}\label{prop:indH}
	If $n \geq k$,
		\begin{equation}
			\Indphi{\TheH{n,k}} \simeq \TheM{n,k},\qquad \Indphi{\TheW{n,k}} \simeq \TheS{n,k},
		\end{equation}
\end{Prop}
\begin{proof}
We focus on the result for $\Indphi{\TheH{n,k}} $, as those for $\Indphi{\TheW{n,k}} $ are obtained using very similar arguments. The proof hinges on a crucial fact used in the proof of proposition \ref{prop:hom.h.and.w}: the modules $\TheH{n,k}$ and $\TheW{n,k}$ are both cyclic, and the generator can be chosen as a rank zero diagram which we call $x_{n,k}$. It follows that
	\begin{equation}
		\Indphi{\TheH{n,k}}= \tl{n}\otimes_{\atl{n}} \TheH{n,k} = \lbrace a \otimes_{\atl{n}} x_{n,k}\ |\ a \in \tl{n} \rbrace.
	\end{equation}
We start with cases with $n>2$. 
If we write $z_{n,k}$ for the diagram $x_{n,k}$ seen as an element of $\tl{n}$, then the desired morphism $\Indphi{\TheH{n,k}}\rightarrow \TheM{n,k}$ is obtained by sending $a \otimes_{\atl{n}}x_{n,k} \to a \circ z_{k,n}$. Its surjectivity follows from the cyclicity of $\TheM{n,k}$. Moreover, in every conjugacy class of $\tl{n}\otimes_{\atl{n}} \TheH{n,k}$, there is a single representative of the form $1_{\End{n}}\otimes y$ with $y\in\tl n$. The map is thus injective.
\end{proof}
Computing the induction of the cell modules requires a bit more work, and relies on the braiding $\eta$ introduced in section \ref{sec:tlcategories}.
%
%
\begin{Prop}\label{prop:indphi.thew}
	If $n\geq k$,
		\begin{equation}
			\Indphi{\TheW{n,k;z}} \simeq
			\begin{cases}
				\TheS{n,k}, & \text{if } z = (-q)^{\frac{k+2}{2}},\\
				0, & \text{otherwise}.
			\end{cases}
		\end{equation}

\end{Prop}
\begin{proof}
	Recall the element $x_{n,k}$ introduced in the proof of proposition \ref{prop:hom.h.and.w}:
	\begin{equation*}
\begin{tikzpicture}[scale = 1/3]
	\draw[line width = 2pt] (0,-2) -- (0,8);
	\draw[line width = 2pt] (3,-2) -- (3,8);
	\draw[dashed] (0,-2) -- (3,-2);
	\draw[dashed] (0,8) -- (3,8);
	\foreach \s in {0,2,3,4,5,7,8}
	{	
		\filldraw[black] (0,\s -1) circle (5pt);
	};
	\filldraw[black] (3,-1) circle (5pt);
	\filldraw[black] (3,1) circle (5pt);
	\filldraw[black] (3,2) circle (5pt);
	\draw (0,-1) -- (3,-1);
	\draw (0,1) -- (3,1);
	\draw (0,2) -- (3,2);
	\draw (0,3) .. controls (1,3) and (1,4) .. (0,4);
	\draw (0,6) .. controls (1,6) and (1,7) .. (0,7);
	\node[anchor = west] at (0.5,0) {$\dots$};
	\node[anchor = west] at (0,11/2) {$\vdots$};
	\node[anchor = east] at (-0.5,3) {$x_{n,k} =$ };
	\draw[<->] (4,-1) -- (4,2);
	\node[anchor = west] at (4,0.5) {$k$};
	\node[anchor = west] at (4.5,3) {$ \in \Hom_{\ctl}{(k,n)} \subset \Hom_{\catl}{(k,n)}.$};
\end{tikzpicture}
	\end{equation*}
For $k>0$, on can verify that
	\begin{equation*}
		e_1e_3\hdots e_{n-k-1} \tau(n) x_{n,k} = x_{n,k} \circ \tau(k).
	\end{equation*}
(If $k=n$, then the generators $e_1e_3\hdots e_{n-k-1}$ are simply omitted. This holds also in the following argument.) It follows that in $\TheW{n,k;z}$, 
	\begin{equation*}
		e_1e_3\hdots e_{n-k-1} \tau(n)( x_{n,k} + I) = z x_{n,k} + I,
	\end{equation*}
where $I$ is the kernel of the canonical projection $\TheH{n,k} \to \TheW{n,k;z}$. However, 
	\begin{equation*}
		e_1e_3\hdots e_{n-k-1} \eta_{n-1,1} x_{n,k} = x_{n,k}\circ \eta_{k-1,1},
	\end{equation*}
where we used the fact that $x_{n,k} = (x_{n-1,k-1} \otimes x_{1,1}) \in \Hom_{\ctl}{(k,n)} $ and that, 
$\eta$ being a braiding, $\eta_{n-1,1} x_{n,k}=(x_{1,1}\otimes x_{n-1,k-1})\eta_{k-1,1}$. Since $\eta_{k-1,1} = (-q)^{\frac{k-1}{2}}1_{\tl{k}} + B$, where $B$ is a sum of non-monic diagrams,
	\begin{equation*}
		e_1e_3\hdots e_{n-k-1} \eta_{n-1,1}( x_{n,k} + I) = (-q)^{\frac{k-1}{2}} x_{n,k} + I,
	\end{equation*}
	so that
	\begin{equation}\label{eq:nonMonic}
		\iota \circ \phi(e_1e_3\hdots e_{n-k-1} \tau(n))(x_{n,k} + I) = (-q)^{\frac{k+2}{2}} x_{n,k} + I.
	\end{equation}
We can check directly that $x_{n,k} + I$ is a generator of $\TheW{n,k;z}$, so that every element of $\Indphi{\TheW{n,k;z}}$ can be put in the form $a \otimes_{\atl{n}} (x_{n,k} + I)$, for some $a \in \tl{n}$. Since $z \neq 0$ and $\phi \circ \iota \circ \phi = \phi$,
\begin{align*}	
	a \otimes_{\atl{n}} (x_{n,k} + I) 	& = z^{-1} a \otimes_{\atl{n}} z(x_{n,k} + I)\\
										& = z^{-1}a\otimes_{\atl{n}} e_1e_3\hdots e_{n-k-1}\tau(n) (x_{n,k} + I)\\
										& = z^{-1}a \phi( e_1e_3\hdots e_{n-k-1}\tau(n)) \otimes_{\atl{n}} (x_{n,k} + I)\\
										& = z^{-1} a \phi\circ\iota\circ\phi( e_1e_3\hdots e_{n-k-1}\tau(n)) \otimes_{\atl{n}} (x_{n,k}+I)\\
										& = z^{-1}a  \otimes_{\atl{n}} \iota\circ\phi( e_1e_3\hdots e_{n-k-1}\tau(n))(x_{n,k}+I)\\
										& = z^{-1}(-q)^{\frac{k+2}{2}} a\otimes_{\atl{n}} (x_{n,k} + I).
\end{align*}
We thus conclude that $a \otimes_{\atl{n}} (x_{n,k} + I) = 0$ unless $z = (-q)^{\frac{k+2}{2}}$. The kernel $J$ of the canonical projection $\TheW{n,k}\rightarrow \TheW{n,k;z}$ is $\{a\otimes_{\atl{n}}x_{n,k}(\tau(k)-z)\,|\, a\in\tl n\}$ and the above computation also shows that, if $z= (-q)^{\frac{k+2}{2}}$, then $\Indphi{J}=0$. Since $\Indphi{}$ is right exact, this implies that $\Indphi{\TheW{n,k}}\simeq\Indphi{\TheW{n,k;z}}$.

For $k=0$, the proof is similar except that we now use the fact that $x_{n,0} = x_{n,2}\circ b$, where
\begin{equation*}
	\begin{tikzpicture}[scale = 1/3]
	\draw[line width = 2pt] (0,-1) -- (0,2);
	\draw[line width = 2pt] (3,-1) -- (3,2);
	\draw[dashed] (0,-1) -- (3,-1);
	\draw[dashed] (0,2) -- (3,2);
	\foreach \s in {1,2}
	{	
		\filldraw[black] (0,\s -1) circle (5pt);
	};
	\draw (0,0) .. controls (1,0) and (1,1) .. (0,1);
	\node[anchor = west] at (4,0) {$ \in \Hom_{\ctl}{(0,2)} \subset \Hom_{\catl}{(0,2)},$};
	\node[anchor = east] at (0,0) {$b =$};
	\end{tikzpicture}
\end{equation*}
that $\eta_{1,1} \circ b = -(-q)^{-3/2} b$, and that $e_1e_3\hdots e_{n-1} \tau x_0(n) =  g x_{n,0}$, where $g$ is the endomorphism that adds a non-contractible loop. (See the definition of $\TheW{n,k;z}$ in section \ref{sec:affine.modules}.)
With these, equation \eqref{eq:nonMonic} is replaced by 
$$\iota\circ\phi[e_1e_3\dots e_{n-1}\tau(n)](x_{n,0}+I)=-\beta(x_{n,0}+I),$$
and the result follows.
\end{proof}

\begin{Prop}\label{coro:dim.irred.affine}
	For all $0 \leq k \leq n$,
		\begin{equation}
			\dim \TheL{n,k;z_{k+2}} = \dim \TheI{n,k},
		\end{equation}
	where $z_{r} = (-q)^{r/2}$.
\end{Prop}
\begin{proof}This result follows immediately from the corollary \ref{cor:dimL} of the appendix \ref{app:wpo}. However, we give here a proof that rests only on the previous proposition. 
If $q$ is generic, $$ \TheL{n,k;z_{k+2}} \simeq \TheW{n,k;z_{k+2}}/ \TheW{n,k+2;z_{k}}, $$
	so that
		$$ \dim \TheL{n,k;z_{k+2}} = \binom{n}{\frac{n-k}{2}} - \binom{n}{\frac{n-(k+2)}{2}} = \dim \TheS{n,k}.$$
	If $q$ is critical, then $ \TheW{n,k,z_{k+2}}/ \TheW{n,k+2;z_{k}}$  has two composition factors, $\TheL{n,k;z_{k+2}}$ and $\TheL{n,k^{+};z_{k^{+}+2}}$. Let $k_0\in [k]$ be such that $ k_0^{+} > n$, that is, let $k_0$ be the largest integer in the orbit $[k]$. Then 
$$\dim \TheL{n,k_0;z_{k_0+2}} = \dim \TheW{n,k_0;z_{k_0+2}}/ \TheW{n,k_0+2;z_{k_0}} =  \dim\TheS{n,k_0} = \dim \TheI{n,k_0}.$$
It then follows that 
	\begin{align*}
		\dim \TheL{n,k_0^{-};z_{k_0^{-}+2}} & = \dim \big(\TheW{n,k_0^{-};z_{k_0^{-}+2}}/ \TheW{n,k_0^{-}+2;z_{k_0^{-}}} \big) - \dim \TheL{n,k_0,z_{k_0+2}} \\
			& = \dim \TheS{n,k_0^{-}} - \dim \TheI{n,k_0} \\
			& = \dim \TheI{n,k_0^{-}},
	\end{align*}
where we used the fact that the composition factors of $\TheS{n,r}$ are $\TheI{n,r}$ and $\TheI{n,r^{+}}$. This argument can then be repeated to bring the second index to the desired value $k$ in the orbit $[k]$.
\end{proof}
%
%
This result can be used to get the induction and the restriction of irreducible modules. We leave out the trivial cases where $z \neq z_{k+2}=(-q)^{(k+2)/2}$ because, since induction is right exact, if the induction of $\TheW{n,k;z}$ is zero, so is the induction of all its quotients.
\begin{Prop}\label{prop:indphi.resphi.irred}
	For $0 \leq k \leq n$,
		\begin{equation}\label{eq:resphi.irred}
			\Resphi{\TheI{n,k}} \simeq \TheL{n,k;z_{k+2}}
		\end{equation}
	and 
		\begin{equation}
			 \TheI{n,k} \simeq \Indphi{\TheL{n,k;z_{k+2}}}.
		\end{equation}
\end{Prop}
\begin{proof}
 Proposition \ref{coro:dim.irred.affine} gives
		\begin{equation*}
			\dim \Resphi{\TheI{n,k}} =  \dim \TheL{n,k;z_{k+2}},
		\end{equation*}
since restriction preserves the dimensions of modules. The result will thus follow from the existence of a non-zero morphism between these two modules, since $\TheL{n,k;z_{k+2}}$ is irreducible. Theorem \ref{thm:GL34}, together with the proposition \ref{prop:weakest}, shows that there is an exact sequence
	\begin{equation*}
		 \TheW{n,k+2;z_{k}}  
		 + \TheW{n,k^+;z_{k^++2}}\longrightarrow \TheW{n,k;z_{k+2}} \longrightarrow \TheL{n,k;z_{k+2}} \rightarrow 0,
	\end{equation*}
where, again, $z_{k} = (-q)^{k/2}$, and we assumed that $k$ is non-critial. (If $k$ is critical, the second summand in the leftmost module is simply removed.) 
Proposition \ref{prop:indphi.thew} gives the induction of this sequence. The induced $\Indphi{\TheW{n,k^+;z_{k^++2}}}$ is simply $\TheS{n,k^+}$, but the induction of $\TheW{n,k+2;z_k}$ will be non-zero if and only if $z_k=z_{k+4}$:
	\begin{equation}\label{eq:es.indphi.irred.2}
		\delta_{z_k,z_{k+4}}\TheS{n,k+2}
		+\TheS{n,k^{+}}\overset{f}{\longrightarrow} \TheS{n,k} \longrightarrow  		\Indphi{\TheL{n,k;z_{k+2}}} \rightarrow 0.
	\end{equation}
The condition $z_k = z_{k+4}$ is equivalent to $q^2=1$. However, when $q^2 = 1$, the regular standard modules are all irreducible and non-isomorphic and then $\TheS{n,k+2}$ has to lie in the kernel of $f$. This leaves two possibilities for $\Indphi{\TheL{n,k;z_{k+2}}}$: it is either $\TheS{n,k}$, if $f=0$, or $\TheI{n,k}$ otherwise. Fortunately both $\Hom$-groups $\Hom(\TheS{n,k},\TheI{n,k'})$ and $\Hom(\TheI{n,k},\TheI{n,k'})$ coincide and thus $\Hom(\Indphi{\TheL{n,k;z_{k+2}}},\TheI{n,k'})\simeq\delta_{k,k'}\mathbb C$. Then Frobenius theorem gives
$$\delta_{k,k'}\mathbb C\simeq \Hom(\Indphi{\TheL{n,k;z_{k+2}}},\TheI{n,k'})\simeq \Hom(\TheL{n,k;z_{k+2}}, \Resphi{\TheI{n,k'}})$$
and proves the isomorphism $\TheL{n,k;z_{k+2}}\simeq \Resphi{\TheI{n,k}}$, that is equation \eqref{eq:resphi.irred}. 
	It also follows that $\Indphi(\TheL{n,k;z_{k+2}}) \simeq \Indphi{}\circ\Resphi{\TheI{n,k}} \simeq \TheI{n,k}$ by lemma \ref{lem:functorId}.
\end{proof}
%
%
\begin{Coro}
	For $0 \leq k \leq n$,
		\begin{equation}
			\Resphi{\TheS{n,k}} \simeq \TheW{n,k;z_{k+2}} / \TheW{n,k+2;z_{k}}.			\end{equation}
\end{Coro}
\begin{proof}
	  If either $q$ is not a root of unity or $k$ is critical, $\TheS{n,k} \simeq \TheI{k}(n)$ and the result follows from the previous proposition. If $q$ is a root of unity and $k$ is non-critical, the restriction functor applied to the exact sequence \eqref{eq:sesTheS} of $\tl{n}$-modules
	\begin{equation*}
		0 \rightarrow \TheI{n,k^{+}} \longrightarrow \TheS{n,k} \longrightarrow \TheI{n,k} \rightarrow 0
	\end{equation*}
yields the short exact sequence of $\atl{n}$-modules
	\begin{equation}
		0 \rightarrow \TheL{n,k^{+};z_{k^{+}+2}} \longrightarrow \Resphi{\TheS{n,k}} \longrightarrow \TheL{n,k;z_{k+2}} \rightarrow 0.
	\end{equation}
This sequence cannot split as $\Resphi{\TheS{n,k}}$ is indecomposable by proposition \ref{prop:res.ind.hom}: $\End{\Resphi{\TheS{n,k}}} \simeq \End{\TheS{n,k}} \simeq \ringK$. Furthermore, Frobenius theorem with the previous propositions gives
	\begin{equation}\label{eq:hom.w.resphi.s}
		\Hom(\TheW{n,k;z_{k+2}}, \Resphi{\TheS{n,k}}) \simeq \ringK, \qquad \Hom(\TheW{n,k+2;z_{k}}, \Resphi{\TheS{n,k}}) \simeq 0,
	\end{equation}
	\begin{equation}\label{eq:hom.w.I}
		\Hom(\TheW{n,k;z_{k+2}}, \TheL{n,k^{+};z_{k^{+}+2}}) \simeq 0.
	\end{equation}
It follows from equation \eqref{eq:hom.w.resphi.s} that there exists a non-trivial morphism $f: \TheW{n,k;z_{k+2}} / \TheW{n,k+2;z_{k}} \to \Resphi{\TheS{k}}$. Since $\Resphi{\TheS{n,k}}$ has only one non-trivial submodule, $f$ is either surjective, or its image is isomorphic to $ \TheL{n,k^{+};z_{k^{+}+2}}$. The latter case would contradict equation \eqref{eq:hom.w.I} and $f$ is thus an isomorphism.
\end{proof}
%
%
Now that the restriction of the irreducible and standard modules are known, it remains to compute the restriction of other $\tl n$-modules with three or more composition factors, including the projective ones. However, it can be checked directly that such modules cannot correspond to any of the affine modules discussed so far, simply by comparing the composition factors. We shall thus restrict our dicussion to the restriction of the projective modules $\TheP{n,k}$. Let the \emph{affine module $\ThePaff{n,k}$} be defined as
\begin{equation}
	\ThePaff{n,k} \equiv\ \Resphi{\TheP{n,k}}.
\end{equation}
It is indecomposable and its endomorphism ring is two-dimensional. To see this, suppose that $\ThePaff{n,k}$ is decomposable, say $\ThePaff{n,k} \simeq A \oplus B $. Lemma \ref{lem:functorId} then gives $\TheP{n,k} \simeq \Indphi{A} \oplus \Indphi{B}$, which implies that one of these inductions is the zero module, say $\Indphi{B} = 0 $. Then $\mathbb C^2\simeq \Hom_{\tl{n}}(\TheP{n,k},\TheP{n,k})\simeq \Hom{\tl{n}}(\Indphi{A},\TheP{n,k})\simeq \Hom_{\atl{n}}(A, \ThePaff{n,k})$. It thus follows that $\End{\ThePaff{n,k}} \simeq \mathbb{C}^{2} \simeq \mathbb{C}^{2} \oplus \Hom_{\atl{n}}(B, A\oplus B)$, and thus that $B=0$. 
Finally the restriction of \eqref{eq:sesTheP} leads to the short exact sequence
\begin{equation}
	0 \rightarrow \TheW{n,k^{-};z_{k^{-}+2}} / \TheW{n,k^{-}+2;z_{k^{-}}} 
	\longrightarrow \ThePaff{n,k} \longrightarrow \TheW{n,k;z_{k+2}} / \TheW{n,k+2;z_{k}} \rightarrow 0.
\end{equation}
Figure \ref{fig:one} presents the structure of the affine modules $\Resphi{\TheS{n,k}}$ and $\Resphi{\TheP{n,k}}$ in terms of their composition factors.
\begin{figure}[h]
	\begin{tikzpicture}[scale = 1/3]
		\node (h) at (0,4.5) {$ \TheL{n,k;z_{k+2}} $};
		\node (b) at (3,1.5) {$ \TheL{n,k^{+};z_{k^{+}+2}}$};
		\draw (h) -- (b);
		\node[anchor = north] at (3/2,-1.25) {$\Resphi{\TheS{n,k}}$}; 
		\node[anchor = north] at (1.25,-2.75) { \small{$k+1\not\equiv 0\mod \ell $}};	 
	\end{tikzpicture}
\qquad \qquad\qquad
	\begin{tikzpicture}[scale = 1/3]
		\node (b) at (0,0) {$ \TheL{n,k;z_{k+2}} $};
		\node (g) at (-3,3) { $\TheL{n,k^{-};z_{k^{-}+2}}$ };
		\node (d) at (3,3) {$\TheL{n,k^{+};z_{k^{+}+2}}$ };
		\node (h) at (0,6) {$ \TheL{n,k;z_{k+2}} $};
		\draw (b) -- (g) -- (h) -- (d) -- (b);
		\node[anchor = north] at (0,-1.25) {$\ThePaff{n,k}$};
		\node[anchor = north] at (0,-2.75) 
			{ \small{$k  
		                \geq \ell $, $k+1\not\equiv 0\mod \ell $}};
	\end{tikzpicture}
\caption{Loewy diagrams of restrictions of some $\tl{n}$-modules.}\label{fig:one}
\end{figure}
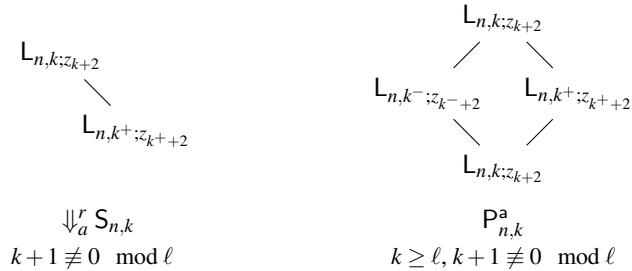
\end{subsection}
%
%
\begin{subsection}{The induction functor $\Indar{}$}\label{sub:Indar}
	We now compute the induction functor, obtained from the inclusion of the regular Temperley-Lieb algebra into its affine version, on the standard module $\TheS{n,k}$. 
\begin{Prop}\label{prop:inductionS}{\em (i)} For $k=1$, $\Indar{\TheS{n,1}}\simeq \TheW{n,1}$. For $k=0$ and $\beta\neq 0$, $\Indar{\TheS{n,0}}\simeq \TheW{n,0}$.

\noindent{\em (ii)} For $k\geq 2$, there is a natural inclusion $i_k:\Indar{\TheS{n,k-2}}{\longrightarrow}\ \Indar{\TheS{n,k}}$ and the short sequence
\begin{equation}\label{eq:exactIndS}0\longrightarrow\ \Indar{\TheS{n,k-2}}\overset{i_k}{\longrightarrow}\ \Indar{\TheS{n,k}}\longrightarrow\TheW{n,k} \longrightarrow 0
\end{equation}
is exact. (In the case $k=2$, this statement holds when $\beta\neq0$.)
\end{Prop}
The proof will use the definition of the generator $x_{n,k}$ introduced earlier and new elements: $y_{n,k}\in\tl n$ and a monic $r_{n,k}^j\in\Hom(n,k-2j)$ defined as:
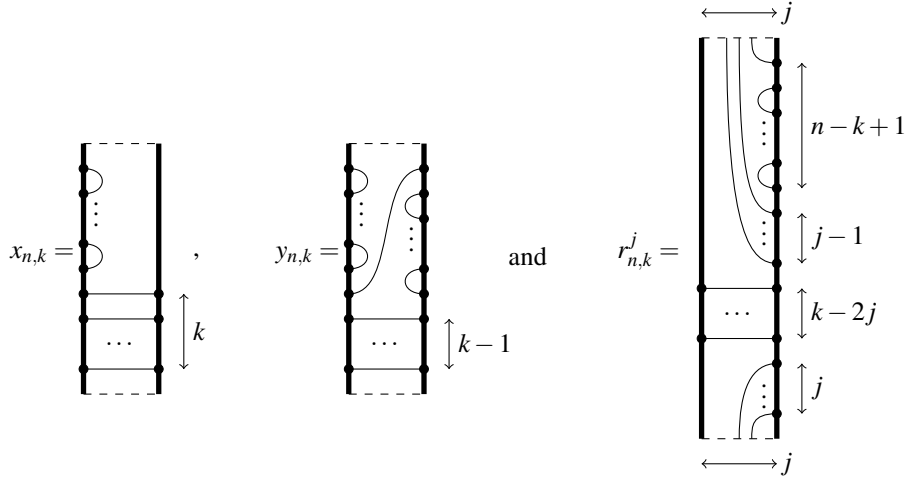
\begin{equation*}
\begin{tikzpicture}[scale = 1/3]
	\draw[line width = 2pt] (0,-2) -- (0,8);
	\draw[line width = 2pt] (3,-2) -- (3,8);
	\draw[dashed] (0,-2) -- (3,-2);
	\draw[dashed] (0,8) -- (3,8);
	\foreach \s in {0,2,3,4,5,7,8}
	{	
		\filldraw[black] (0,\s -1) circle (5pt);
	};
	\filldraw[black] (3,-1) circle (5pt);
	\filldraw[black] (3,1) circle (5pt);
	\filldraw[black] (3,2) circle (5pt);
	\draw (0,-1) -- (3,-1);
	\draw (0,1) -- (3,1);
	\draw (0,2) -- (3,2);
	\draw (0,3) .. controls (1,3) and (1,4) .. (0,4);
	\draw (0,6) .. controls (1,6) and (1,7) .. (0,7);
	\node[anchor = west] at (0.5,0) {$\dots$};
	\node[anchor = west] at (0,11/2) {$\vdots$};
	\node[anchor = east] at (0,7/2) {$x_{n,k} =$ };
	\draw[<->] (4,-1) -- (4,2);
	\node[anchor = west] at (4,0.5) {$k$};
	\node[anchor = west] at (4,7/2) {,};
	\node at (0,-5) {$\phantom{m}$};
\end{tikzpicture}\qquad
\begin{tikzpicture}[scale = 1/3]
	\draw[line width = 2pt] (0,-2) -- (0,8);
	\draw[line width = 2pt] (3,-2) -- (3,8);
	\draw[dashed] (0,-2) -- (3,-2);
	\draw[dashed] (0,8) -- (3,8);
	\foreach \s in {0,2,3,4,5,7,8}
	{	
		\filldraw[black] (0,\s -1) circle (5pt);
	};
	\foreach \s in {0,2,3,4,6,7,8}
	{	
		\filldraw[black] (3,\s -1) circle (5pt);
	};
	\draw (0,3) .. controls (1,3) and (1,4) .. (0,4);
	\draw (0,6) .. controls (1,6) and (1,7) .. (0,7);
	\node[anchor = west] at (0,11/2) {$\vdots$};
	\node[anchor = west] at (0.5,0) {$\dots$};
	\node[anchor = east] at (0,7/2) {$y_{n,k} =$ };
	\draw (0,-1) -- (3,-1);
	\draw (0,1) -- (3,1);
	\draw (0,2) .. controls (2,2) and (1,7) .. (3,7);
	\draw (3,2) .. controls (2,2) and (2,3) .. (3,3);
	\draw (3,5) .. controls (2,5) and (2,6) .. (3,6);
	\node[anchor = east] at (3,9/2) {$\vdots$};
	\draw[<->] (4,-1) -- (4,1);
	\node[anchor = west] at (4,0) {$k-1$};
	\node[anchor = west] at (6,7/2) {and\qquad\qquad};
	\node at (0,-5) {$\phantom{m}$};
\end{tikzpicture}
\begin{tikzpicture}[scale = 1/3]
\draw[line width = 2pt] (0,-1) -- (0,15);
	\draw[line width = 2pt] (3,-1) -- (3,15);
	\draw[dashed] (0,-1) -- (3,-1);
	\draw[dashed] (0,15) -- (3,15);
	\foreach \s in {0,2,3,5,6,8,9,10,12,13,14}
	{	
		\filldraw[black] (3,\s) circle (5pt);
	};
	\filldraw[black] (0,3) circle (5pt);
	\filldraw[black] (0,5) circle (5pt);
	\draw (0,3) -- (3,3);
	\draw (0,5) -- (3,5);
	\draw[<->] (4,3) -- (4,5);
	\node[anchor = west] at (4,4) {$k-2j$};
	\node[anchor = west] at (0.5,4) {$\dots$};
	\node[anchor = east] at (-0.5,13/2) {$r_{n,k}^j = $};
	\draw[<->] (0,-2) -- (3,-2);
	\node[anchor = east] at (4,-2) {$j$};
	\draw[<->] (0,16) -- (3,16);
	\node[anchor = east] at (4,16) {$j$};
	\draw (3,0) .. controls (2,0) and (2,-1/2) .. (2,-1);
	\node[anchor = east] at (3,1) {$\vdots$};
	\draw (3,2) .. controls (2,2) and (1.5,1) .. (1.5,-1);
	\draw[<->] (4,14) -- (4,9); \node[anchor = west] at (4,11.5) {$n-k+1$};
	\draw (3,14) .. controls (2,14) and (2,29/2) .. (2,15);
	\draw (3,13) .. controls (2,13) and (2,12) .. (3,12);
	\node[anchor = east] at (3,11.5) {$\vdots$};
	\draw[<->] (4,8) -- (4,6); \node[anchor = west] at (4,7) {$j-1$};
	\draw (3,10) .. controls (2,10) and (2,9) .. (3,9);
	\draw (3,8) .. controls (2,8) and (1.5,9) .. (1.5,15);
	\node[anchor = east] at (3,7.5) {$\vdots$};
	\draw (3,6) .. controls (2,6) and (1,7) .. (1,15);
	\draw[<->] (4,2) -- (4,0); \node[anchor = west] at (4,1) {$j$};
\end{tikzpicture}
\end{equation*}
The new elements $y_{n,k}$ and $r_{n,k}^j$ are defined for $k>0$ and, in the case of $r_{n,k}^j$, for integers $j$ such that $0\leq j\leq k/2$. 

We recall a definition and a simple observation from section \ref{sec:tlcategories}. A {\em through line} in a diagram in $\Hom(k,n)$ is a line that ties nodes from both sides of the diagram. The number of through lines in the concatenation $ab$ of $a\in\Hom_{\catl}(l,m)$ and $b\in\Hom_{\catl}(n,l)$ is at most the minimum of the number of through lines in $a$ and $b$. An immediate consequence is that, if $a\in\tl n$ has less than $k$ through lines, then $az=0$ for any $z\in\TheS{n,k}$. A more precise definition of the element $r_{n,k}^j$ can also be given: it is obtained from the following element in $\Hom(n,k)$
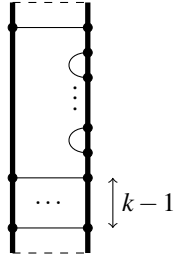
$$\begin{tikzpicture}[scale = 1/3]
	\draw[line width = 2pt] (0,-2) -- (0,8);
	\draw[line width = 2pt] (3,-2) -- (3,8);
	\draw[dashed] (0,-2) -- (3,-2);
	\draw[dashed] (0,8) -- (3,8);
	\foreach \s in {0,2,8}
	{	
		\filldraw[black] (0,\s -1) circle (5pt);
	};
	\foreach \s in {0,2,3,4,6,7,8}
	{	
		\filldraw[black] (3,\s -1) circle (5pt);
	};
	\node[anchor = west] at (0.5,0) {$\dots$};
	\draw (0,-1) -- (3,-1);
	\draw (0,1) -- (3,1);
	\draw (0,7) -- (3,7);
	\draw (3,2) .. controls (2,2) and (2,3) .. (3,3);
	\draw (3,5) .. controls (2,5) and (2,6) .. (3,6);
	\node[anchor = east] at (3,9/2) {$\vdots$};
	\draw[<->] (4,-1) -- (4,1);
	\node[anchor = west] at (4,0) {$k-1$};
\end{tikzpicture}
$$
by bending upward the top $j$ through lines and downward the bottom $j$ ones.

\begin{proof}The proof is broken into several claims, each being stated in italics.

\medskip
%
%
\noindent{\itshape For $k>0$ and $0\leq j\leq k/2$, $y_{n,k}x_{n,k}=x_{n,k}$ and $r_{n,k}^jy_{n,k}=r_{n,k}^j$.} This result is obtained by simply drawing the diagrams involved.

\medskip
%
%
\noindent{\itshape The induced module $\Indar{\TheS{n,k}}$ is isomorphic to $\TheW{n,k}$ if $k=1$, or if $k=0$ and $\beta\neq0$.} The case $n=k=1$ is trivial.
For $n\geq 3,k=1$, the previous proposition shows that any element of $\Indar{\TheS{n,k}}$ is of the form $a\otimes_{\tl n}x_{n,1}=ay_{n,1}\otimes_{\tl n} x_{n,1}$ and thus of the form $b\otimes_{\tl n}x_{n,1}$ with $b$ of the form
$$
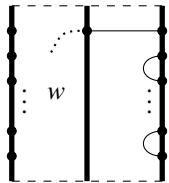
\begin{tikzpicture}[baseline={(current bounding box.center)},scale=1/3]
	\draw[line width = 2pt] (-3,1) -- (-3,8);
	\draw[line width = 2pt] (0,1) -- (0,8);
	\draw[line width = 2pt] (3,1) -- (3,8);
	\draw[dashed] (-3,1) -- (3,1);
	\draw[dashed] (-3,8) -- (3,8);
	\foreach \s in {3,4,6,7,8}
	{	
		\filldraw[black] (-3,\s -1) circle (5pt);
	};
	\filldraw[black] (0,7) circle (5pt);
	\foreach \s in {3,4,6,7,8}
	{	
		\filldraw[black] (3,\s -1) circle (5pt);
	};
	\draw (3,2) .. controls (2,2) and (2,3) .. (3,3);
	\draw (3,5) .. controls (2,5) and (2,6) .. (3,6);
	\node[anchor = west] at (-3,9/2) {$\vdots$};
	\node[anchor = east] at (3,9/2) {$\vdots$};
	\node[anchor = east] at (-0.5,9/2) {$w$};
	\draw[dotted,thick] (0,7) .. controls (-1,7) .. (-1.5,6);
	\draw (0,7) -- (3,7);
\end{tikzpicture}
$$
where $w$ is a diagram in $\Hom(1,n)$. Moreover the action of $\atl n$ on $\Indar{\TheS{n,1}}$ transforms $\Hom(1,n)$ into a module isomorphic to $\TheW{n,1}$. 

If $k=0$ and $\beta\neq0$, then $x_{n,0}=\beta^{-n/2}(u_1u_3\dots u_{n-1})x_{n,0}$. Any element of $\Indar{\TheS{n,0}}$ has the form 
$$a\otimes_{\tl n}x_{n,0}=\beta^{-n/2}a(e_1e_3\dots e_{n-1})\otimes_{\tl n}x_{n,0}=\beta^{-n/2}w(x_{n,0})^t\otimes_{\tl n}x_{n,0}$$ 
where $(x_{n,0})^t$ denotes the left-right mirror image of $x_{n,0}$ (see section \ref{sec:tlcategories}) and $w\in\Hom_{\catl}(0,n)$. Clearly, the action of $\atl n$ on $w(x_{n,0})^t\otimes_{\tl n}x_{n,0}$ is isomorphic to its action on $\Hom_{\catl}(0,n)\simeq\TheW{n,0}$. 

\medskip
%
%
\noindent{\itshape Every non-zero element of $\Indar{\TheS{n,k}}$, $k>0$, can be written as a sum of elements of the form $l \circ r_{n,k}^j \otimes_{\tl{n}} x_{n,k}$ where $0\leq j\leq k/2 $ and $l \in \Hom(k-2j,n)$ is a monic diagram.} By linearity and because $x_{n,k}$ is a generator of $\TheS{n,k}$, it is sufficient to check the claim for elements of the form $a\otimes_{\tl n}x_{n,k}$ with $a\in\atl n$. Let $a$ be chosen in such a way that $a$ has $k'\leq k$ through lines and, for any other choice $a'\in\atl n$ such that $a'\otimes_{\tl n}x_{n,k}=a\otimes_{\tl n}x_{n,k}$, $a'$ has at least $k'$ through lines. By the first observation of this proof, $a\otimes_{\tl n}x_{n,k}=a\otimes_{\tl n}y_{n,k}x_{n,k}$ and it is possible to replace $a$ by $ay_{n,k}$ in the study. The element $ay_{n,k}$ has also $k'$, since this is the minimal number of through lines in this equivalence class $a\otimes_{\tl n}x_{n,k}$ and concatenation cannot increase the number of through lines.

The element $ay_{n,k}\in\atl n$ can be factored into left and right parts, $ay_{n,k}=l\circ r$, where $l\in\Hom_{\catl}(k',n)$ and $r\in\Hom_{\catl}(n,k')$. The right part is partially revealed in
$$ay_{n,k}= 
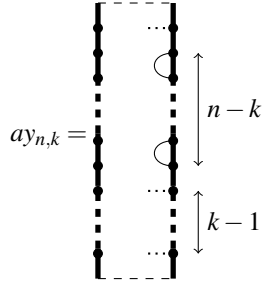
\begin{tikzpicture}[baseline={(current bounding box.center)},scale=1/3]
	\draw[line width = 2pt] (0,-2) -- (0,-0.8);
	\draw[line width = 2pt] (0,1.7) -- (0,3.8);
	\draw[line width = 2pt] (0,5.9) -- (0,9);
	\draw[dashed,line width = 2pt] (0,-0.8) -- (0,1.7);
	\draw[dashed,line width = 2pt] (0,3.8) -- (0,6.0);
	\draw[line width = 2pt] (3,-2) -- (3,-0.8);
	\draw[line width = 2pt] (3,1.7) -- (3,3.8);
	\draw[line width = 2pt] (3,5.9) -- (3,9);
	\draw[dashed,line width = 2pt] (3,-0.8) -- (3,1.7);
	\draw[dashed,line width = 2pt] (3,3.8) -- (3,6.0);
	\draw[dashed] (0,-2) -- (3,-2);
	\draw[dashed] (0,9) -- (3,9);
	\foreach \s in {0,2.5,3.5,4.5,7,8,9}
	{	
		\filldraw[black] (0,\s -1) circle (5pt);
		\filldraw[black] (3,\s -1) circle (5pt);
	};
	\draw[dotted,thick] (2,-1) -- (3,-1);
	\draw[dotted,thick] (2,1.5) -- (3,1.5);
	\draw[dotted,thick] (2,8) -- (3,8);
	\draw (3,2.5) .. controls (2,2.5) and (2,3.5) .. (3,3.5);
	\draw (3,6) .. controls (2,6) and (2,7) .. (3,7);
	\draw[<->] (4,-1) -- (4,1.5);
	\node[anchor = west] at (4,0.25) {$k-1$};
	\draw[<->] (4,2.5) -- (4,7);
	\node[anchor = west] at (4,4.75) {$n-k$};
\end{tikzpicture}
$$
where, of the remaining $k$ nodes on the right side, $k'$ must be tied by a through line and the remaining $n-(n-k)-k'=k-k'$ must tied pairwise by arcs. If $k=k'$, then there are no additional arcs and the right part $r$ coincide with $r_{n,k}^0$. If $k>k'$, then the number of additional arcs is $j=(k-k')/2$. In any of the these $j$ arcs do not cross the top and bottom boundaries, then an element $b\in\tl n$ that annihilates $x_{n,k}$ can be factored out as in $ay_{n,k}=ay_{n,k}b$ and then $ay_{n,k}\otimes_{\tl n}x_{n,k}=ay_{n,k}\otimes_{\tl n}bx_{n,k}=0$, a contradiction. (This factorisation is always possible. The following examples, drawn for $n=5, k=3, k'=1$, show how this is done by drawing only the right part of $ay_{n,k}$:
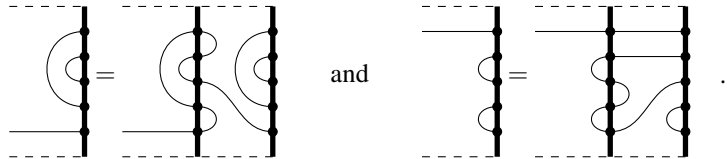
$$\begin{tikzpicture}[baseline={(current bounding box.center)},scale=1/3]
	\draw[line width = 2pt] (3,1) -- (3,7);
	\draw[dashed] (0,1) -- (3,1);
	\draw[dashed] (0,7) -- (3,7);
	\foreach \s in {2,3,4,5,6}
	{	
		\filldraw[black] (3,\s) circle (5pt);
	};
	\draw (0,2) -- (3,2);
	\draw (3,6) .. controls (1,6) and (1,3) .. (3,3);
	\draw (3,4) .. controls (2,4) and (2,5) .. (3,5);
\end{tikzpicture}=
\begin{tikzpicture}[baseline={(current bounding box.center)},scale=1/3]
	\draw[line width = 2pt] (3,1) -- (3,7);
	\draw[line width = 2pt] (6,1) -- (6,7);
	\draw[dashed] (0,1) -- (6,1);
	\draw[dashed] (0,7) -- (6,7);
	\foreach \s in {2,3,4,5,6}
	{	
		\filldraw[black] (3,\s) circle (5pt);
		\filldraw[black] (6,\s) circle (5pt);
	};
	\draw (0,2) -- (3,2);
	\draw (3,6) .. controls (1,6) and (1,3) .. (3,3);
	\draw (3,4) .. controls (2,4) and (2,5) .. (3,5);
	\draw (6,6) .. controls (4,6) and (4,3) .. (6,3);
	\draw (6,4) .. controls (5,4) and (5,5) .. (6,5);
	\draw (3,6) .. controls (4,6) and (4,5) .. (3,5);
	\draw (3,3) .. controls (4,3) and (4,2) .. (3,2);
	\draw (3,4) .. controls (4.5,4) and (4.5,2) .. (6,2);
\end{tikzpicture}\qquad \textrm{and}\qquad
\begin{tikzpicture}[baseline={(current bounding box.center)},scale=1/3]
	\draw[line width = 2pt] (3,1) -- (3,7);
	\draw[dashed] (0,1) -- (3,1);
	\draw[dashed] (0,7) -- (3,7);
	\foreach \s in {2,3,4,5,6}
	{	
		\filldraw[black] (3,\s) circle (5pt);
	};
	\draw (0,6) -- (3,6);
	\draw (3,4) .. controls (2,4) and (2,5) .. (3,5);
	\draw (3,2) .. controls (2,2) and (2,3) .. (3,3);
\end{tikzpicture}=
\begin{tikzpicture}[baseline={(current bounding box.center)},scale=1/3]
	\draw[line width = 2pt] (3,1) -- (3,7);
	\draw[line width = 2pt] (6,1) -- (6,7);
	\draw[dashed] (0,1) -- (6,1);
	\draw[dashed] (0,7) -- (6,7);
	\foreach \s in {2,3,4,5,6}
	{	
		\filldraw[black] (3,\s) circle (5pt);
		\filldraw[black] (6,\s) circle (5pt);
	};
	\draw (0,6) -- (6,6);
	\draw (3,5) -- (6,5);
	\draw (3,4) .. controls (2,4) and (2,5) .. (3,5);
	\draw (3,2) .. controls (2,2) and (2,3) .. (3,3);
	\draw (3,3) .. controls (4,3) and (4,4) .. (3,4);
	\draw (3,2) .. controls (4.5,2) and (4.5,4) .. (6,4);
	\draw (6,2) .. controls (5,2) and (5,3) .. (6,3);
	\node[anchor = west] at (7,4) {.};
\end{tikzpicture}
$$
The example on the left is when the top node on the right is part of an arc. The $b$ has then less than $k=3$ through lines and $bx_{n,k}=0$. The example on the right has an arc between the $k-1$ nodes on the bottom and, now, the $b$ will join two through lines of $x_{n,k}$.)
Hence all $j$ arcs on the right side must cross the top and bottom boundaries. Because a through line may not cross an arc, there is a single possible positions for the through lines and additional arcs and the right part must be $r=r_{n,k}^j$. Every element $a\otimes_{\tl n}$, with an $a$ having the minimal number $k'$ of through lines, can be written as $l\circ r_{n,k}^j\otimes_{\tl n}x_{n,k}$ as claimed.

\medskip
%
%

\noindent{\itshape For $k\geq 2$, the map $i_k:\Indar{\TheS{n,k-2}}\to\Indar{\TheS{n,k}}$ defined by $l\circ r_{n,k-2}^j\otimes_{\tl n}x_{n,k-2}\mapsto l\circ r_{n,k}^{j+1}\otimes_{\tl n}x_{n,k}$ is an injective $\atl n$-morphism.} The previous step has shown that, for a fixed $j$, the elements of the form $l\circ r_{n,k}^j\otimes_{\tl n}x_{n,k}$ are completely determined by $l$. The map $i_k$ is thus injective. Moreover $\atl n$ acts, in both modules, only on this $l$. The two actions coincide and the statement follows. Note that $i_k(\Indar{\TheS{n,k-2}})$ is spanned by elements of the form $l\circ r_{n,k}^j\otimes_{\tl n}x_{n,k}$ with $j\geq 1$.

\medskip
%
%
\noindent{\itshape The cokernel of the inclusion $i_k$ is isomorphic to the standard module $\TheW{n,k}$.}
Since $i_k(\Indar{\TheS{n,k-2}})=\textrm{span}\{l\circ r_{n,k}^j\otimes_{\tl n}x_{n,k}\, |\, j\geq 1, l\textrm{ a diagram}\in\Hom(k-2j,n)\}$, then $\Indar{\TheS{n,k}}/i_k(\Indar{\TheS{n,k-2}})$ has a basis $\{l\circ r_{n,k}^0\otimes_{\tl n}x_{n,k}\, |\, l\textrm{ a diagram}\in\Hom(k,n)\}$ that is in one-to-one correspondence with the basis of $\TheW{n,k}$ made of the diagrams in $\Hom(k,n)$. It remains to compare the action on $\Indar{\TheS{n,k}}/i_k(\Indar{\TheS{n,k-2}})$ with that on $\TheW{n,k}$. The action of $a\in\atl n$ on $l\circ r_{n,k}^0\otimes_{\tl n}x_{n,k}\in\Indar{\TheS{n,k}}$ and on $l\in\TheW{n,k}$ coincide as long as no pair of through lines of $l$ are joined. But if such a pair if joined, the diagram $al\in\TheW{n,k}$ will have less than $k$ through lines and will thus be identified to $0$. In $\Indar{\TheS{n,k}}$ the element $al\circ r_{n,k}^0\otimes_{\tl n}x_{n,k}$ is not zero, but the diagram $al\circ r_{n,k}^0$ has now less than $k$ through lines. It thus belongs to the image $i_k(\Indar{\TheS{n,k-2}})$ and maps to zero under the projection onto $\Indar{\TheS{n,k}}/i_k(\Indar{\TheS{n,k-2}})$. Hence the $\atl n$-actions on $\Indar{\TheS{n,k}}/i_k(\Indar{\TheS{n,k-2}})$ and $\TheW{n,k}$ coincide. The exactness of \eqref{eq:exactIndS} follows.
\end{proof}
We shall prove that $\Indar{\TheS{n,k}}$ is indecomposable; the following lemma will slightly lighten the proof.
\begin{Lem}\label{lem:hom.indS.W}
	If $k'>k$, then
		\begin{equation*}
			\Hom\left( \Indar{\TheS{n,k}}, \TheW{n,k'} \right) \simeq 0.
		\end{equation*}
\end{Lem}
\begin{proof}
	The proof proceeds by induction on $k$. Proposition \ref{prop:hom.w.to.w} gives the cases $k=0$ and $k=1$. Assume thus that the result stands for $k\geq 2$. The previous discussion states the existence of a short exact sequence
	\begin{equation}\label{eq:indarTheS}
		0 \longrightarrow \Indar{\TheS{n,k-2}} \overset{i_k}{\longrightarrow} \Indar{\TheS{n,k}} \overset{p_k}{\longrightarrow} \TheW{n,k}\longrightarrow 0,
	\end{equation}
	which gives the exact sequence
	\begin{equation}
		0 \rightarrow \Hom(\TheW{n,k},\TheW{n,k'}) \longrightarrow \Hom(\Indar{\TheS{n,k}},\TheW{n,k'}) \longrightarrow \Hom(\Indar{\TheS{n,k-2}},\TheW{n,k'}).
	\end{equation}
	The term on the left vanishes by proposition \ref{prop:hom.w.to.w}, while the one on the right vanishes by the induction hypothesis. The middle term must vanish as well.
\end{proof}
\begin{Prop}\label{pro:indS}
	For all $0 \leq k \leq n$, $\Indar{\TheS{n,k}}$ is indecomposable.
\end{Prop}
\begin{proof}
	The proof is again by induction on $k$. If $k = 0$ or $1$, then $\Indar{\TheS{n,k}} \simeq \TheW{n,k}$ which is indecomposable by corollary \ref{coro:hom.h.and.w}. Assume now that the result holds for integers $\geq 2$ up to (but excluding) $k$. The module $\Indar{\TheS{n,k}}$ is decomposable if and only if there exists a non-zero idempotent $f \in \End{\Indar{\TheS{n,k}}}$ that is not an isomorphism. Suppose such an non-zero idempotent $f$ exists. The previous lemma gives $p_k \circ f \circ i_{k} = 0$ where the maps $i_k$ and $p_k$ are those appearing in the short sequence \eqref{eq:indarTheS}. By universality of kernels and cokernels, it follows that there exists $g \in \End{\Indar{\TheS{n,k-2}}}$ and $h \in \End{\TheW{n,k}}$ such that $p_{k}\circ f = h \circ p_{k}$ and $f\circ i_{k} =  i_{k} \circ g$. Furthermore, since $f$ is an idempotent, so are $g$ and $h$. But because $\Indar{\TheS{n,k-2}}$ and $\TheW{n,k}$ are indecomposable, they are either isomorphisms or trivial. Clearly, if they are both isomorphisms, then $f$ is an isomorphism too, and if they are both zero, then so is $f$. Either conclusion would be a contradiction.

Suppose therefore that $h$ is an isomorphism and $g=0$. By universality of cokernels, there exists $\bar{h}: \TheW{n,k} \to \Indar{\TheS{n,k}}$ such that $\bar{h} \circ p_{k} = f$ and thus 
$$h^{-1}\circ p_{k}\circ \bar{h} \circ p_{k} = h^{-1}\circ p_{k} \circ f = h^{-1}\circ h \circ p_{k} = p_{k}.$$
Since $p_{k}$ is an epimorphism, $h^{-1}\circ p_{k}\circ \bar{h} = 1_{\TheW{n,k}}$, and the exact sequence \eqref{eq:indarTheS} splits:
$$\Indar{\TheS{n,k}} \simeq \Indar{\TheS{n,k-2}} \oplus \TheW{n,k}.$$
Dual arguments for the case when $g$ is an isomorphism and $h = 0$ leads to the same conclusion. 

Hence, if $\Indar{\TheS{n,k}}$ is decomposable, then $\Indar{\TheS{n,k}} \simeq \Indar{\TheS{n,k-2}} \oplus \TheW{n,k}$. However, lemma \ref{lem:functorId} has shown that $\Indphi{} \circ \Indar{} \xrightarrow{\sim} \id_{\Mod\tl{n}}$ and proposition \ref{prop:indH} then gives
	\begin{equation}
		\TheS{n,k} \simeq\ \Indphi{}(\Indar{\TheS{n,k-2}}) \oplus \Indphi{}\TheW{n,k} \simeq \TheS{n,k-2} \oplus \TheS{n,k},
	\end{equation} 
which is impossible. (If $n=k=2$, the last term $\TheS{n,k}$ must be replaced by $\TheM{2,2}$, which is still impossible.) Thus $\Indar{\TheS{n,k}}$ must be indecomposable.
\end{proof}
If $q$ is generic, then $\TheS{n,k}$ is irreducible and projective. Since induction transforms projective modules into projective ones, it follows that $\Indar{\TheS{n,k}}$ is an indecomposable projective $\atl{n}$-module. This simple observation leads to the Peirce decomposition of $\atl n$.
\begin{Prop}If $q$ is generic, then 
\begin{equation}
	\atl{n} \simeq {\bigoplus}'_{0\leq k\leq n} \dim (\TheS{n,k}) \Indar{\TheS{n,k}}
\end{equation}
as a left-module over itself. Here $\oplus'$ indicates that the sum runs over $k$'s of the parity of $n$. 
\end{Prop}
Furthermore, it also follows that the projective dimension of $\TheW{n,k}$ is always one when $q$ is generic, except for $\TheW{n,0}$ and $\TheW{n,1}$ which are projective.
\begin{Prop}
	If $q$ is generic and $0< t \leq (n-k)/2$,
		\begin{equation}
			\Resar{\TheW{n,k;z}} \simeq \bigoplus_{r=0}^{(n-k)/2} \TheS{n,k+2r}, \qquad \Resar{\TheL{n,k;z_{k+2t}}} \simeq \bigoplus_{r=0}^{t-1} \TheS{n,k+2r}.
		\end{equation}
\end{Prop}
\begin{proof} Recall that the rank of a diagram is defined as the minimal number of intersections of links with the top of the fundamental rectangle (minimal among all isotopic diagrams). If $a$ and $w$ are diagrams of rank $0$ and $r$ respectively, and $w$ is drawn with $r$ such intersections, their concatenation has $r$ intersections and the diagram $aw$ has thus a rank equal to $r$ or less. The subspace ${\mathsf U}_r\subset \Resar{\TheW{n,k;z}}$ spanned by diagrams of rank at most $r$ is thus a $\tl n$-submodule and there is natural filtration
$${\mathsf U}_0\subset {\mathsf U}_1\subset \dots \subset {\mathsf U}_{(n-k)/2}=\Resar{\TheW{n,k;z}}.$$
The cokernel of the injection ${\mathsf U}_r\to {\mathsf U}_{r+1}$ is spanned by elements of the form $w + {\mathsf U}_{r}$ where $w$ is a diagram of rank $r$. The rank of $aw$, again with $a\in\tl n$ and $\rk w=r$, may be less than $r$. This occurs {\em (i)} when $a$ joins two arcs crossing the top or bottom of the rectangle or {\em (ii)} when it links one such arc with a through line of $w$ or {\em (iii)} when it closes an uncontractible loop. Here are examples of these three scenarios, drawn for $a$'s in $\tl 5, \tl 3$ and $\tl 2$ and $w$'s in $\TheW{5,1;z}$, $\TheW{3,1;z}$ and $\TheW{2,0;z}$, respectively:
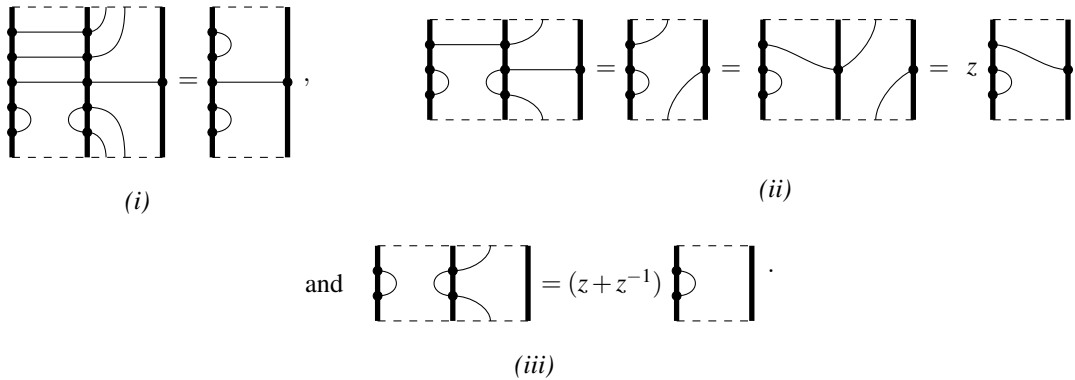
\begin{equation*} 
\begin{tikzpicture}[baseline={(current bounding box.center)},scale=1/3]
	\draw[line width = 2pt] (0,-1) -- (0,5);
	\draw[line width = 2pt] (3,-1) -- (3,5);
	\draw[line width = 2pt] (6,-1) -- (6,5);
	\draw[dashed] (0,-1) -- (6,-1);
	\draw[dashed] (0,5) -- (6,5);
	\foreach \s in {1,...,5}
	{	
		\filldraw[black] (0,\s -1) circle (5pt);
		\filldraw[black] (3,\s -1) circle (5pt);
	};
	\filldraw[black] (6,2) circle (5pt);
	\draw (0,0) .. controls (1,0) and (1,1) .. (0,1);%
	\draw (3,0) .. controls (2,0) and (2,1) .. (3,1);%
	\draw (0,2) -- (6,2); %
	\draw (0,3) -- (3,3);
	\draw (0,4) -- (3,4);
	\draw (3,4) .. controls (7/2,4) and (3.75,4.5) .. (3.75,5);
	\draw (3,3) .. controls (4,3) and (9/2,3.5) .. (9/2,5);
	\draw (3,0) .. controls (7/2,0) and (3.75,-1/2) .. (3.75,-1);
	\draw (3,1) .. controls (4,1) and (9/2,1/2) .. (9/2,-1);
	\node[anchor = west] at (6.25, 2.0) {$=$};
	\node[anchor = north] at (5,-2) {\em (i)};
\end{tikzpicture} 
\begin{tikzpicture}[baseline={(current bounding box.center)},scale=1/3]
	\draw[line width = 2pt] (0,-1) -- (0,5);
	\draw[line width = 2pt] (3,-1) -- (3,5);
	\draw[dashed] (0,-1) -- (3,-1);
	\draw[dashed] (0,5) -- (3,5);
	\foreach \s in {1,...,5}
	{	
		\filldraw[black] (0,\s -1) circle (5pt);
	};
	\filldraw[black] (3,2) circle (5pt);
	\draw (0,0) .. controls (1,0) and (1,1) .. (0,1);%
	\draw (0,3) .. controls (1,3) and (1,4) .. (0,4);%
	\draw (0,2) -- (3,2); %
	\node[anchor = west] at (3.25, 2.0) {$,$};
	\node[anchor = north] at (1,-2) {$\phantom{\textrm{\em (i)}}$};
\end{tikzpicture} 
\qquad\qquad 
\begin{tikzpicture}[baseline={(current bounding box.center)},scale=1/3]
	\draw[line width = 2pt] (0,-1) -- (0,3);
	\draw[line width = 2pt] (3,-1) -- (3,3);
	\draw[line width = 2pt] (6,-1) -- (6,3);
	\draw[dashed] (0,-1) -- (6,-1);
	\draw[dashed] (0,3) -- (6,3);
	\foreach \s in {1,...,3}
	{	
		\filldraw[black] (0,\s -1) circle (5pt);
		\filldraw[black] (3,\s -1) circle (5pt);
	};
	\filldraw[black] (6,1) circle (5pt);
	\draw (0,0) .. controls (1,0) and (1,1) .. (0,1);%
	\draw (3,0) .. controls (2,0) and (2,1) .. (3,1);%
	\draw (0,2) -- (3,2); %
	\draw (3,1) -- (6,1);
	\draw (3,2) .. controls (3.75,2) and (4.5,2.5) .. (4.5,3);
	\draw (3,0) .. controls (3.75,0) and (4.5,-1/2) .. (4.5,-1);
	\node[anchor = west] at (6.25, 1.0) {$=$};
	\node[anchor = north] at (1,-3) {$\phantom{\textrm{\em (i)}}$};
\end{tikzpicture} 
\begin{tikzpicture}[baseline={(current bounding box.center)},scale=1/3]
	\draw[line width = 2pt] (0,-1) -- (0,3);
	\draw[line width = 2pt] (3,-1) -- (3,3);
	\draw[dashed] (0,-1) -- (3,-1);
	\draw[dashed] (0,3) -- (3,3);
	\foreach \s in {1,...,3}
	{	
		\filldraw[black] (0,\s -1) circle (5pt);
	};
	\filldraw[black] (3,1) circle (5pt);
	\draw (0,0) .. controls (1,0) and (1,1) .. (0,1);%
	\draw (0,2) .. controls (0.75,2) and (1.5,2.5) .. (1.5,3);
	\draw (3,1) .. controls (2.25,0.666) and (1.5,-0.333) .. (1.5,-1);
	\node[anchor = west] at (3.25, 1.0) {$=$};
	\node[anchor = north] at (1,-3) {$\phantom{\textrm{\em (i)}}$};
\end{tikzpicture} 
\begin{tikzpicture}[baseline={(current bounding box.center)},scale=1/3]
	\draw[line width = 2pt] (0,-1) -- (0,3);
	\draw[line width = 2pt] (3,-1) -- (3,3);
	\draw[line width = 2pt] (6,-1) -- (6,3);
	\draw[dashed] (0,-1) -- (6,-1);
	\draw[dashed] (0,3) -- (6,3);
	\foreach \s in {1,...,3}
	{	
		\filldraw[black] (0,\s -1) circle (5pt);
	};
	\filldraw[black] (3,1) circle (5pt);
	\filldraw[black] (6,1) circle (5pt);
	\draw (0,0) .. controls (1,0) and (1,1) .. (0,1);%
	\draw (0,2) .. controls (1,2) and (2,1) .. (3,1);
	\draw (3,1) .. controls (3.75,1.333) and (4.5,2.333) .. (4.5,3);
	\draw (6,1) .. controls (5.25,0.666) and (4.5,-0.333) .. (4.5,-1);
	\node[anchor = north] at (0.5,-3) {\em (ii)};
	\node[anchor = west] at (6.25, 1.0) {$=$};
\end{tikzpicture} 
\begin{tikzpicture}[baseline={(current bounding box.center)},scale=1/3]
	\draw[line width = 2pt] (0,-1) -- (0,3);
	\draw[line width = 2pt] (3,-1) -- (3,3);
	\draw[dashed] (0,-1) -- (3,-1);
	\draw[dashed] (0,3) -- (3,3);
	\foreach \s in {1,...,3}
	{	
		\filldraw[black] (0,\s -1) circle (5pt);
	};
	\filldraw[black] (3,1) circle (5pt);
	\draw (0,0) .. controls (1,0) and (1,1) .. (0,1);%
	\draw (0,2) .. controls (1,2) and (2,1) .. (3,1);
	\node[anchor = east] at (-0.25, 1.0) {$z$};
	\node[anchor = north] at (1,-3) {$\phantom{\textrm{\em (i)}}$};
\end{tikzpicture} 
\end{equation*}
\begin{equation*}
\begin{tikzpicture}[baseline={(current bounding box.center)},scale=1/3]
	\draw[line width = 2pt] (0,-1) -- (0,2);
	\draw[line width = 2pt] (3,-1) -- (3,2);
	\draw[line width = 2pt] (6,-1) -- (6,2);
	\draw[dashed] (0,-1) -- (6,-1);
	\draw[dashed] (0,2) -- (6,2);
	\foreach \s in {1,...,2}
	{	
		\filldraw[black] (0,\s -1) circle (5pt);
		\filldraw[black] (3,\s -1) circle (5pt);
	};
	\draw (0,0) .. controls (1,0) and (1,1) .. (0,1);%
	\draw (3,0) .. controls (2,0) and (2,1) .. (3,1);%
	\draw (3,1) .. controls (3.75,1) and (4.5,1.5) .. (4.5,2);
	\draw (3,0) .. controls (3.75,0) and (4.5,-1/2) .. (4.5,-1);
	\node[anchor = east] at (-1, 0.5) {and};
	\node[anchor = east] at (11.75, 0.5) {$=(z+z^{-1})$};
	\node[anchor = north] at (6.25,-2) {\em (iii)};
\end{tikzpicture} 
\begin{tikzpicture}[baseline={(current bounding box.center)},scale=1/3]
	\draw[line width = 2pt] (0,-1) -- (0,2);
	\draw[line width = 2pt] (3,-1) -- (3,2);
	\draw[dashed] (0,-1) -- (3,-1);
	\draw[dashed] (0,2) -- (3,2);
	\foreach \s in {1,...,2}
	{	
		\filldraw[black] (0,\s -1) circle (5pt);
	};
	\draw (0,0) .. controls (1,0) and (1,1) .. (0,1);%
	\node[anchor = west] at (3.25, 1.0) {$.$};
	\node[anchor = north] at (1,-2) {$\phantom{\textrm{\em (i)}}$};
\end{tikzpicture} 
\end{equation*}
We next introduce the following ``surgery" of diagrams. An affine monic diagram $w$ of rank $r$ having $k$ through lines is sent to the rank $0$ monic diagram with $k+2r$ through lines obtained by ``cutting" open the $r$ arcs crossing the top in the diagram $w$ and glueing the $2r$ ends to the right side. For instance
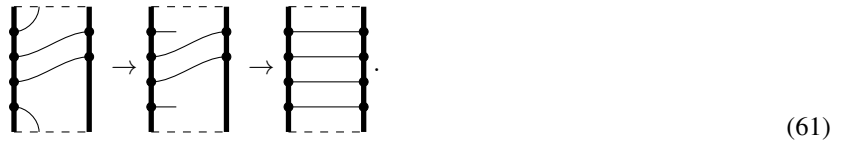
\begin{equation}
\begin{tikzpicture}[scale=1/3]
			\draw[line width = 2pt] (0,-1) -- (0,4);
			\draw[line width = 2pt] (3,-1) -- (3,4);
			\draw[dashed] (0,-1) -- (3,-1);
			\draw[dashed] (0,4) -- (3,4);
			\foreach \s in {1,...,4}
			{	
				\filldraw[black] (0,\s -1) circle (5pt);
			};
			\filldraw[black] (3,2) circle (5pt);
			\filldraw[black] (3,3) circle (5pt);
			\draw (0,0) .. controls (1/2,0) and (1,-1/2) .. (1,-1);
			\draw (0,3) .. controls (1/2,3) and (1,7/2) .. (1,4);
			\draw (0,1) .. controls (1,1) and (2,2) .. (3,2);
			\draw (0,2) .. controls (1,2) and (2,3) .. (3,3);
			\node[anchor = west] at (7/2,3/2) {$\to$};
\end{tikzpicture}
\begin{tikzpicture}[scale=1/3]
			\draw[line width = 2pt] (0,-1) -- (0,4);
			\draw[line width = 2pt] (3,-1) -- (3,4);
			\draw[dashed] (0,-1) -- (3,-1);
			\draw[dashed] (0,4) -- (3,4);
			\foreach \s in {1,...,4}
			{	
				\filldraw[black] (0,\s -1) circle (5pt);
			};
			\filldraw[black] (3,2) circle (5pt);
			\filldraw[black] (3,3) circle (5pt);
			\draw (0,0) -- (1,0);
			\draw (0,3) -- (1,3);
			\draw (0,1) .. controls (1,1) and (2,2) .. (3,2);
			\draw (0,2) .. controls (1,2) and (2,3) .. (3,3);
			\node[anchor = west] at (7/2,3/2) {$\to$};
\end{tikzpicture}	
\begin{tikzpicture}[scale=1/3]
			\draw[line width = 2pt] (0,-1) -- (0,4);
			\draw[line width = 2pt] (3,-1) -- (3,4);
			\draw[dashed] (0,-1) -- (3,-1);
			\draw[dashed] (0,4) -- (3,4);
			\foreach \s in {1,...,4}
			{	
				\filldraw[black] (0,\s -1) circle (5pt);
				\filldraw[black] (3,\s -1) circle (5pt);
				\draw (0,\s -1) -- (3,\s -1);
			};
			\node[anchor = west] at (3,3/2) {$.$};
\end{tikzpicture}	
\end{equation}
The linear map $\phi:{\mathsf U}_r\to \oplus_{0\leq i\leq r}\TheS{n,k+2i}$ obtained by extending linearly the above surgery is clearly invertible, a diagram from $\TheS{n,k+2i}$ giving back the element in ${\mathsf U}_r$ upon bending upward the $i$ highest through lines and downward the $i$ lowest ones. This morphism actually projects onto a $\tl n$-isomorphism ${\mathsf U}_r/{\mathsf U}_{r-1}\to \TheS{n,k+2r}$. The proof of this claim consists of a check of what happens when the only right arc in the generator $u_i,1\leq i\leq n-1$ of $\tl n$ links zero, one or two ends of arcs intersecting the top of $w+{\mathsf U}_{r-1}$ where $w$ is of rank $r$. The example {\em (ii)} above shows what happens when the generator $u_i$ connects one end of an intersecting arc:  on one hand,
$$\phi\Big(\ \begin{tikzpicture}[baseline={(current bounding box.center)},scale=1/3]
	\draw[line width = 2pt] (3,-1) -- (3,3);
	\draw[line width = 2pt] (6,-1) -- (6,3);
	\draw[dashed] (3,-1) -- (6,-1);
	\draw[dashed] (3,3) -- (6,3);
	\foreach \s in {1,...,3}
	{	
		\filldraw[black] (3,\s -1) circle (5pt);
	};
	\filldraw[black] (6,1) circle (5pt);
	\draw (3,1) -- (6,1);
	\draw (3,2) .. controls (3.75,2) and (4.5,2.5) .. (4.5,3);
	\draw (3,0) .. controls (3.75,0) and (4.5,-1/2) .. (4.5,-1);
\end{tikzpicture} 
\ +{\mathsf U}_0 \Big)=
\begin{tikzpicture}[baseline={(current bounding box.center)},scale=1/3]
	\draw[line width = 2pt] (3,-1) -- (3,3);
	\draw[line width = 2pt] (6,-1) -- (6,3);
	\draw[dashed] (3,-1) -- (6,-1);
	\draw[dashed] (3,3) -- (6,3);
	\foreach \s in {1,...,3}
	{	
		\filldraw[black] (3,\s -1) circle (5pt);
		\filldraw[black] (6,\s -1) circle (5pt);
	};
	\draw (3,0) -- (6,0);
	\draw (3,1) -- (6,1);
	\draw (3,2) -- (6,2);
\end{tikzpicture} \ ,
\qquad \textrm{and}\qquad u_2 \phi\Big(\ \begin{tikzpicture}[baseline={(current bounding box.center)},scale=1/3]
	\draw[line width = 2pt] (3,-1) -- (3,3);
	\draw[line width = 2pt] (6,-1) -- (6,3);
	\draw[dashed] (3,-1) -- (6,-1);
	\draw[dashed] (3,3) -- (6,3);
	\foreach \s in {1,...,3}
	{	
		\filldraw[black] (3,\s -1) circle (5pt);
	};
	\filldraw[black] (6,1) circle (5pt);
	\draw (3,1) -- (6,1);
	\draw (3,2) .. controls (3.75,2) and (4.5,2.5) .. (4.5,3);
	\draw (3,0) .. controls (3.75,0) and (4.5,-1/2) .. (4.5,-1);
\end{tikzpicture} 
\ +{\mathsf U}_0 \Big)=0,
$$
and on the other hand,
$$
u_2\Big(\ \begin{tikzpicture}[baseline={(current bounding box.center)},scale=1/3]
	\draw[line width = 2pt] (3,-1) -- (3,3);
	\draw[line width = 2pt] (6,-1) -- (6,3);
	\draw[dashed] (3,-1) -- (6,-1);
	\draw[dashed] (3,3) -- (6,3);
	\foreach \s in {1,...,3}
	{	
		\filldraw[black] (3,\s -1) circle (5pt);
	};
	\filldraw[black] (6,1) circle (5pt);
	\draw (3,1) -- (6,1);
	\draw (3,2) .. controls (3.75,2) and (4.5,2.5) .. (4.5,3);
	\draw (3,0) .. controls (3.75,0) and (4.5,-1/2) .. (4.5,-1);
\end{tikzpicture} 
\ +{\mathsf U}_0\Big)=z\ 
\begin{tikzpicture}[baseline={(current bounding box.center)},scale=1/3]
	\draw[line width = 2pt] (0,-1) -- (0,3);
	\draw[line width = 2pt] (3,-1) -- (3,3);
	\draw[dashed] (0,-1) -- (3,-1);
	\draw[dashed] (0,3) -- (3,3);
	\foreach \s in {1,...,3}
	{	
		\filldraw[black] (0,\s -1) circle (5pt);
	};
	\filldraw[black] (3,1) circle (5pt);
	\draw (0,0) .. controls (1,0) and (1,1) .. (0,1);%
	\draw (0,2) .. controls (1,2) and (2,1) .. (3,1);
\end{tikzpicture} 
\ +{\mathsf U}_0=0+{\mathsf U}_0,\qquad\textrm{and}\qquad
\phi\Big(\ u_2\Big(\ \begin{tikzpicture}[baseline={(current bounding box.center)},scale=1/3]
	\draw[line width = 2pt] (3,-1) -- (3,3);
	\draw[line width = 2pt] (6,-1) -- (6,3);
	\draw[dashed] (3,-1) -- (6,-1);
	\draw[dashed] (3,3) -- (6,3);
	\foreach \s in {1,...,3}
	{	
		\filldraw[black] (3,\s -1) circle (5pt);
	};
	\filldraw[black] (6,1) circle (5pt);
	\draw (3,1) -- (6,1);
	\draw (3,2) .. controls (3.75,2) and (4.5,2.5) .. (4.5,3);
	\draw (3,0) .. controls (3.75,0) and (4.5,-1/2) .. (4.5,-1);
\end{tikzpicture} 
\ +{\mathsf U}_0\Big)\ \Big)=\phi(0+{\mathsf U}_0)=0.
$$
The other verifications are left to the reader. Since ${\mathsf U}_r/{\mathsf U}_{r-1}$ is sent isomorphically onto $\TheS{n,k+2r}$ and this module $\TheS{n,k+2r}$ is projective when $q$ is generic, it follows that for all $0<i\leq (n-k)/2$
\begin{equation*}{\mathsf U}_r\simeq {\mathsf U}_{r-1}\oplus \TheS{n,k+2r}.\end{equation*}
The result for $\Resar{\TheW{n,k;z}}$ then follows from the limiting cases: ${\mathsf U}_{0} \simeq \TheS{n,k}$ and ${\mathsf U}_{(n-k)/2}\simeq\ \Resar{\TheW{n,k;z}}$.

When $q$ is generic, the pair $(k,z)\in\paires$ may have at most one successor in the weakest partial order $\preceq$. If $(k,z)$ has no successor, then $\TheW{n,k;z}=\TheL{n,k;z}$ and $\Resar{\TheL{n,k;z}}$ is obtained from the previous result. If $(k,z)$ has one, this successor $(j,y)$ solves equations (A) of \eqref{eq:AandB} and there is a positive integer $t\leq (n-k)/2$ such that $(k,z)=(k,z_{k+2t})\preceq(j,y)=(k+2t,z_k)$. The irreducible $\TheL{n,k;z_{k+2t}}$ is then obtained by the short exact sequence
\begin{equation*}
	0 \rightarrow \TheW{n,k + 2t;z_{k}} \longrightarrow \TheW{n,k;z_{k+2t}} \longrightarrow \TheL{n,k;z_{k+2t}} \rightarrow 0.
\end{equation*}
The second statement of the proposition then follows by applying the restriction functor to this sequence and using the first statement.
\end{proof}
Frobenius theorem then readily gives the following corollary.
\begin{Coro}
If $q$ is generic and $0<t\leq \frac{n-k}{2}$, then
\begin{align*}
	\Hom(\Indar{\TheS{n,r}}, \TheL{n,k;z_{k+2t}}) &\simeq \begin{cases}
		\ringK, & \text{ if }  k \leq r \leq k+2t, \\
		0, & \text{otherwise,}
	\end{cases} \\
	\Hom(\Indar{\TheS{n,r}}, \TheW{n,k;z}) &\simeq \begin{cases}
		\ringK, & \text{ if }  k \leq r\\
		0, & \text{otherwise,}
	\end{cases}\\
	\Hom(\Indar{\TheS{n,r}},\TheW{n,k})
		 &\simeq \begin{cases}
				\ringK[\tau,\tau^{-1}], & \text{ if }  k \leq r,\\
				0 & \text{otherwise}.
			\end{cases}
\end{align*}
\end{Coro}
\begin{proof}The first two $\Hom$-groups are obtained from the previous proposition through Frobenius theorem. For the third one, $k\leq r $ is assumed, because otherwise the result is given by lemma \ref{lem:hom.indS.W}. Let $\mu_w:\TheW{n,k} \to \TheW{n,k;w}$ denote the canonical projection and $f: \Indar{\TheS{n,r}} \to \TheW{n,k} $ be any non-zero morphism. Such a non-zero morphism exists. For example, since $\mu_w$ is surjective and $\Indar{\TheS{n,r}}$ is projective, the map $\mu_w$ can be lifted to a morphism $\Indar{\TheS{n,r}} \to \TheW{n,k}$. An argument similar to the one used in Proposition \ref{prop:hom.w.to.w} proves that $\mu_\omega \circ f=0$ may occur only for a finite number of $\omega\in\ringK$. It is thus possible to write $f$ as 
\begin{equation}\label{eq:factoOfF}
		 f = \phi \circ \bar{f}, \qquad\textrm{such that }\quad \mu_{w}\circ \bar{f} \neq 0 \quad \textrm{for any } w \neq 0.
	\end{equation}
Here $\bar{f}:\Indar{\TheS{n,r}} \to \TheW{n,k}$, $\phi=\prod_{\omega\in\Omega_0}(\tau-\omega)\in \ringK[\tau,\tau^{-1}]$ with $\tau-\omega$ being understood as the endomorphism of $\TheW{n,k}$ obtained by right action of $(\tau-\omega\id)$. Some $\omega$ may need to appear more than once in the finite set $\Omega_0$ to allow for $\bar f$ to satisfy the condition $\mu_{w}\circ \bar{f} \neq 0$ for all $w\neq 0$.

Let $\bar{g}:\Indar{\TheS{n,r}} \to \TheW{n,k} $ be another morphism such that $\mu_{w}\circ \bar{g} \neq 0 $ for all non-zero $w$. Because $\Hom(\Indar{\TheS{n,r}},\TheW{n,k})\simeq \ringK$, there exists $\lambda \in \ringK $ such that $\mu_{z} \circ \bar{g} = \lambda \mu_{z} \circ \bar{f}$. Thus $\bar g=\lambda \bar f+(\tau-z)h$ for some $h:\Indar{\TheS{n,r}} \to \TheW{n,k}$. But, then again, this $h$ can be factorized $h=\phi'\circ \bar h$ as in \eqref{eq:factoOfF} with $\mu_{z} \circ \bar{h} = \lambda' \mu_{z} \circ \bar{f}$ for some $\lambda'\in\ringK$. Thus $\bar g=(\lambda+\lambda'\phi'(\tau-z)) \bar f +(\tau-z)^2 i$ where $i$ is yet another morphism in $\Hom(\Indar{\TheS{n,r}}, \TheW{n,k})$. The argument can be repeated, but it will terminate because the number of zeroes of $\mu_z\circ \bar g$ (counted with multiplicities) is finite. Thus $\bar{g} = \psi \circ \bar{f}$, for some $\psi \in \End(\TheW{n,k})$.
\end{proof}
%
%
Note that some extension groups follows easily from this corollary. For instance if $r \geq k$
\begin{equation*}
		\mathsf{Ext}^{1}_{\atl{n}}(\TheW{n,k}, \TheW{n,r}) \simeq \mathsf{Ext}^{1}_{\atl{n}}(\TheW{n,k}, \TheW{n,r;z}) \simeq  0, 
	\end{equation*}
	\begin{equation*}
		\mathsf{Ext}^{1}_{\atl{n}}(\TheW{n,k;z}, \TheW{n,r;w}) 
		\simeq  \delta_{z,w}\delta_{r,k} \mathbb{C}, \qquad \mathsf{Ext}^{1}_{\atl{n}}(\TheW{n,k;z}, \TheW{n,r}) \simeq \delta_{k,r} \mathbb{C},
	\end{equation*}
	\begin{equation*}
		\mathsf{Ext}^{1}_{\atl{n}}(\TheW{n,k}, \TheL{n,r;z_{r+2t}}) \simeq 0. 
	\end{equation*}
These can all be readily found by applying the appropriate functors on the sequence \eqref{eq:exactIndS}, or 
\begin{equation}
	0 \longrightarrow \TheW{n,k} \xrightarrow{\tau-z\ } \TheW{n,k} \longrightarrow \TheW{n,k;z} \longrightarrow 0.
\end{equation}
\end{subsection}
\begin{subsection}{Examples of fusion rules}
This last subsection brings to bear the explicit expressions obtained for the restriction and induction functors $\Resphi, \Indphi, \Resar, \Indar$ acting on the basic modules over $\tl{n}$ and $\atl{n}$. These expressions reduce the computation of fusion rules $( \blank \xfone \blank)$ and $( \blank \xftwo \blank)$ (and even $( \blank \xfthree \blank)$) defined in section \ref{sec:fusionFunctors} to known expressions for fusion rules of $\tl{}$-modules. The latter were obtained by Gainutdinov and Vasseur \cite{GainutdinovVasseur} and Bellet\^ete \cite{BelFus15}. 

When $q$ is generic, general expressions can be given and we display several fusion rules between the standard modules $\TheW{n,k}$, the cell $\TheW{n,k;z}$ and the irreducible $\TheL{n,k;z}$ over $\atl n$. When $q$ is a root of unity, the expressions are messy. (See, for example, Proposition 4.1.1 of \cite{GainutdinovVasseur}.) However the fusion of projective and standard modules of $\tl n$ can be computed easily by simple algorithmic rules (Propositions 4.4 and 5.12 of \cite{BelFus15}.) We shall use these to give one example of affine fusion at $q$ a root of unity. 

\begin{Prop}\label{prop:voila}
Let $q$ be generic, $r \geq s \geq 0$ and $n,m > 2$. Then
	\begin{align*}
 \TheW{n,r} \xfone \TheW{m,s}  
	& \simeq  \TheL{n,r;z_{r+2}} \xfone \TheL{m,s;z_{s+2}}  \simeq \mybigoplus{t=r-s}{r+s} \TheL{n+m,t;z_{t+2}}, \\
 \TheW{n,r} \xftwo \TheW{m,s} 
	& \simeq \TheL{n,r;z_{r+2}} \xftwo \TheL{m,s;z_{s+2}} \simeq \mybigoplus{t = r-s}{r+s} \Indar{\TheS{n+m,t}},\\
 \TheW{n,r;z_{r+2}} \xfone \TheW{m,s;z_{s+2}} 
	&\simeq \mybigoplus{t = r-s}{r+s} \TheL{n+m,t;z_{t+2}},\\
 \TheL{n,r;z_{r+2a}} \xfthree \TheL{m,s;z_{s+2b}}(m)  
	& \simeq 
\bigoplus_{x=0}^{a-1}\bigoplus_{y=0}^{b-1}
	\big( \TheL{n;r+2x;z_{r+2(x+1)}} \xfone \TheL{m,s+2y;z_{s+2(y+1)}} \big)
\end{align*}
where the sign $\myoplus{}{}$ denotes a direct sum with an increment of $2$.
\end{Prop}
\begin{proof}The proofs of these fusion rules simply use the definition of the fusion products and the results of the previous sections. We give two examples.
\begin{align*}
 \TheW{n,r} \xfone \TheW{m,s}  & \equiv\ \Resphi{\big( (\Indphi\TheW{n,r})\xf (\Indphi\TheW{m,s})\big)}, \qquad \textrm{by definition \ref{def:fusionFunctors}},\\
&\simeq\ \Resphi{\big( \TheS{n,r}\xf\TheS{m,s}\big)},\qquad \textrm{by proposition \ref{prop:indH}},\\
&\simeq\ \Resphi {\big( \mybigoplus{t=r-s}{r+s} \TheS{n+m,t} \big)},\qquad\textrm{see section 4.5 of \cite{BelFus15}},\\
& \simeq\ \mybigoplus{t=r-s}{r+s} \TheL{n+m,t;z_{t+2}},\qquad\textrm{by proposition \ref{prop:indphi.resphi.irred}}
\end{align*}
where we have used the fact that, at a generic $q$, the projective, the standard and the irreducible modules all coincide: $\TheP{n,k}\simeq\TheS{n,k}\simeq\TheI{n,k}$. This fact is also used in the next proof:
\begin{align*}
 \TheL{n,r;z_{r+2}} \xftwo \TheL{m,s;z_{s+2}}  & \equiv\ \Indar{\big( (\Indphi{\TheL{n,r;z_{r+2}}})\xf (\Indphi{\TheL{m,s;z_{s+2}}})\big)}, \qquad \textrm{by definition \ref{def:fusionFunctors}},\\
&\simeq\ \Indar{\big(} \TheS{n,r}\xf\TheS{m,s}{\big)},\qquad \textrm{by proposition \ref{prop:indphi.resphi.irred}},\\ 
&\simeq\ \Indar {\big( \mybigoplus{t=r-s}{r+s} \TheS{n+m,t} \big)},\qquad\textrm{see section 4.5 of \cite{BelFus15}},\\
& \simeq\ \mybigoplus{t=r-s}{r+s} \Indar{\TheS{n+m,t}},
\end{align*}
where $\Indar{\TheS{n+m,t}}$ is the $\atl n$-module whose structure is described in propositions \ref{prop:inductionS} and \ref{pro:indS}.
\end{proof}
We end by giving an example of a computation at a root of unity. Let $q$ be such that $q^{2\ell}=1$ for $\ell=5$. Then the fusion $(\TheW{5,3}\xfone \TheW{7,5})$ is computed as follows:
\begin{align*}
\TheW{5,3}\xfone\TheW{7,5}  & \equiv \ \Resphi{\big( (\Indphi{\TheW{5,3}})\xf (\Indphi{\TheW{7,5}})\big)},  \qquad & \textrm{by definition \ref{def:fusionFunctors}},\\
&\simeq\ \Resphi{\big( \TheS{5,3}\xf\TheS{7,5}\big)},\qquad & \textrm{by proposition \ref{prop:indH}},\\
&\simeq\ \Resphi{\big( \TheP{12,4}\oplus \TheP{12,6}\oplus  \TheS{12,8} \big)},& \qquad\textrm{by proposition 5.12 of \cite{BelFus15}},\\
& \simeq\ {\big( \Resphi{\TheP{12,4}}\oplus \Resphi{\TheP{12,6}}\oplus \Resphi{\TheS{12,8}} \big)},\qquad &\ 
\end{align*}
where the structures of the $\atl n$-modules $\Resphi{\TheP{n,k}}$ and $\Resphi{\TheS{n,k}}$ are given in Figure 1, end of section \ref{sec:twoFunctors}. Because of proposition \ref{prop:indphi.thew}, the fusion $\TheW{5,3;y}\xfone\TheW{7,5;z}$ is zero unless $y=(-q)^{5/2}$ and $z=(-q)^{7/2}$, in which case it is isomorphic to $\TheW{5,3}\xfone\TheW{7,5}$.

\end{subsection}
\end{section}

%
\begin{section}{Concluding remarks}\label{sec:conclusion}
%

This article introduced two fusion products over the affine Temperley-Lieb algebra  through the two morphisms $\tl n\overset{\iota}\rightarrow \atl n$ and $\atl n\overset{\phi}\rightarrow \tl n$. Since their definitions rely on the fusion $\xf$ on the regular Temperley-Lieb algebra and due to proposition \ref{lem:functorId}, these products are commutative and associative. The products computed all lead to a finite direct sum of indecomposable modules. Contrarily to the usual tensor products, say of modules over group algebras, the two new products are {\em not} binary operations, that is, they do not take two modules of a given algebra to produce a module of the same algebra. Instead, the product of $\mathsf M\in\Mod \atl m$ with $\mathsf N\in\Mod \atl n$ yields a module in $\Mod \atl{m+n}$. Nevertheless they enjoy a property of {\em stability} in the sense that the product of, say, $\TheW{m,k_1}$ with $\TheW{n,k_2}$ depends only on $k_1$ and $k_2$, not on $m$ or $n$. This property is crucial if these products are to be useful in the context of lattice models and their continuum limits. Beside the proposal of these fusions, the computations of the products in proposition \ref{prop:voila} form a key result. As a byproduct of these computations, results on the representation theory of the affine Temperley-Lieb algebra were obtained: the $\Hom$-groups of the standard $\TheW{n,k}$ and cellular modules $\TheW{n,k;z}$, the dimensions of the irreducible ones $\TheL{n,k;z}$ and, more importantly, the Peirce decomposition of $\atl n(q)$ for $q$ generic in terms of the induced $\tl{}$-standard modules.

There are still many questions to investigate. Our computations left aside the cases for $q$ a root of unity for the computation of fusion products, $\Hom$-groups of $\Indar{\TheS{n,k}}$ and the Peirce decomposition of the affine algebra. These could be interesting. Both the fusion product proposed by Gainutdinov and Saleur \cite{GaiSalAffine} and ours share several important categorical properties, a natural requirement. But they are clearly different: even the simplest products given as examples in \cite{GaiSalAffine} differ from ours. Moreover their recent work \cite{GaiSalAffine2} also claims that the cell modules are closed under fusion (the fusion of two cell modules yields a direct sum of cell modules) which is not the case under our product. Despite these differences, one may ask whether they are related in some ways. Finally it is the comparison with the continuum limit of lattice models, e.g.~the XXZ models, that will decide on the physical relevance of the proposed fusion products.

\end{section}

\appendix
%
\section{Yet another fusion product}\label{app:otherfusion}
%

In this appendix, we introduce another presentation of the affine Temperley-Lieb algebra. This presentation reveals an injective morphism $\atl m\times\atl n\to\atl{m+n}$ and thus another way to do fusion. This appendix is partially based on ideas borrowed from \cite{GLatRootsOfUnity} and \cite{halverson} and ties in with Gainutdinov and Saleur's work on affine fusion \cite{GaiSalAffine}.

Note that the morphism $\phi $, introduced in proposition \ref{prop:defphi}, suggests a different way of choosing the generators of the affine Temperley-Lieb algebra. If we replace the two generators $e_{n}, \tau$ with $g \equiv (-q)^{-3/2} \tau t_{n-1}^{-1}t_{n-2}^{-1}\hdots t_{1}^{-1}$, the defining relations of the algebra are then equivalent to
\begin{equation}
	e_i e_i = (q+q^{-1})e_{i}, \qquad e_{i}e_{i \pm 1}e_i = e_i, \quad \forall i = 1, 2,\hdots n-1,
\end{equation}
\begin{equation}
	e_i e_j = e_j e_i, \quad \text{ if } |i-j| \geq 2,
\end{equation}
\begin{equation}
	e_{1}g e_{1} = (\underbrace{q g + q^{-1} g^{-1} }_{\Lambda})e_{1}, \qquad \Lambda e_{i} = e_{i} \Lambda \quad  i = 1,\hdots , n-1. 
\end{equation}
For example, when $n=5$ the diagram representing $g$ is
\begin{equation}
 (-q)^{3/2} g \equiv \quad
 \begin{tikzpicture}[scale = 1/3, baseline = {(current bounding box.center)}]
	\draw[line width = 2pt] (0,-1) -- (0,5);
	\draw[line width = 2pt] (3,-1) -- (3,5);
	\draw[dashed] (0,-1) -- (3,-1);
	\draw[dashed] (0,5) -- (3,5);
	\draw[line width = 1pt, black] (0,4) .. controls (.5,4) and (1.5,4.5) .. (1.5,5);
	\draw[line width = 1pt, black] (1.5,-1) .. controls (1.5,3.5) and (2.5,4) .. (3,4);
	\foreach \s in {0,...,3}{
		\draw[line width = 3pt,white] (0,\s) -- (3,\s);
		\draw[line width = 1pt,black] (0,\s) -- (3,\s);
	}
	\foreach \s in {1,...,5}
	{	
		\filldraw[black] (0,\s -1) circle (5pt);
		\filldraw[black] (3,\s -1) circle (5pt);
	};
	\end{tikzpicture}\quad , \qquad
	\Lambda \equiv - \quad
	\begin{tikzpicture}[scale = 1/3, baseline = {(current bounding box.center)}]
	\draw[line width = 2pt] (0,-1) -- (0,5);
	\draw[line width = 2pt] (3,-1) -- (3,5);
	\draw[dashed] (0,-1) -- (3,-1);
	\draw[dashed] (0,5) -- (3,5);
	\draw[line width = 1pt, black] (1.5,5) -- (1.5,-1);
	\foreach \s in {0,...,4}{
		\draw[line width = 3pt,white] (0,\s) -- (3,\s);
		\draw[line width = 1pt,black] (0,\s) -- (3,\s);
	}
	\foreach \s in {1,...,5}
	{	
		\filldraw[black] (0,\s -1) circle (5pt);
		\filldraw[black] (3,\s -1) circle (5pt);
	};
	\end{tikzpicture}\quad ,
\end{equation}
where we used the braid notation:
\begin{equation}
	\begin{tikzpicture}[scale = 1/3, baseline = {(current bounding box.center)}]
		\draw[line width = 1pt,black] (0,2) -- (2,0);
		\draw[line width = 3pt, white] (0,0) -- (2,2);
		\draw[line width = 1pt, black] (0,0) -- (2,2);
	\end{tikzpicture}\quad  = (-q)^{1/2} \quad
	\begin{tikzpicture}[scale = 1/3, baseline = {(current bounding box.center)}]
		\draw[line width = 1pt,black] (0,2) .. controls (0,1) and (2,1) ..  (2,2);
		\draw[line width = 1pt,black] (0,0) .. controls (0,1) and (2,1) ..  (2,0);
	\end{tikzpicture} \quad + (-q)^{-1/2} \quad 
	\begin{tikzpicture}[scale = 1/3, baseline = {(current bounding box.center)},rotate = 90]
		\draw[line width = 1pt,black] (0,2) .. controls (0,1) and (2,1) ..  (2,2);
		\draw[line width = 1pt,black] (0,0) .. controls (0,1) and (2,1) ..  (2,0);
	\end{tikzpicture}\quad .
\end{equation}
Note that it follows directly from these relations that
\begin{equation} (q g)^{n} = \Lambda (q g)^{n-1} - (q g)^{n-2} = \hdots = q g U_{n-1}(\Lambda /2) -U_{n-2}(\Lambda/2)\end{equation}
where $U_{n-1}(x)$ is a Chebyshev polynomial of the second kind. 

	Since $\Lambda $ is central, one can consider the algebra obtained by taking the quotient of $\atl{n}$ by the ideal generated by $(\Lambda - \lambda \id )^{m}$, for some $\lambda \in \mathbb{C}$; we call the canonical projection to this quotient algebra $\phi_{\lambda}^{m}$. In particular, if $m = 1$ we obtain the (generalized\footnote{If $Q \neq 1$, one can rescale the generator $b$ and obtain the proper blob algebra, but $Q=1$ one gets a slightly different algebra \cite{MDRR}.}) blob algebra \cite{blob}, which can be seen by defining $b = q g - Q^{-1}$, $Q+ Q^{-1} \equiv \lambda$, which gives:
\begin{equation*}
	\begin{split}
		(q g) + (q g)^{-1} &= Q + Q^{-1} \\
		e_{1}ge_{1} &= \Lambda e_{1}
	\end{split}
	\quad \Leftrightarrow \quad
	\begin{split}
		b^{2} & =  (Q - Q^{-1}) b\\
		e_{1}be_{1} & = ( q Q - Q^{-1}q^{-1}) e_{1}
	\end{split}.
\end{equation*}
It follows that every module on which $\Lambda $ acts like a multiple of the identity is isomorphic to the induction of some module of the blob algebra \cite{GLatRootsOfUnity}. Note that this also proves that $\atl{n}$ is finite dimensional when seen as a $\mathbb{C}[\Lambda]$-algebra, and thus explains why it has such a nice Peirce decomposition despite being an infinite dimensional $\mathbb{C}$-algebra.

These properties also give restrictions on the idempotents of $\atl{n}$. Let $f\in \atl{n} $ be an idempotent such that $\phi_{\lambda}^{1}(f) = 0 $, but $f \neq 0$. It follows that $f = a_{1}(\Lambda - \lambda \id)$ for some $a_{1}\in \atl{n}$, and since $f$ is an idempotent,
\begin{equation*}
	f^{m} = (a_{1})^{m}(\Lambda - \lambda \id)^{m} = f, \qquad \forall m=1,2,\hdots,
\end{equation*}
which gives $\phi_{\lambda}^{m}(f) = 0 $ for all positive integer $m$, but since $f$ must be a finite product of generators, this is impossible. We thus conclude that for all idempotents $f$, $\phi^{m}_{\lambda}(f) \neq 0$ for all integer $m$ and all complex numbers $\lambda $, and thus that $\phi^{1}_{\lambda}(f)$ is a non-zero idempotent of the blob algebra. Note however that if $f$ is primitive, $\phi^{1}_{\lambda}(f)$ might not be. 

 It is also straightforward to show that the affine Temperley-Lieb algebra is a quotient of the type $B_{n}$ braid algebra, that is
	\begin{equation*}
 		t_{1}gt_{1}g = g t_{1}g t_{1}.
 	\end{equation*}
However, it is not a quotient of the type $B_{n}$ Hecke algebra, which would require that $\Lambda$ be a multiple of the identity \cite{halverson}.

  Finally, if $g_{k} \equiv t_{k}\hdots t_{1}g t_{1}\hdots t_{k}$, one can show that
\begin{equation}
 g_{k} g_{t} = g_{t} g_{k},
\end{equation}
\begin{equation}
	e_{k+1}g_{k}e_{k+1} = (\underbrace{q g_{k} + q^{-1}g_{k}^{-1}}_{\Lambda_{k}}) e_{k+1}, \qquad \Lambda_{k}e_{j}= e_{j}\Lambda_{k}, \quad j= 1,\hdots, k-1, k+1, \hdots n-1.
\end{equation}
 This shows in particular that $\atl{n+m}$ has (at least) two subalgebras isomorphic to $\atl{n} \otimes_{\mathbb{C}} \atl{m}$. They can be used to build a natural fusion product, equivalent to the one constructed in \cite{GaiSalAffine}; they could also be used to build a tower of algebras $\atl{1}\subset \atl{2}\subset \hdots \subset \atl{n}$ to construct the projective modules by induction, using the same kind of arguments as those used for the regular Temperley-Lieb algebra \cite{RSA}. 
 
%
\section{The weakest partial order}\label{app:wpo}
%
This appendix describes the weakest partial order in full generality. (Proposition \ref{prop:weakest} is restricted to the cases that are relevant to compute fusion.) It also give the dimension of the irreducible $\atl n$-module $\TheL{n,k;z}$ in terms of the dimension of the $\tl n$-modules $\TheS{n,k}$ and $\TheI{n,k}$.

This appendix uses the definitions of a {\em generic} $q$ and {\em critical} $k$ introduced in section \ref{sec:modulesOfTln} and of $\Lambda_n^a$ and $\Lambda^a$ from section \ref{sec:affine.modules}. If $q$ is a root of unity, $\ell$ will be the smallest positive integer such that $q^{2\ell}=1$. The non-negative integers with the parity of $n$ will be partitioned into {\em chambers} that consist, for a given $b\geq 0$, of integers $m$ with $b\ell\leq m\leq (b+1)\ell-1$. There is at most one critical $k$ in a chamber and, if there is, it is the largest element in this chamber.
%
%
\begin{Lem}Let $q,z\in\mathbb C^\times$ with $q$ generic. Then, in the weakest partial order $\preceq$ on $\paires$, there is at most one element $(s,y)$ that succeeds strictly the element $(t,z)$.
\end{Lem}
The proof rests upon an easy analysis of the solutions of conditions (A) and (B) in \eqref{eq:AandB}.
%
%
\begin{Lem}Let $q,z\in\mathbb C^\times$ with $q$ a root of unity. If, in the weakest partial order, there is an element strictly succeeding $(t,z)$, then this $(t,z)$ has infinitely many distinct successors in $\paires$.
\end{Lem}
This follows from the observation that, if $j$ and $y$ solve either (A) or (B), so do $j+2b\ell$ and $y(-q)^{b\ell}$ for all $b\geq 0$.

If $q$ is a root of unity and $z\in\mathbb C^\times$, the equations $z^2=(-q)^{\pm j}$ can have a solution only if $z$ is a $4\ell$-th root of unity. Such a $z$ can always be written as $(-1)^m(-q)^{r/2}$ with $r\in\mathbb Z$ and $m\in\{0,1\}$.

%
%
\begin{Prop}Let $q\in\mathbb C^\times$ be a root of unity other than $\pm 1$ and $\ell\geq 2$ be the smallest integer such that $q^{2\ell}=1$. Let $(s,z)\in\paires$ be a pair with infinitely many distinct successors. Let $r$ be the smallest integer larger than $s$ that solves $z^2=(-q)^r$. Finally let $m\in\{0,1\}$ be such that $z=(-1)^m(-q)^{r/2}$. Then the successors of $(s,z)=(k_0,u_0)$ in the weakest partial order are organized as in the figure of proposition \ref{prop:weakest} where again $(k,z)\leftarrow (j,y)$ stands for $(k,z)\preceq (j,y)$, a vertical arrow means that the pair joined solves condition (A) and a diagonal one condition (B) of \eqref{eq:AandB}. The pairs are
\begin{equation*}\left.
\begin{aligned}(k_a,u_a)&=(s+2a\ell,(-1)^m(-q)^{r/2+a\ell}),\\
(j_a,y_a)&=(r+2a\ell, (-1)^m(-q)^{s/2+a\ell}),\\
(i_a,x_a)&=(-r+\delta_i+2a\ell, (-1)^m(-q)^{-s/2+\delta_i/2+a\ell}),\\
(h_a,v_a)&=(-s+\delta_h+2a\ell, (-1)^m(-q)^{-r/2+\delta_h/2+a\ell}),
\end{aligned}\right\}\qquad \textrm{for }a\geq 0,
\end{equation*}
where the $\delta_i$ and $\delta_h$ are the smallest integer multiples of $2\ell$ such that $i_0$ is larger than $k_0$, and $h_0$ larger than both $i_0$ and $j_0$. The list of successors is exhaustive, but there might be arrows missing in the diagram and some of the pairs may coincide.
\end{Prop}
\begin{proof}The successors are obtained by solving (A) or (B) of \eqref{eq:AandB}. We first discuss the case when the four families $(k_a,u_a),\allowbreak(j_a,y_a),(i_a,x_a)$ and $(h_a,v_a)$ do not intersect. The pair $(j_a,y_a)$ is obtained by solving (A) starting with $(k_a,u_a)$. Since $u_a^2=(-q)^{2(r/2+a\ell)}=(-q)^r$, the equation for $j_a$ reads simply $(-q)^r=(-q)^{j_a}$. By definition of $r$, $j_a=r+2a\ell$ is the smallest solution that is larger than $k_a$ (recall that $j_a$ and $k_a$ must of the same parity). Under the assumptions that the four families do no intersect, then $j_a$ with $a\geq 1$ is also larger than $h_{a-1}$. Indeed the integers $h_a-j_a$ and $h_a-i_a$ must be larger than $0$ and $\leq 2\ell$, since they solve equations between $2\ell$-th roots of unity and, at each step of the solution process, the smallest possible is chosen. Moreover, since there are no coincidences, the differences $h_a-j_a$ and $h_a-i_a$ cannot be $2\ell$. Thus $j_{a-1}<h_{a-1}<j_{a-1}+2\ell=j_{a}$. The other arrows are obtained similarly by solving the appropriate (A) or (B). The list is exhaustive since the successors of the four families belong to these four families. For example solving (A) starting from $(h_a,v_a)$ will lead to a member of the family of the $(i_a,x_a)$'s and so on.

Coincidences among the four families of successors may occur in several ways. The analysis of these cases is not much more difficult and we discuss a single example. Suppose $s=r$. Then $(j_a,y_a)=(r+2(a+1)\ell,(-1)^m(-q)^{s/2+a\ell})$ and $(j_a,y_a)=(k_{a+1},u_{a+1})$ if $(-q)^\ell=1$, which may occur for example when $(-q)=exp(2\pi i e/\ell)$ with $\ell$ odd and $\gcd(e,\ell)=1$ Suppose that $(-q)^\ell$ is indeed $1$ and thus that these coincidences occur. Then it is easy to check that new arrows appear in the figure, namely
$$(j_a,y_a)=(k_{a+1},u_{a+1})\longleftarrow (i_a,x_a),$$
because then the equation (B) starting from $(i_a,x_a)$ is $u_{a+1}^{-2}=(-q)^j$ with $x_a^{-2}=(-q)^s$ and, thus, $k_{a+1}=s+2b\ell$ for some integer $b$. The expression of $i_a$ gives $k_{a+1}=i_a+(s+r-\delta_i-2(b-a)\ell)$ and the second equation in (B) is $u_{a+1}=x_a(-q)^{(s+r)/2}= (-1)^m(-q)^{r/2}=y_a$ because, again $(-q)^\ell=1$. In this case the figure of the weakest partial order can be streamlined into the simpler 
$$(k_0,u_0)\longleftarrow (i_0,x_0) \longleftarrow (k_1,u_1) \longleftarrow (i_1,x_1) \longleftarrow \dots$$
In all other cases when two families coincide, the ladder-like diagram of successors degenerate to a single line like the above one.
\end{proof}
%
%
From now on we shall refer to a {\em weakest partial order with coincidences} when some of the four families intersect. When there are coincidences, the cardinality of $\{(k_0,u_0),\allowbreak(j_0,y_0),(i_0,x_0),(h_0,v_0)\}$ is $2$, otherwise it is $4$.
\begin{Coro}\label{cor:dimL}Let the hypotheses of the previous proposition hold. With the notation above, the dimension $\dim \TheL{n,k;z}$ of the irreducible $\atl n$-modules is as follows. The sign $\mysum{}{}$ denotes a sum with an increment of $2$.

\noindent{\bfseries $\bullet$ $s$ and $r$ non-critical} --- If there are no coincidences, the dimension is
$$\begin{cases}
\mysum{j=s}{s^+} \dim\TheS{n,j}+\mysum{j=s^++2}{\min(i_0,j_0)-2}\dim\TheI{n,j}, & \textrm{if }s^++2<\min(i_0,j_0),\\
\mysum{j=s}{\min(i_0,j_0)-2}\dim\TheI{n,j}, & \textrm{otherwise},$$
\end{cases}
$$
and if there are
$$\sum_{b\geq 0}\mysum{j=s}{\min(i_0,j_0)-2}\dim \TheS{n,j+2b\ell}.$$

\noindent{\bfseries $\bullet$ $s$ critical, $r$ non-critical} --- If there are no coincidences, the dimension is
$$\dim\TheS{n,s}+\mysum{j=s+2}{\min(i_0,j_0)-2}\dim\TheI{n,j},\quad\textrm{and if there are}\quad
\sum_{b\geq 0}\TheS{n,s+2b\ell}.$$

\noindent{\bfseries $\bullet$ $s$ non-critical, $r$ critical} --- If there are no coincidences, the dimension is
$$\begin{cases}
\mysum{j=s}{s^+} \dim\TheS{n,j}+\mysum{j=s^++2}{r-2}\dim\TheI{n,j}, & \textrm{if }r-s>\ell,\\
\mysum{j=s}{r-2}\dim\TheI{n,j}, & \textrm{if }r-s<\ell,
\end{cases}\quad\textrm{
and if there are,}\quad
\sum_{b\geq 0}\mysum{j=s}{s^+}\dim \TheS{n,j+2b\ell}.$$

\noindent{\bfseries $\bullet$ $s$ and $r$ critical} --- The dimension is
$$\begin{cases}
\sum_{b\geq 0}\TheS{n,s+2b\ell},& \textrm{if }r=s,\\
\dim\TheS{n,s}+\mysum{j=s+2}{r-2}\dim\TheI{n,j}, & \textrm{if }r=s+\ell.
\end{cases}$$
The sums $\sum_{b\geq 0}$ all truncate after a finite number of terms ($\TheS{n,k}=0$ for $k>n$).
\end{Coro}
\begin{proof}
The proof starts with the usual alternating argument. Theorem \ref{thm:GL34} gives the composition factors of $\TheW{n,k=k_0;z=u_0}$. To obtain $\dim \TheL{n,k;z}$, the dimensions of $\TheW{n,j_0;y_0}$ and $\TheW{n,i_0;x_0}$ are first subtracted from $\dim \TheW{n,k_0;z}$ as they share with $\TheW{n,k=k_0;z}$ the composition factors $\TheL{n,j_0;y_0}$ and $\TheL{n,i_0;x_0}$. (The dimensions of the $\TheW{}$'s are given by \eqref{eq:dimWnkz}.) This subtraction removes twice the dimensions of the composition factors $\TheL{n,k_1;u_1}$ and $\TheL{n,h_0;v_0}$ which are put back by simply adding $\dim\TheW{n,k_1;u_1}$ and $\dim\TheW{n,h_0;v_0}$. But the last operation has added twice the dimensions of $\TheL{n,j_1;y_1}$ and $\TheL{n,i_1;x_1}$. Alternating additions and subtractions thus leads to the formula:
\begin{equation}\label{eq:telescope}\dim \TheL{n,k;z}=\sum_{a\geq 0}\left(\dim\TheW{n,k_a;u_a}
-\dim\TheW{n,j_a;y_a}
-\dim\TheW{n,i_a;x_a}
+\dim\TheW{n,h_a;v_a}\right).\end{equation}
The dimensions of the $\TheS{n,k'}$ and $\TheW{n,k;z}$ are simply related by equations \eqref{eq:dimSnk} and \eqref{eq:dimWnkz}:
$$\dim\TheW{n,k;z}=\mysum{j\geq 0}{}\ \dim\TheS{n,k+j}.$$
In turn the dimension of the irreducible $\TheI{n,k}$ is given by (\cite{Martin, RSA}):
$$\dim\TheI{n,k}=\sum_{b\geq 0}(-1)^b\dim\TheS{n,k^b}$$
where $k^0=k, k^1=k^+, k^2=k^{++},\dots,$ are the integers larger or equal to $k$ in the class $[k]$. (See section \ref{sec:modulesOfTln}.) Again these two sums contain only a finite number of terms. 
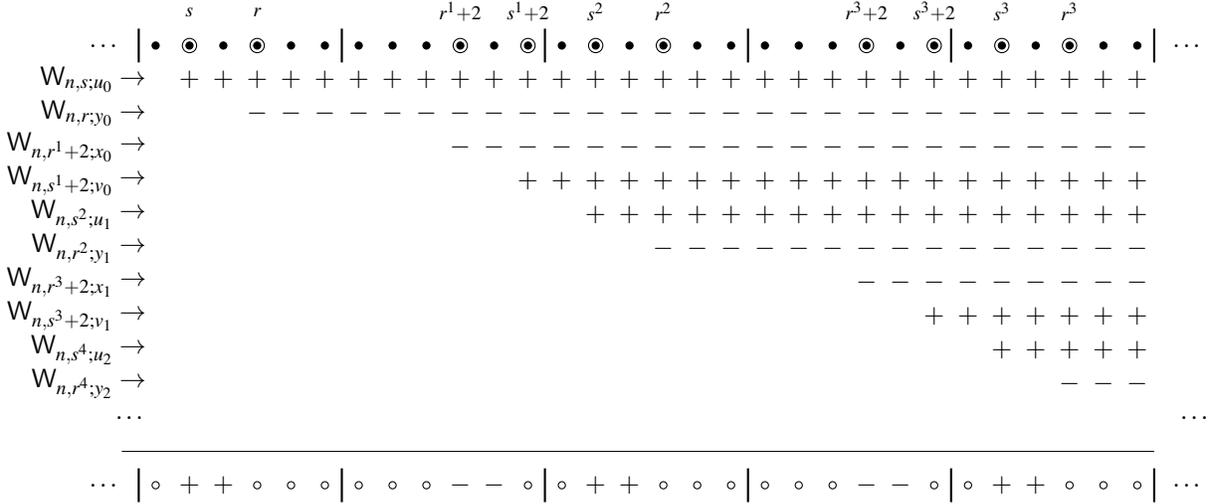
\begin{equation*}
\begin{tikzpicture}[baseline={(current bounding box.center)},scale=0.45]
	\foreach \s in {1,...,30}{\filldraw[black] (\s,0) circle (3pt);};
	\draw (2,0) circle (6pt);
	\draw (4,0) circle (6pt);
	\draw (10,0) circle (6pt);
	\draw (12,0) circle (6pt);
	\draw (14,0) circle (6pt);
	\draw (16,0) circle (6pt);
	\draw (22,0) circle (6pt);
	\draw (24,0) circle (6pt);
	\draw (26,0) circle (6pt);
	\draw (28,0) circle (6pt);
	\foreach \s in {0,...,5}{\draw[thick] (0.5+6*\s,0.5) -- (0.5+6*\s,-0.5) ;};	
	\node at (-0.5,0) {$\dots$};
	\node at (31.5,0) {$\dots$};
	\node at (2,1) {$\scriptstyle{s}$};
	\node at (4,1) {$\scriptstyle{r}$};
	\node at (10,1) {$\scriptstyle{r^1+2}$};
	\node at (12,1) {$\scriptstyle{s^1+2}$};
	\node at (14,1) {$\scriptstyle{s^2}$};
	\node at (16,1) {$\scriptstyle{r^2}$};
	\node at (22,1) {$\scriptstyle{r^3+2}$};
	\node at (24,1) {$\scriptstyle{s^3+2}$};
	\node at (26,1) {$\scriptstyle{s^3}$};
	\node at (28,1) {$\scriptstyle{r^3}$};
    \node at (1,-1) [anchor=east] {$\TheW{n,s;u_0}\to$};
    \foreach \s in {2,...,30}{\node at (\s,-1) {$+$};};
    \node at (1,-2) [anchor=east] {$\TheW{n,r;y_0}\to$};
    \foreach \s in {4,...,30}{\node at (\s,-2) {$-$};};
    \node at (1,-3) [anchor=east] {$\TheW{n,r^1+2;x_0}\to$};
    \foreach \s in {10,...,30}{\node at (\s,-3) {$-$};};
    \node at (1,-4) [anchor=east] {$\TheW{n,s^1+2;v_0}\to$};
    \foreach \s in {12,...,30}{\node at (\s,-4) {$+$};};
    \node at (1,-5) [anchor=east] {$\TheW{n,s^2;u_1}\to$};
    \foreach \s in {14,...,30}{\node at (\s,-5) {$+$};};
    \node at (1,-6) [anchor=east] {$\TheW{n,r^2;y_1}\to$};
    \foreach \s in {16,...,30}{\node at (\s,-6) {$-$};};
    \node at (1,-7) [anchor=east] {$\TheW{n,r^3+2;x_1}\to$};
    \foreach \s in {22,...,30}{\node at (\s,-7) {$-$};};
    \node at (1,-8) [anchor=east] {$\TheW{n,s^3+2;v_1}\to$};
    \foreach \s in {24,...,30}{\node at (\s,-8) {$+$};};
    \node at (1,-9) [anchor=east] {$\TheW{n,s^4;u_2}\to$};
    \foreach \s in {26,...,30}{\node at (\s,-9) {$+$};};
    \node at (1,-10) [anchor=east] {$\TheW{n,r^4;y_2}\to$};
    \foreach \s in {28,...,30}{\node at (\s,-10) {$-$};};
	\node at (1,-11) [anchor=east] {$\dots$};
	\node at (31,-11) [anchor=west] {$\dots$};
	\draw (0,-12) -- (30.5,-12);
	\node at (2,-13) {$+$};\node at (3,-13) {$+$};\node at (10,-13) {$-$};\node at (11,-13) {$-$};
	\node at (14,-13) {$+$};\node at (15,-13) {$+$};\node at (22,-13) {$-$};\node at (23,-13) {$-$};
	\node at (26,-13) {$+$};\node at (27,-13) {$+$};
	\foreach \s in {0,...,5}{\draw[thick] (0.5+6*\s,-12.5) -- (0.5+6*\s,-13.5) ;};
	\draw (1,-13) circle (3pt);	
	\foreach \s in {4,...,9}{\draw (\s,-13) circle (3pt);};
	\foreach \s in {12,...,13}{\draw (\s,-13) circle (3pt);};
	\foreach \s in {16,...,21}{\draw (\s,-13) circle (3pt);};
	\foreach \s in {24,...,25}{\draw (\s,-13) circle (3pt);};
	\foreach \s in {28,...,30}{\draw (\s,-13) circle (3pt);};
	\node at (-0.5,-13) {$\dots$};
	\node at (31.5,-13) {$\dots$};
\end{tikzpicture}
\end{equation*}
The rest of the argument keeps track of how many terms of the form $\pm\dim\TheS{n,j}$ appear in the telescopic sum \eqref{eq:telescope}. We give two examples of such exercises. The first one is a case without coincidences. Suppose $s$ and $r$ are non non-critical. When $s$ is non-critical, then $h_a=-s+\delta_h+2a\ell$ is equal to $s^++2+2b\ell$ for some integer $b$. Suppose then that the first elements $k_0,i_0,j_0,h_0$ of the four families are all distincts (that is, there is no coincidences) and ordered as
$$k_0=s<i_0=r<j_0=r^++2<h_0=s^++2<k_1=s+2\ell<\dots$$
The sum \eqref{eq:telescope} is shown diagrammatically in the above figure. All dots on the top line depict integers of the same parity and small vertical lines the critical integers. (In the figure $\ell=12$, the dots are at even integers and the critical lines are odd ones.) The elements of the families $k_a, j_a, i_a, h_a$, $a\geq 0$, are also shown by  circles around the dots. Under the top line are the contributions of the $\TheW{n,l;w}$'s ordered with increasing index $l$. The dimension of $\TheW{n,s;u_0}$ is the sum of the dimensions $\dim\TheS{n,j}$ for $j\geq s$ with the parity of $s$. These contributions form the top row of $+$'s starting under the column $s$. Next $\dim\TheW{n,r;y_0}$ contributes also $\dim\TheS{n,j}$ for all $j$, now starting with $j=r$. But now the contribution of $\dim\TheW{n,r;y_0}$ is weighted by $(-1)$ and $-$'s are drawn under all positions starting at $j=r$. Once the contributions of all $\TheW{}$'s are drawn, the sum $\sum_{a\geq 0}$ is now done column by column, and the result is shown in the bottom line. The results are shown by $+$'s and $-$'s, or $\begin{tikzpicture}[scale=0.45] \draw (0,0) circle (3pt);\end{tikzpicture}$ if the various contributions cancel. The first non-zero contribution is $\dim\TheS{n,s}$ and all terms in the class $[s]$ after $s$ appear with alternating sign; the sum of all these is $\dim\TheI{n,s}$. The same result holds for the contributions larger and equal to $s+2$ in the class $[s+2]$. The telescopic sum is thus
$$\mysum{j=s}{r-2}\dim\TheI{n,j}$$
as claimed ($r=\min(i_0,j_0)$ in the chosen order).

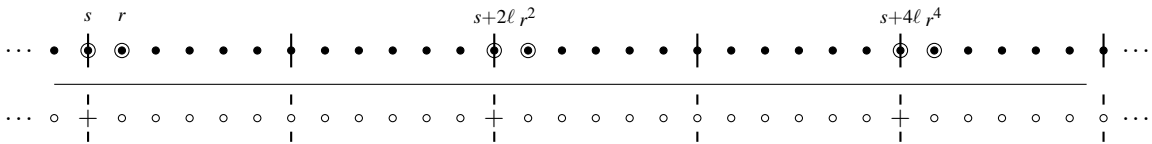
\begin{equation*}
\begin{tikzpicture}[baseline={(current bounding box.center)},scale=0.45]
	\foreach \s in {0,...,31}{\filldraw[black] (\s,0) circle (3pt);};
	\draw (1,0) circle (6pt);
	\draw (2,0) circle (6pt);
	\draw (13,0) circle (6pt);
	\draw (14,0) circle (6pt);
	\draw (25,0) circle (6pt);
	\draw (26,0) circle (6pt);
	\foreach \s in {0,...,5}{\draw[thick] (6*\s+1,0.5) -- (6*\s+1,-0.5) ;};	
	\node at (-1,0) {$\dots$};
	\node at (32,0) {$\dots$};
	\node at (1,1) {$\scriptstyle{s}$};
	\node at (2,1) {$\scriptstyle{r}$};
	\node at (13,1) {$\scriptstyle{s+2\ell}$};
	\node at (14,1) {$\scriptstyle{r^2}$};
	\node at (25,1) {$\scriptstyle{s+4\ell}$};
	\node at (26,1) {$\scriptstyle{r^4}$};
	\draw (0,-1) -- (30.5,-1);
	\foreach \s in {0,...,5}{\draw[thick] (6*\s+1,-1.3) -- (6*\s+1,-1.6) ;};	
	\foreach \s in {0,...,5}{\draw[thick] (6*\s+1,-2.4) -- (6*\s+1,-2.7) ;};	
	\node at (-1,-2) {$\dots$};
	\node at (32,-2) {$\dots$};
	\draw (0,-2) circle (3pt);
	\foreach \s in {0,...,2}{\node at (1+12*\s,-2) {$+$};};
	\foreach \s in {0,...,2}{\draw  (7+12*\s,-2) circle (3pt);};
	\foreach \s in {2,...,6}{\draw (\s,-2) circle (3pt);};
	\foreach \s in {8,...,12}{\draw (\s,-2) circle (3pt);};
	\foreach \s in {14,...,18}{\draw (\s,-2) circle (3pt);};
	\foreach \s in {20,...,24}{\draw (\s,-2) circle (3pt);};
	\foreach \s in {26,...,30}{\draw (\s,-2) circle (3pt);};
\end{tikzpicture}
\end{equation*}
The second example is one with coincidences. Take $s$ critical and $r$ non-critical. If $r$ is equal to $s+2$, then its reflexion through the critical line on its right is $s+2\ell-2$ and thus $i_0=r^++2=s+2\ell$, and the two families $(k_a,u_a)$ and $(i_a,x_a)$ coincide, except for $(k_0,u_0)$ that does not belong to $\{(i_a,x_a), a\geq 0\}$. The families $(j_a,y_a)$ and $(h_a,v_a)$ also coincide and the weakest partial order is
$$(k_0,u_0)\longleftarrow (j_0,y_0) \longleftarrow(k_1,u_1) \longleftarrow (j_1,y_1) \longleftarrow \dots$$
Consequently the telescopic sum takes the simpler form
\begin{equation*}\label{eq:dimLnkz}\dim \TheL{n,k;z}=\sum_{a\geq 0}\left(\dim\TheW{n,k_a;u_a}
-\dim\TheW{n,j_a;y_a}\right).\end{equation*}
The figure above depicts the new situation. The lines of $+$'s and $-$'s are omitted as the cancellations are simpler. The rows of $+$'s start at $s, s+2\ell, s+4\ell, \dots$, and they are almost completely canceled by rows of $-$'s starting at $r,r^2, r^4, \dots$: only the $+$'s immediately under the elements $s, s+2\ell, s+4\ell, \dots$, survive, as indicated. The telescopic sum thus becomes in this case
$$\dim\TheL{n,s;z}=\sum_{b\geq 0}\dim\TheS{n,s+2b\ell}$$
as claimed. All other cases are treated similarly. 
\end{proof}


\section*{Acknowledgments}

YSA is supported by a Discovery Grant of the National Sciences and Engineering Research Council of Canada. JB is supported by a scholarship from Fonds de recherche Nature et Technologies (Qu\'ebec), and by the European Research
Council (advanced grant NuQFT). JB would also like to thank Azat Gainutdinov, Hubert Saleur, and Jesper Lykke Jacobsen for interesting discussions.

\raggedright

\end{document}